\documentclass[12pt]{article}
\usepackage[margin=1in]{geometry}
\usepackage[utf8]{inputenc}
\usepackage{setspace} 
\usepackage{booktabs} 
\usepackage{amsmath,amssymb,amsthm, bm, bbm} 
\usepackage{graphicx, xcolor}
\usepackage{natbib} 
\usepackage{enumerate} 
\usepackage{comment} 
\usepackage[colorinlistoftodos,textsize=scriptsize]{todonotes}
\usepackage{hyperref} 
\usepackage{subcaption}
\hypersetup{colorlinks,citecolor=blue,urlcolor=blue,linkcolor=blue}

\newtheorem{assumption}{Assumption}
\newtheorem{theorem}{Theorem}
\newtheorem{proposition}[theorem]{Proposition}
\newtheorem{corollary}[theorem]{Corollary}
\newtheorem{lemma}[theorem]{Lemma}
\newtheorem{example}{Example} 

\theoremstyle{definition}

\graphicspath{{figures/}}
\usepackage[]{natbib}
\newcommand{\indep}{\perp \!\!\! \perp}

\providecommand{\blind}{0} %

\begin{document}

\if0\blind
\title{Causal Inference under Interference through Designed Markets \footnote{
I thank
Claudia Allende-Santa Cruz, 
Isaiah Andrews,
Susan Athey, 
Lanier Benkard, 
Han Hong, 
Guido Imbens, 
Michael Leung, 
Peter Reiss, 
David Ritzwoller, 
Brad Ross, 
Paulo Somaini, 
Sofia Valdivia, 
Stefan Wager, and participants of numerous conferences and seminars for helpful comments and discussions.}
\vspace{2mm}
}

\author{
Evan Munro  \thanks{Chicago Booth School of Business, University of Chicago}}
\else
\title{Causal Inference under Interference through Designed Markets
\vspace{2mm}
}
\author{}
\fi
\date{ December 1, 2026 \\
 }

  \maketitle

  \bigskip
\begin{abstract} 
Interactions among buyers make it challenging to evaluate the impact of buyer-side treatments in a single market, even with data from a randomized trial. In many auction and matching markets, however, a centralized mechanism allocates goods and imposes useful structure on these interactions. For a class of ``cutoff'' mechanisms, we introduce a continuum-market model of counterfactuals, where buyers interact only through the cutoffs. Under unconfoundedness, we propose an estimator for the Global Treatment Effect that runs a weighted and perturbed version of the mechanism on data from a single market. The estimator is asymptotically normal and semi-parametrically efficient. We characterize optimal spillover-aware treatment rules and propose a policy learning approach with $1/\sqrt{n}$ regret. For counterfactuals defined in the observed finite market, inference is conservative and the regret rate is unchanged. Empirically, spillovers diminish the effect of information on inequality in the Chilean school system.
\end{abstract}

\newpage
\setstretch{1.3}

\section{Introduction}

An individual-level intervention in an economic system rarely affects agents in isolation. Interactions among market participants lead to spillover effects, where the treatment of one individual affects the outcomes of others.
Spillover effects make it challenging to estimate the Global Treatment Effect (GTE), the difference in expected outcomes when everyone is treated compared to when nobody is treated. Existing approaches in the causal inference literature assume either partial interference, where there are no spillovers across clusters of agents \citep{baird2018optimal, hudgens2008toward}, or local network interference, where spillovers occur through a sparse network \citep{aronow2017estimating, leung2020treatment}. Except for parametric model-based approaches, there has been limited progress in estimating the GTE under complete interference, where the treatment of any individual may impact anyone else's outcome. We show that in markets where spillover effects are mediated by a specific class of centralized allocation mechanisms, even though there is complete interference, semi-parametric estimation of global effects is possible.

Settings where a centralized mechanism allocates scarce items are increasingly common in practice. In the U.S., versions of the deferred acceptance algorithm allocate students to schools \citep{abdulkadirouglu2003school} and medical school graduates to residency programs \citep{roth2003origins}. Auctions allocate advertisements to search queries \citep{varian2014vcg} and Treasury bonds to investors \citep{mcmillan2003market}. Often, policymakers are interested in estimating how an intervention that affects the reported preferences of market participants will impact the resulting allocations from the mechanism. For example, \citet{allende2019approximating} provide information about school quality to families in a randomized experiment in Chile, where a centralized mechanism determines allocations to schools. One of their target estimands is the GTE, where the treatment is the information intervention and the outcome is the allocation of low-income families to high-quality schools. 

Estimators of the Average Treatment Effect (ATE) fail to recover the GTE, even when the treatment is randomly assigned. By increasing the number of applicants to schools with limited capacity, the treatment affects admissions probabilities, which introduces spillovers and violates the Stable Unit Treatment Value Assumption (SUTVA) on outcomes \citep{heckman1998general}. To estimate the GTE, \citet{allende2019approximating} use data from the experiment to fit a parametric structural model of reported preferences over schools, and simulate the relevant counterfactuals using this model and the centralized mechanism.

 In this paper, we take a new approach, which is a step towards extending more ``credible'' approaches based on quasi-experimental variation \citep{angrist2009mostly} to a richer set of counterfactuals beyond the ATE. We start by extending the continuum-market model of \citet{azevedo2016supply}, in which a \textit{cutoff} mechanism allocates $J$ capacity-constrained items to a superpopulation of buyers. In a cutoff mechanism, such as the uniform-price auction and deferred acceptance, buyers' allocations are a known function of their report (e.g. a bid or rank-order list) and a finite vector of cutoffs that ensure that demand does not exceed capacity \citep{azevedo2016supply, agarwal2018demand}. Each buyer has two potential reports to the mechanism, one under treatment and the other under control, and outcomes are a known function of reports, covariates, and cutoffs. In the continuum model, buyers interact only through $J$ market-clearing cutoffs, so the continuum-market GTE has a moment-based representation in terms of the distributions of the treated and control reports. Under unconfoundedness and overlap, these distributions, and so also the GTE, are identified from buyer-level data from a single market. A key assumption for point identification is that although SUTVA does not hold for outcomes, it does hold for potential reports. In an appendix extension, we provide sharp bounds for the continuum-market GTE under a weaker assumption that allows additional spillovers through reports.

When treatment effects are heterogeneous, a policymaker can improve on treating everyone or nobody by assigning treatment based on pre-treatment covariates. There is a large literature on policy learning under SUTVA for outcomes \citep{kitagawa2018should, athey2021policy}, but the problem is not always tractable when there are spillover effects, as discussed in the network setting by \citet{vivianorestud}. In the continuum market, we show that the policy value for a fixed targeting rule is point identified under the same assumptions as for the GTE. Furthermore, the optimal unrestricted targeting rule treats a group when the effect of the treatment on its externality-adjusted outcome is positive, where the adjustment captures the effect of a group's change in demand on others' outcomes due to capacity constraints. In contrast to settings without interactions, randomized rules can strictly improve upon deterministic rules. 

A plug-in approach estimates the conditional distribution of reports given covariates under treatment and under control, integrates over the covariate distribution, and runs the mechanism on the resulting counterfactual report distributions. Non-parametric conditional density estimation, however, converges slowly when reports or covariates are even of moderate dimension, and parametric approaches are prone to misspecification. To avoid estimating these densities, we target the moment-based representation of the continuum-market GTE directly.

Our estimator solves an empirical version of the moment-based representation of the continuum-market GTE. It relies on doubly-robust scores, and requires a careful adaptation of results on localization approaches for GMM models with missing data, specifically the work of  \citet{kallus2019localized}.  In the first step, we use a propensity-score approach to estimate counterfactual market-clearing cutoffs. In the second step, a debiased estimate of counterfactuals is computed by running a re-weighted version of the mechanism with perturbed capacities, where the weights and perturbations are estimated using a simple set of machine learning regressions. Data splitting is used to control bias, allowing for weak conditions on the convergence rates of the machine learning estimators. Using techniques from the theory of empirical processes, we show that the estimator is asymptotically normal and semi-parametrically efficient. The estimator is robust to a variety of specifications for buyer behavior -- as long as certain statistics of the counterfactual report distributions are sufficiently smooth and the machine learning estimators meet regularity conditions, we can perform inference on the continuum-market GTE.

We provide asymptotic regret results for policy learning with market interactions, employing a two-step, doubly-robust approach for empirical welfare maximization. This yields an asymptotic regret bound that is of the same order as the lower bounds in the literature on policy learning without spillover effects \citep{athey2021policy}.  Without interactions, the value of a policy is linear in the policy, and regret bounds follow from VC-class assumptions on a restricted policy class. In our setting, the policy enters non-linearly into the value function, through the counterfactual cutoffs, which solve a policy-specific market-clearing condition and are estimated with a policy-specific set of nuisances. To obtain regret bounds in this more challenging setting, we bound the complexity of doubly-robust scores jointly in the policy and cutoffs and prove uniform convergence of estimated market-clearing cutoffs to their continuum-market counterparts. We use an empirical process theory argument that relies on both a VC-class assumption on the policy class and entropy conditions on the class of allocation functions. The argument applies whenever the value of a decision depends on a decision-specific method-of-moments estimator, so it may be useful in the policy learning literature more broadly. 

 The observed market allocates a finite number of buyers to items, and the most policy-relevant estimand is usually defined in the finite market, rather than in the continuum market. The finite-market GTE no longer has a simple moment-based representation and is difficult to work with directly, since it is an average of dependent terms. Instead, we show that the finite-market GTE is close to the continuum-market GTE asymptotically. This approximation result implies that inference that is asymptotically valid for the continuum-market estimand is conservative for the finite-market estimand. We also show that the $1/\sqrt n$ regret rate from the continuum market holds even when we define regret relative to the optimal policy in the finite market. 

In simulations of a uniform-price auction, we illustrate the robustness properties of our preferred estimator, in contrast to approaches based on parametric structural modeling. In an additional simulation of a market for schools, we demonstrate that even in markets with a modest number of students per school, the continuum market is a good approximation to the finite market. Finally, we apply our methods in a real-world setting using data from Chile, where children are allocated to public schools via a version of deferred acceptance. We compile a dataset from the Ministry of Education that replicates many of the features of the data in  \citet{allende2019approximating}, except that the treatment is self-reported awareness of government-provided information on school quality, rather than an explicitly randomized intervention. We estimate the GTE, where the outcome measures the allocation of low-income families to good-quality schools. We find that if spillover effects are ignored, then the ATE estimate is significant, raising access of low-income families to good schools by 1.5 percentage points. However, the GTE estimate is much smaller at 0.7 percentage points. The estimated spillover-aware optimal targeting rule, chosen from a class of logistic rules, raises access of low-income families to good schools by 1 percentage point.

\subsection{Related Work} 

There is a body of existing work that estimates different types of causal effects in designed markets. \citet{abdulkadirouglu2017research} estimate causal effects of allocations on future outcomes, such as test scores or income, using randomness in the matching mechanism for identification. \citet{abdulkadirouglu2022breaking}, \citet{chen2021nonparametric}, and \citet{bertanha2023causal} extend this work to settings where individual scores are non-random but the cutoff structure of the mechanism allows a regression discontinuity design-based analysis.   \citet{bertanha2023causal} also consider partial identification of  preferences from strategic reports when mechanisms are not strategy-proof. In contrast to this body of work, our paper focuses on an earlier step in the causal chain of events, which is the effect of a pre-allocation intervention on resulting allocations.

\citet{athey2007nonparametric} survey non-parametric identification and estimation methods for primitives in a wide range of auction models. We focus on the estimation of specific counterfactuals rather than model primitives. The disadvantage of our approach is that estimating primitives is useful for estimating a wider range of counterfactuals, and can sometimes be of independent interest. The advantage is efficiency and robustness, in that for some markets, we can obtain precise estimates without imposing strong (e.g. parametric or distributional) assumptions.

\citet{bright2022reducing} also consider a causal inference setting where interactions occur through a known algorithm. They characterize the bias of a randomized controlled trial (RCT) in a parametric model of a matching market, where a linear program computes the matching. They propose a simulation-based estimator of the GTE that requires estimating their model using maximum likelihood estimation. Our paper studies markets with a different class of matching mechanisms that have a cutoff structure; in this class of mechanisms, we estimate causal effects without imposing a parametric model of behavior. 

 \citet{munro2023market} also constrain spillovers to occur through a set of market statistics. Without the presence of a centralized mechanism, the model primitives are the distribution of demand and supply functions, rather than counterfactual distributions of bids. This means the authors are limited to counterfactuals that are local to the current equilibrium, and require more complex experimental designs with price randomization for identification. In contrast, the current paper is more directly related to the literature on structural modeling in that a) we use data with standard treatment variation and b) we identify global effects that extrapolate from the observed market. Furthermore, the methods for handling nuisance function estimation and uniform convergence are novel compared to the approach in \citet{munro2023market}. 

To analyze the properties of the estimators in the paper, we use a framework where the allocation mechanism operates on a continuum of agents rather than a finite number of agents. Asymptotic approximations have been used to characterize the bias and variance of A/B testing-based approaches in online marketplaces in  \citet{johari2022experimental}, \citet{bright2022reducing} and \citet{liao2023statistical}. 

\section{Identification of Continuum-Market Counterfactuals}

\label{sec:model} 

\subsection{Identifying the Global Treatment Effect} 

This section defines estimands in a continuum economy and shows they are identified from the distribution of bidders observed in a single market. In a continuum economy $[F, s^*]$, as in \citet{azevedo2016supply}, a finite set of $J$ suppliers with capacities $s^* \in [0, 1]^J$ are matched using a centralized mechanism to a superpopulation of buyers, represented by a distribution $F$. 

 We augment the model in \citet{azevedo2016supply} to introduce counterfactuals. The econometrician observes buyer-level data from a market of $n$ buyers; for a single buyer, this is $(B(W), W, X)$, which corresponds to an i.i.d. draw from the  distribution $F_{\text{obs}}$. $B(1)$ is the buyer's report to the mechanism if they are treated and $B(0)$ is the report if they are not. $W \sim \mbox{Bernoulli}(e(X))$ is a binary treatment and $X \in \mathcal X$ are pre-treatment covariates.  $[F_{\text{obs}}, s^*]$ defines the continuum economy under the observed treatment assignment rule $e$. 
 
 A counterfactual policy $\pi$ is a measurable function $\pi: \mathcal X \rightarrow [0, 1]$ that assigns $W_{\pi} \sim \mbox{Bernoulli}(\pi(X))$. The observed policy is $e$. Although we will later discuss policies that target treatments based on pre-treatment covariates, the continuum-market GTE is the difference in expected outcomes between two uniform policies: $\pi_1(x) = 1$, which treats everyone, and $\pi_0(x) = 0$, which treats nobody. The corresponding all-treated and all-control continuum economies are $[F_w, s^*]$ for $w \in \{0, 1\}$, where $F_w$ is the distribution of $(B(w),w,X)$.
 
We next define allocations in the continuum economy, which are determined by a cutoff mechanism. A cutoff mechanism is a centralized mechanism that takes a specific form: there is a bounded, measurable and known allocation function $d(b, x, p)$ so that market clearing is determined by a finite vector of cutoffs. In an economy $[F, s^*]$, the market-clearing cutoffs $p^* \in \mathcal S \subseteq \mathbb R^J$ solve $ \mathbb E_{F}[ d(B,  X, p^*)] = s^*$, where $d(B, X, p^*) \in \{0, 1\}^J$.

So, the counterfactual market-clearing cutoffs solve $z_w(p^*_w) =  0$, where excess demand is defined as $z_w(p) = \mathbb E_{F_w}[ d(B(w),  X, p)] - s^*,  $ for cutoffs $p \in \mathcal S$ and treatment status $w$. 
The equality condition implies that all $J$ suppliers are capacity-constrained, but it can easily be extended to an inequality for settings with excess supply. If no suppliers are capacity-constrained, however, then there are no interactions in our model. 

Counterfactuals are defined in terms of a uniformly bounded, measurable and known function $y(b, x, p)$  such that $y(b, x, p)$ is a proxy for welfare in a market with cutoffs $p$. In a variety of empirical settings, including the inequality measure studied in the market for schools in Section \ref{sec:empirical} and bidder surplus in auctions, outcomes are match value $y(b, x, p) = m(b, x, p)^{\top} d(b, x, p)$, where $m(b, x, p) \in \mathbb R^J$ gives the value of each item.\footnote{For an identification strategy that accommodates a more general class of outcomes and combines our approach with the lottery-based identification strategy of  \citet{abdulkadirouglu2017research}, see Appendix \ref{ap:angrist}.}
 
 The continuum-market GTE is now formally defined as  \[ \tau^*_{\text{GTE}} = y_1(p^*_1) - y_0(p^*_0), \] where $y_w(p) = \mathbb E_{F_w}[y(B(w), X, p)]$.  In the standard causal inference setting under SUTVA, the population ATE is a linear functional of the marginal distributions of potential outcomes. Here, there are interactions through capacity constraints. Because of this, $\tau^*_{\text{GTE}}$ is a non-linear functional of the model primitives $(F_1, F_0)$: $F_w$ enters $y_w(p^*_w)$ directly, through the expectation, and indirectly due to capacity constraints, represented by the change in market-clearing cutoffs from $p^*_1$ to $p^*_0$. Our first assumption ensures $\tau^*_{\text{GTE}}$ is well-defined by imposing uniqueness on the market-clearing cutoffs.

\begin{assumption} \label{as:unique_new} For $w \in \{0, 1\}$, $p^*_w$ exists and is unique. 
\end{assumption}

\begin{example} \label{ex1} \textbf{Uniform-Price Auction.}
In a uniform-price auction for a single item with unit demand and private values, each buyer submits a bid $B(w) \in \mathbb R$, and bidding one's value is a dominant strategy. In the continuum market, a fraction $s^* \in (0, 1)$ of buyers receive the item. The auction has a cutoff structure with $d(b, x, p) = \mathbbm 1(b \geq p)$, so $z_w(p) = \mathbb E_{F_w}[\mathbbm 1(B(w) \geq p)] - s^*$ and the market-clearing cutoff $p^*_w$ is the $(1-s^*)$ quantile of the bid distribution.  When the outcome is bidder surplus, $y(b,x,p)=(b-p)\mathbbm 1(b\ge p)$, a treatment that raises bids raises the surplus of treated bidders at fixed $p$, but also raises the cutoff from $p^*_0$ to $p^*_1$, which lowers surplus of winning bidders. 
\end{example}

\begin{example} \label{ex2} \textbf{Deferred Acceptance.}
In many cities, students are matched to schools using deferred acceptance with priority scores, which is strategy-proof and has a cutoff structure, as shown in \citet{azevedo2016supply}. In the continuum market, school $j$ has capacity $s^*_j$. A student's report is $B(w) = (R(w), S)$, where $R(w)$ is a  preference ranking over schools that may be affected by the treatment and $S \in \mathbb R^J$ is a vector of school-specific priority scores. The cutoff $p_j$ is the lowest score admitted to school $j$, and a student is assigned to their highest-ranked school whose cutoff their score clears.\footnote{An explicit function for $d(b, x, p)$ in deferred acceptance is given in \citet{agarwal2018demand}. }
\end{example}

To establish identification of $\tau^*_{\text{GTE}}$, we next provide assumptions under which there is a mapping from $F_{\text{obs}}$, the population distribution corresponding to the observed buyers, to $\tau^*_{\text{GTE}}$.

\begin{assumption} \label{as:id_new}  Identification of the Continuum-Market GTE. 
\begin{enumerate} 
\item[(i)] For $ w \in \{0, 1\}$, almost surely $B(w)$ is invariant to the treatment rule $\pi$. 
\item[(ii)] Unconfoundedness and overlap hold, so $ \{ B(1), B(0) \} \indep W \mid X ,$ and for all $x \in \mathcal X$, $0 < e(x) < 1$. 
\end{enumerate} 

\end{assumption}

The structure of the cutoff mechanism and the knowledge of the outcome function $y(b, x, p)$ imply that recovering the distribution of bids in the all-treated and all-control economies is sufficient to identify $\tau^*_{\text{GTE}}$. Part (i), which is the continuum-market version of SUTVA, is the key policy invariance condition. Since $B(w)$ is invariant across economies with different treatment rules, a buyer's report under treatment status $w$ is the same in the observed market as it
would be if every buyer were treated.  Part (ii) is
the standard selection-on-observables condition that identifies the model primitives $F_w$.\footnote{It is possible to use an IV-type assumption as an identifying condition instead at the cost of only identifying a restricted version of $ \tau^*_{\text{GTE}}$, see Appendix \ref{app:iv}.}

We briefly discuss Assumption \ref{as:id_new}(i) in more detail. A sufficient condition is that truthful reporting is a dominant strategy under the market-clearing mechanism, as in Examples 1--2. Truthful reporting, however, is not a necessary condition; misreporting meets Assumption \ref{as:id_new}(i) as long as the reporting strategy does not change with the treatment rule. The assumption fails when equilibrium reports depend on the distribution of
others' reports, as with immediate acceptance in matching markets or first-price auctions. It also fails when the treatment distribution reaches others' reports indirectly through non-equilibrium pathways, for
example through information sharing.
\begin{proposition} \label{prop:id}
Under Assumptions \ref{as:unique_new}--\ref{as:id_new}, $\tau^*_{\text{GTE}}$ is point identified. 
\end{proposition}

Assumption \ref{as:id_new}(i) is not testable, so in Appendix \ref{ap:partial_id}, we bound the bias from its violation, in the spirit of \citet{rosenbaum2002observational}.  Specifically, we allow reports $B(w, a)$ to depend on treatment $w$ and a known exposure $a(x, p, \pi)$ that depends on the cutoffs and policy. Then, when reports are real-valued, under a $\lambda$-Lipschitz bound in bids as a function of $a$, we obtain sharp bounds on $\tau^*_{\text{GTE}}$ that depend on $\lambda$. We provide two practical examples of exposures. In the first, bids depend on market-clearing cutoffs directly, which accommodates settings where there is a common value component to a bidder's value for the good. The second allows for local spillover effects; bids depend on the fraction of neighbors treated. 

\subsection{Policy Values in the Continuum Market}

There are a variety of counterfactuals that are of interest beyond the Global Treatment Effect, including the value of targeted treatments, where a subset of units are treated. We next discuss policies that assign treatment according to observed covariates.

In the counterfactual economy with policy $\pi$, a buyer with covariates $X$ is treated and reports $B(1)$ with probability $\pi(X)$, and otherwise reports $B(0)$. Let $F_{\pi}$ denote the distribution of $(B(W_{\pi}), W_{\pi}, X)$ when $W_{\pi} \sim \mbox{Bernoulli}(\pi(X))$, so $dF_{\pi}(b, w, x) = \pi(x)^w (1 - \pi(x))^{1-w} dF_w(b, w, x)$. The continuum economy is $[F_{\pi}, s^*]$.

Excess demand is $z_{\pi}(p) = \mathbb E_{F_{\pi}}[d(B(W_{\pi}), X, p)] - s^*$ and market-clearing cutoffs solve $z_{\pi}(p^*_{\pi}) = 0$. Defining $y_{\pi}(p) = \mathbb E_{F_{\pi}}[y(B(W_{\pi}), X, p)]$ analogously, the policy value in the continuum market is $V^*(\pi) = y_{\pi}(p^*_{\pi})$. $F_{\pi}$ is a known functional of $\pi$ and $(F_1, F_0)$, and, by Assumption \ref{as:id_new}(i), $F_w$ does not vary with $\pi$. By the same argument as in Proposition \ref{prop:id}, $V^*(\pi)$ is point identified under Assumption \ref{as:id_new} whenever $p^*_{\pi}$ exists and is unique.

Let $\Pi$ denote a candidate class of policies, which we assume contains $\pi_1$ and $\pi_0$.\footnote{This assumption is without loss of generality, since we can always augment a given class to add the two uniform treatment rules.} $\tau^*_{\text{GTE}} = V^*(\pi_1) - V^*(\pi_0)$ compares the two uniform policies, and informs the decision of whether or not to introduce a treatment to all buyers. When treatment response is heterogeneous, a targeted rule may outperform both. For example, when $\mathcal X \subseteq \mathbb R^k$, the class of logistic rules is $\Pi = \{ x \mapsto \Lambda(\beta_0 + x^{\top} \beta_1) : (\beta_0, \beta_1) \in \mathbb R^{k+1} \}$, where $\Lambda(z) = 1/(1+e^{-z})$. Since $\Pi$ typically contains infinitely many rules, the goal is not only to evaluate a given rule but to find an optimal policy in the class, 

\begin{equation} 
\pi^* \in \arg \max_{\pi \in \Pi} V^*(\pi). \end{equation} 

To ensure that $\pi^*$ can be characterized in terms of a simple necessary condition when $\Pi$ is unrestricted, we next impose a set of regularity conditions on the continuum market. 

\begin{assumption} \label{as:continuum_regularity} \textbf{Continuum-Market Regularity} 
There are strictly positive and finite constants $c_1$, $c_2$, $c_3$, $c_4$, and  $c'$ such that 
\begin{enumerate} 
 \item[(i)] $\mathcal S$ is a compact set. For all $\pi \in \Pi$, $\mathcal S$ contains a ball of radius $c_1$ centered at $p^*_{\pi}$. In addition, $p^*_{\pi}$ exists, is unique and well-separated, so for any $p \in \mathcal S$ with $|| p - p^*_{\pi} ||\geq \frac{c_3}{2J c'}$,  we have $2 || z_{\pi}(p) || \geq c_2 $, where $|| \cdot ||$ is the $l_2$ norm. 
\item[(ii)]  For all $\pi \in \Pi$, the singular values of the $J \times J$ Jacobian matrix $\nabla_p z_{\pi}(p^*_{\pi})$  are bounded between $c_3$ and $c_4 $.
\item[(iii)] For $r \in \{d, y\}$, $w \in \{0, 1\}$ and $x \in \mathcal X$, $\mu_w^r(x, p) = \mathbb E[r(B(w), X, p) \mid X = x]$ is twice continuously differentiable in $p$ with first and second derivatives bounded uniformly by $c'$. 
\end{enumerate}  
\end{assumption} 

These are sufficient for showing that market-clearing cutoffs are smooth in $\pi$, which is a key step in the proof of Theorem \ref{thm:global_target}. They are also useful for providing uniform control of approximate market-clearing cutoffs, which is needed in the next section. Appendix \ref{app:upa} verifies these conditions in Example \ref{ex1} under bounds on the density of bids. 

For $r \in \{d, y\}$, let $\tau^r(x, p) = \mu^r_1(x, p) - \mu^r_0(x, p)$ denote the conditional average effect of the treatment at fixed cutoffs $p$.  The externality-adjusted outcome $q_{\pi}(b, x, p) = y(b, x, p) - \nu^*_{\pi} (d(b, x, p) - s^* )$ is a buyer's outcome adjusted for the cost their demand imposes on others through the capacity constraints, where $\nu^*_{\pi} = \nabla_p^{\top} y_{\pi}(p^*_{\pi}) [\nabla_p z_{\pi}(p^*_{\pi})]^{-1}$. Theorem \ref{thm:global_target} provides a condition on $\tau^{q_{\pi}}(x, p^*_{\pi}) = \tau^y(x, p^*_{\pi}) - \nu^*_{\pi} \tau^d(x, p^*_{\pi})$ that any optimal rule must satisfy when $\Pi$ is unrestricted.
 
\begin{theorem}\label{thm:global_target} Let $\Pi$ be the class of all measurable functions from $\mathcal X$ to $[0, 1]$ and suppose Assumption \ref{as:continuum_regularity} holds. For any $\pi^* \in \arg \max_{\pi \in \Pi} V^*(\pi)$, for almost all $x \in \mathcal X$, $\pi^*(x) = 1$ when $\tau^{q_{\pi^*}}(x, p^*_{\pi^*}) > 0$,  $\pi^*(x) =0$ when $\tau^{q_{\pi^*}}(x, p^*_{\pi^*}) < 0$, and $0 \leq \pi^*(x) \leq 1$ when $\tau^{q_{\pi^*}}(x, p^*_{\pi^*}) = 0$. 
\end{theorem}

The Conditional Average Treatment Effect (CATE) on the externality-adjusted outcome consists of two components. The first captures the direct impact on outcomes of treating buyers with $X = x$, while holding market-clearing cutoffs fixed. The second accounts for spillovers through the market-clearing cutoffs: treating these buyers shifts demand by $\tau^d(x, p^*_{\pi})$, and $\nu^*_{\pi}$ is the marginal effect of this demand shift on the outcomes of all buyers. A group is treated only if the sum of the two effects is positive. This is in contrast to the optimal rule without interactions through the market, where only the sign of the CATE on outcomes matters. Randomized policies can strictly improve on deterministic rules even without covariates, as demonstrated by the simple parametric example that follows. 

\paragraph{Randomized Optimal Policy in a Uniform-Price Auction. } The policy is a treated share $\pi \in [0,1]$, and bids/values are $B(0) \sim U[0,1]$, $B(1) = B(0) + \Delta$, where $\Delta \in [0, \bar\delta]$ is independent of $B(0)$ with mean $\delta$ and variance $ 0 < \sigma^2_{\Delta} < \delta ^2$, and $\bar\delta < 1 - s^*$, $\delta < s^*$. Then $p^*_{\pi} = 1 - s^* + \pi\delta$ and the market-clearing cutoff offsets the mean increase in bidder value. Here,  when outcome is bidder surplus, although the direct effect is always positive, $\tau^{q_{\pi}}(p^*_{\pi}) = \sigma^2_{\Delta}/2 + \delta^2(1/2 - \pi)$ and no deterministic policy is optimal. The unique optimal rule randomizes with $\pi^* = 1/2 + \sigma^2_{\Delta}/(2\delta^2)$, which raises the variance of bidder values and increases expected surplus. 

In this section, we characterized the continuum-market estimands $\tau^*_{\text{GTE}}$, $V^*(\pi)$ and $\pi^*$ in terms of known functionals of $F_w$ for $w \in \{0, 1\}$, the distribution of bids and covariates under the all-treated and all-control counterfactuals. In the next section, we use this representation for estimation and policy learning. 

\section{Estimation in the Continuum Market} 
\label{sec:estimation} 

\subsection{Policy Value Estimation} 
\label{sec:est_gte}

We observe an i.i.d. sample $(B_i, W_i, X_i)_{i=1}^n$ from $F_{\text{obs}}$, where $B_i = B_i(W_i)$ is the report under observed treatment. We first describe how to use this data to estimate $V^*(\pi)$ and $\tau^*_{\text{GTE}}$. Given our identification result in Proposition \ref{prop:id}, a natural approach is plug-in: estimate the conditional
distribution of reports given $X_i$ for observations with $W_i=w$, and combine with the empirical distribution of $X_i$ to estimate the counterfactual distribution of reports and covariates $\hat F^{\text{SM}}_w$. Applying the allocation mechanism to this estimated distribution with capacity constraints $s^*$ then yields estimates of the counterfactual cutoffs and welfare. 

When the conditional distribution of reports can be estimated accurately, 
this plug-in procedure can perform well. The difficulty is that it requires
learning a complex nuisance object before estimating the counterfactuals of interest. 

\paragraph{Uniform-Price Auction (continued).} A plug-in estimator specifies a model for bids given covariates, for example $\log B_i(w) \mid X_i = x \sim N(x^{\top} \theta_w, \omega^2_w)$, and estimates $(\hat \theta_w, \hat \omega^2_w)$ by regressing $\log B_i$ on $X_i$ among bidders with $W_i = w$. Paired with the empirical distribution of $X_i$, this gives $$\hat F^{\text{SM}}_w(b, x) = \frac1n \sum_{i=1}^n \mathbbm 1(X_i \le x)\, \Phi\!\left(\frac{\log b - X_i^\top \hat\theta_w}{\hat\omega_w}\right).$$ The estimated cutoff $\hat P^{\text{SM}}_w$ is the $(1 - s^*)$ quantile of the marginal distribution of bids under $\hat F^{\text{SM}}_w$, and 
\begin{equation} \label{eq:sm_gte} \hat \tau^{\text{SM}}_{\text{GTE}} = \mathbb E_{\hat F^{\text{SM}}_1}[y(B, X, \hat P^{\text{SM}}_1)] - \mathbb E_{\hat F^{\text{SM}}_0}[y(B, X, \hat P^{\text{SM}}_0)], \end{equation}  where the expectation is conditional on the estimated cutoffs. This is the estimator \texttt{SM-GTE} in Section \ref{sec:simulation}. The conditional distribution of bids could also be estimated non-parametrically \citep{athey2007nonparametric}. 

The plug-in estimator is biased when the model for the conditional distribution is misspecified. While a non-parametric estimator for the conditional distribution of reports can be used instead, this leads to slow estimation rates that depend on the dimension of $\mathcal X$. The representation of $\tau^*_{\text{GTE}}$ as the solution to a set of moment conditions suggests that we can reduce the dimension of the nuisance estimation problem. 

Our estimator combines two ideas to accomplish this. First, we characterize $\tau^*_{\text{GTE}}$ using doubly-robust scores. Second, we localize the regression components of
these scores around a pilot estimate of the counterfactual cutoffs. This reduces the nuisance problem from estimating a conditional density to estimating a handful of regressions. 

Let $\eta = (h, \{g^r_w\}_{r \in \{d, y\}, w \in \{0, 1\}})$ collect candidate nuisance functions, where $h$ is a candidate propensity score and $g^r_w$ is a candidate conditional mean of $r$ under treatment $w$. For $r \in \{d, y\}$ and $w \in \{0, 1\}$, the doubly-robust score is
\begin{equation} \label{eq:dr}
\Gamma^r_{wi}(p; \eta) = g^r_w(X_i) + \frac{\mathbbm 1(W_i = w)}{h(X_i)^{w} (1 - h(X_i))^{1 - w}} \left ( r(B_i, X_i, p) - g^r_w(X_i) \right ).
\end{equation}
 For a policy $\pi$, let $\Gamma^r_{\pi i}(p; \eta) = \pi(X_i) \Gamma^r_{1i}(p; \eta) + (1 - \pi(X_i)) \Gamma^r_{0i}(p; \eta)$, so that $\Gamma^r_{\pi_w i} = \Gamma^r_{wi}$, and let $\eta^*_{\pi} = (e, \{\mu^d_w(\cdot, p^*_{\pi}), \mu^y_w(\cdot, p^*_{\pi})\})$ collect the true nuisance functions localized at $p^*_{\pi}$. Then $z_{\pi}(p) = \mathbb E[\Gamma^d_{\pi i}(p; \eta^*_{\pi})] - s^*$ and $y_{\pi}(p) = \mathbb E[\Gamma^y_{\pi i}(p; \eta^*_{\pi})]$, so  a set of doubly-robust identifying conditions for $V^*(\pi)$ is
\begin{equation} \label{eqn:drmoment}
\mathbb E[\Gamma^d_{\pi i}(p^*_{\pi}; \eta^*_{\pi})] = s^*, \qquad V^*(\pi) = \mathbb E[\Gamma^y_{\pi i}(p^*_{\pi}; \eta^*_{\pi})].
\end{equation}

Equation \eqref{eqn:drmoment} is one of several possible sets of identifying conditions under Proposition \ref{prop:id}. For example, setting $g^r_w = 0$ corresponds to the inverse propensity-weighted approach. We prefer the doubly-robust version, since inference and efficiency then require much weaker conditions on the estimated propensity score \citep{hirano2003efficient}.

The estimator of \citet{chernozhukov2018double} solves the sample analog of \eqref{eqn:drmoment}, but requires an estimate of $\mu^d_w(x, p)$ for every $p \in \mathcal S$.\footnote{This is still a complex nuisance object. In the uniform-price auction, it is equivalent to learning the CDF of the conditional distribution of bids at every threshold.} Following the localization approach of \citet{kallus2019localized}, we instead estimate the conditional means only at a pilot estimate of $p^*_{\pi}$.  We use $5$-fold cross-fitting: for each evaluation fold, two of the remaining folds are used to estimate the pilot cutoff, and the other two are used to estimate the nuisance functions localized to the pilot cutoff. To keep notation readable, we suppress the fold index.  For a candidate propensity score $h$, let $\gamma_{\pi i}(h) = \pi(X_i) W_i / h(X_i) + (1 - \pi(X_i))(1 - W_i) / (1 - h(X_i))$ denote the IPW weights for policy $\pi$. The proposed Localized Doubly-Robust (LDR) estimator requires three steps: 
\begin{enumerate}
\item[(i)] \textit{Pilot cutoff.} On the pilot fold, estimate the propensity score by $\tilde e$, and let $\tilde P_{\pi}$ be a market-clearing cutoff of the economy $[\tilde F_{\pi}, s^*]$, where $\tilde F_{\pi}$ is the empirical distribution of $(B_i, X_i)$ on the pilot fold, with observation $i$ re-weighted by $\gamma_{\pi i}(\tilde e)$.
\item[(ii)] \textit{Localized nuisances.} On the nuisance fold, estimate the propensity score by $\hat e$ and, for $r \in \{d, y\}$, let $\hat \mu^r_w(\cdot, \tilde P_{\pi})$ be a flexible regression of $r(B_i, X_i, \tilde P_{\pi})$ on $(W_i, X_i)$, evaluated at $W_i = w$. Let $\hat \eta_{\pi} = (\hat e, \{\hat \mu^d_w(\cdot, \tilde P_{\pi}), \hat \mu^y_w(\cdot, \tilde P_{\pi})\})$.
\item[(iii)] \textit{Debiased estimation.} Let $\hat P_{\pi}$ be the market-clearing cutoff of the economy $[\hat F^{DR}_{\pi}, \hat s_{\pi}]$, where $\hat F^{DR}_{\pi}$ is the empirical distribution of $(B_i, X_i)$ with observation $i$ re-weighted by $\gamma_{\pi i}(\hat e)$ and
\[ \hat s_{\pi} = s^* + \frac{1}{n} \sum_{i=1}^n \sum_{w \in \{0, 1\}} \pi(X_i)^w (1 - \pi(X_i))^{1 - w} ( \gamma_{\pi_w i}(\hat e) - 1 ) \hat \mu^d_w(X_i, \tilde P_{\pi}). \] Then, $\hat V_n(\pi) = \frac{1}{n} \sum \limits_{i=1}^n \Gamma_{\pi i}^y (\hat P_{\pi}; \hat \eta_{\pi})$ and 
\begin{equation} \label{eq:ldr_gte}
\hat \tau_{\text{GTE}} = \hat V_n(\pi_1) - \hat V_n(\pi_0).
\end{equation} 
\end{enumerate}

Step (i) is a pilot step estimator of the market-clearing cutoffs using the propensity score only to re-weight the empirical distribution of reports to mimic $F_{\pi}$. Step (iii) is a debiased estimator of the market-clearing cutoffs that combines the propensity score and a localized regression adjustment. When the distribution of buyers is an empirical measure, the mechanisms in Example \ref{ex1}--\ref{ex2} have computationally efficient algorithms available for finding the market-clearing cutoffs. For example, in step (iii), the uniform-price auction sorts bids, and the auction awards the item to the largest set of top bids whose average weight $\frac{1}{n_{\text{win}}}\sum \limits_{k} \gamma_{\pi k} (\hat e)$ is at most $\hat s_{\pi}$. 

To prove the asymptotic normality of $\hat \tau_{\text{GTE}}$, we require regularity conditions on individual-level allocations and outcomes and assumptions on the nuisance estimators. Assumption \ref{as:regularo} allows individual allocations to be discontinuous in $p$, but restricts their complexity as functions of $p$. For a class $\mathcal F$ of real-valued functions and a probability measure $Q$, let $\| f \|_{L_2(Q)} = ( \int f^2 \, dQ )^{1/2}$. The covering number $N(\epsilon, \mathcal F, L_2(Q))$ is the smallest $m$ for which there exist $f_1, \ldots, f_m$ with $\min_k \| f - f_k \|_{L_2(Q)} \leq \epsilon$ for every $f \in \mathcal F$, and the uniform covering number $\sup_Q N(\epsilon, \mathcal F, L_2(Q))$ takes the supremum over all $Q$ with finite support \citep{vaart1997weak}.

\begin{assumption} \label{as:regularo} \textbf{Regularity of Outcomes.} For each $r \in \{d_1, \ldots, d_J, y\}$: \begin{enumerate} 
\item[(i)] There are constants $C, \alpha > 0$ such that $\mathcal F_r = \{ (b, x) \mapsto r(b, x, p) : p \in \mathcal S \}$ satisfies $\sup_Q N(\epsilon, \mathcal F_r, L_2(Q)) \leq C \epsilon^{-\alpha}$ for all $0 < \epsilon < 1$; 
\item[(ii)] There is a constant $L > 0$ such that $\mathbb E[( r(B(w), X, p) - r(B(w), X, p') )^2] \leq L \| p - p' \|$ for all $p, p' \in \mathcal S$ and $w \in \{0, 1\}$.
\end{enumerate} 
\end{assumption}

In Example \ref{ex1}, allocations are indicators in bids, and in Example \ref{ex2}, they are finite products of indicators in the priority scores; both meet Assumption \ref{as:regularo}(i) and (ii), as long as the report distribution has bounded density. See Appendix \ref{app:upa} for a verification in the uniform-price auction case.

\begin{assumption} \label{as:nuisance} \textbf{Nuisance Estimation.}  Let $\mu_w(x, p) = (\mu^d_w(x, p)^{\top}, \mu^y_w(x, p))^{\top} \in \mathbb R^{J+1}$, let $\hat \mu_w$ be its estimate from step (ii), and let $\mathbb E_T[\cdot]$ denote the expectation over an evaluation observation conditional on the training data. 
 \begin{enumerate} 
\item[(i)] Almost surely, $\hat e(X_i) \in (\kappa, 1 - \kappa)$ for some $\kappa > 0$ and $\sup_{w, x, p} \| \hat \mu_w(x, p) \|_\infty \leq M < \infty$. 
\item[(ii)] There is a $c > 0$ such that with probability at least $1 - e^{-cn}$, for all $\pi \in \Pi$, \begin{align*}
\big( \mathbb E_T [ \| \hat \mu_w(X_i, \tilde P_{\pi}) - \mu_w(X_i, \tilde P_{\pi}) \|^2 ] \big)^{1/2} &\leq \rho_{\mu, n}, & \big( \mathbb E_T [ ( \hat e(X_i) - e(X_i) )^2 ] \big)^{1/2} &\leq \rho_{e, n}, \\
\| \tilde P_{\pi} - p^*_{\pi} \| &\leq \rho_{\theta, n}, & \Big\| \tfrac{1}{n} \sum_{i=1}^n \Gamma^d_{\pi i}(\hat P_{\pi}; \hat \eta_{\pi}) - s^* \Big\| &\leq \rho_{g, n},
\end{align*} where $\rho_{e, n} + \rho_{\mu, n} + \rho_{\theta, n} = o(1)$, $\rho_{e, n} (\rho_{\mu, n} + \rho_{\theta, n}) = o(n^{-1/2})$ and $\rho_{g, n} = o(n^{-1/2})$.
\end{enumerate} 
\end{assumption}

Our asymptotic framework follows \citet{azevedo2016supply}, where $n \to \infty$ with $J$ and fractional capacity $s^*$ fixed, so that the number of units of each item, $n s^*_j$, grows in proportion to the number of buyers. Assumption \ref{as:nuisance}(ii) requires that the product of the propensity score error and the error in either the pilot cutoff or the conditional means is $o(n^{-1/2})$, so that each may converge slowly, which allows for flexible estimators for the nuisances. It also requires that $\hat P_{\pi}$ approximately clears the re-weighted market up to an error that is $o(n^{-1/2})$. In Example \ref{ex1}, since individual demand is discrete and supply is a real number, the estimated cutoff can lead to a small amount of unused capacity, but this leftover capacity is $O(1/n)$. 

Our first result on estimation provides an asymptotically linear representation for the policy value, where the influence function is the doubly-robust score of the externality-adjusted outcome $q_{\pi}$, which is $\Gamma^q_{\pi i}(p; \eta) = \Gamma^y_{\pi i}(p; \eta) - \nu^*_{\pi} ( \Gamma^d_{\pi i}(p; \eta) - s^* )$.

\begin{theorem} \label{thm:norm}
Under Assumptions \ref{as:id_new}--\ref{as:nuisance}, for each $\pi \in \Pi$, $\hat V_n(\pi) = \frac{1}{n} \sum_{i=1}^n \Gamma^q_{\pi i}(p^*_{\pi}; \eta^*_{\pi}) + o_p(n^{-1/2})$.
\end{theorem}

When the policy value is estimated using a sample of buyers, the estimated value fluctuates because of two components: a) the variation in the direct effect on outcomes  and b) fluctuations in the estimated market-clearing cutoffs. With known nuisance functions, standard techniques for method-of-moments estimators can be used to prove Theorem \ref{thm:norm} for a propensity-score-based estimator. With an unknown propensity score and a doubly-robust estimator, we show that even when estimated nuisance functions depend on market-clearing cutoffs, their estimation error does not have a first-order impact on the error of the estimator.\footnote{Under weaker entropy conditions than in Assumption \ref{as:regularo}, the main result in \citet{kallus2019localized} can be used to prove Theorem \ref{thm:norm}. However, the stronger conditions that we impose, which are met by economic mechanisms used in practice, provide a) a more concise proof and b) uniform control over $\Pi$ of the fluctuations of the estimated market-clearing cutoffs, which is needed in the next section. }

Theorem \ref{thm:norm} implies asymptotic normality of $\hat \tau_{\text{GTE}}$. Appendix \ref{app:eff} also shows that the influence function from Theorem \ref{thm:norm} is efficient. The proof uses the methodology presented in \citet{bickel1993efficient} and \citet{newey1990semiparametric}, and is closely related to the bound for quantile treatment effects in \citet{firpo2007efficient}.

\begin{corollary} \label{cor:norm}
Under the assumptions of Theorem \ref{thm:norm}, $\sqrt n (\hat \tau_{\text{GTE}} - \tau^*_{\text{GTE}}) \rightarrow_D N(0, \sigma^2)$, where $\sigma^2 = \mathrm{Var}[\Gamma^q_{1i}(p^*_1; \eta^*_1) - \Gamma^q_{0i}(p^*_0; \eta^*_0)]$. Furthermore, $\sigma^2$ meets the semi-parametric efficiency bound for $\tau^*_{\text{GTE}}$. 
\end{corollary}

 Theorem \ref{thm:norm} tells us that $\hat V_n(\pi)$ can be used to effectively estimate the value of a policy, or the difference in value between two alternative policies. Next, we discuss the implications of using this estimator to choose an optimal policy from an infinite class $\Pi$. 

\subsection{Policy Learning in a Restricted Class} 
\label{sec:est_ewm}

Theorem \ref{thm:global_target} characterizes the optimal unrestricted rule as a fixed point: the score $\tau^{q_{\pi^*}}(x, p^*_{\pi^*})$ that determines $\pi^*(x)$ depends on $\pi^*$ itself. While this is useful for understanding the structure of the optimal rule, it is not well-suited for estimation. A treatment rule that thresholds the CATE on the externality-adjusted outcome at the observed rule has different market-clearing cutoffs than the observed market, and so may not meet the conditions in Theorem \ref{thm:global_target}. While an iterative approach is possible, analyzing even the convergence of such a procedure is beyond the scope of this paper. 

Instead, in this section, we constrain $\Pi$ to be a VC class, which allows for an empirical welfare maximization approach that has asymptotic regret guarantees. Furthermore, this constraint is often useful in practice where interpretable treatment rules are desirable. For example, the class of logistic treatment rules $\Pi = \{ x \mapsto \Lambda(\beta_0 + x^{\top} \beta_1) : (\beta_0, \beta_1) \in \mathbb R^{k+1} \}$ given in Section \ref{sec:model} is both a VC class and interpretable. The estimated rule maximizes the estimated value, $\hat \pi \in \arg \max \limits_{\pi \in \Pi} \hat V_n(\pi)$.  As in empirical welfare maximization without spillovers \citep{kitagawa2018should}, the objective is not concave, so we use a derivative-free optimizer with multiple starting values.\footnote{Specifically, we use a version of Nelder-Mead from Julia's \texttt{Optim.jl}.} For the logistic class, this is a maximization over $(\beta_0, \beta_1) \in \mathbb R^{k+1}$. Each evaluation of $\hat V_n(\pi)$ follows the procedure described in Section \ref{sec:est_gte}. 

For the uniform convergence required for regret bounds, we impose Assumption \ref{as:anuis}, which is an additional assumption on the nuisance estimators.

\begin{assumption} \label{as:anuis}  For each $r \in \{d_1, \ldots, d_J, y\}$ and $w \in \{0, 1\}$, there are constants $C, \alpha > 0$ such that, with probability $1 - o(1)$, $\mathcal F_{\hat \mu^r_w} = \{ x \mapsto \hat \mu^r_w(x, p) : p \in \mathcal S \}$ satisfies $\sup_Q N(\epsilon, \mathcal F_{\hat \mu^r_w}, L_2(Q)) \leq C \epsilon^{-\alpha}$ for all $0 < \epsilon < 1$.

\end{assumption}

Although we allow the estimated conditional mean functions to be complex functions of $X_i$, they must be relatively simple functions of $p$. Since we already impose a metric entropy condition on individual-level outcome functions in $p$, in some cases, such as for the $k$-nearest-neighbors estimator used for allocations in Section \ref{sec:empirical}, this is automatically implied by Assumption \ref{as:regularo}. For more general machine learning estimators, verifying this type of condition  may require additional effort. 

We can now provide a convergence rate for $\hat P_{\pi}$ that holds uniformly over $\Pi$, which is needed for regret bounds and  strengthens the pointwise convergence used in Theorem \ref{thm:norm}. 

\begin{proposition} \label{prop:unifP} Under Assumptions \ref{as:id_new}--\ref{as:anuis}, if $\Pi$ is a VC class of dimension $v$, then $\sup \limits_{\pi \in \Pi} \| \hat P_{\pi} - p^*_{\pi} \| = O_p(n^{-1/2})$.
\end{proposition}

Proposition \ref{prop:unifP}, proved as part of Lemma \ref{lem:concp} in the appendix, goes beyond the analysis in \citet{kallus2019localized}, who analyze a single moment-based estimator. In our setting, there are infinitely many policies, which each are associated with a set of moment conditions and nuisance estimators. First, we show that oracle doubly-robust scores have uniform covering numbers that are polynomial in $1/\epsilon$, since $\Pi$ is a VC class and we restrict the complexity of allocations, outcomes and conditional mean estimators in $p$. Then, a tail bound for empirical processes from \citet{vaart1997weak} controls the difference between oracle and population excess demand, with sample splitting and the rates in Assumption \ref{as:nuisance} controlling nuisance error. The uniform conditions in Assumption \ref{as:continuum_regularity} allow us to convert the bound on excess demand at $\hat P_{\pi}$ into the cutoff bound in Proposition \ref{prop:unifP}. 

Without market-clearing cutoffs, $\hat V_n(\pi)$ is linear in $\pi$ and the VC-class assumption on $\Pi$ leads to regret bounds \citep{kitagawa2018should, athey2021policy}. Here, Proposition \ref{prop:unifP} is a crucial first step for Theorem \ref{thm:regret}, because $\pi$ also enters $\hat V_n(\pi)$ non-linearly, through $\hat P_{\pi}$ and the nuisance functions. 

\begin{theorem} \label{thm:regret}  Under Assumptions \ref{as:id_new}--\ref{as:anuis}, also assume $\Pi$ is a VC class of dimension $v$. Then, regret in the continuum market from the empirical welfare maximization procedure decays asymptotically at a $1/\sqrt n$ rate: 
\begin{equation*} 
\begin{split} 
   &     \sup \limits_{\pi \in \Pi} V^*(\pi) - V^*(\hat \pi)     = O_p\left ( \frac{1}{\sqrt n} \right). 
    \end{split}
    \end{equation*} 
\end{theorem}

This is a notable result -- when interactions are mediated by a cutoff mechanism, it is possible to learn the optimal policy at an asymptotic rate that matches the lower bound for policy learning without spillover effects \citep{athey2021policy}. The techniques developed here are useful beyond our specific setting. They can be used to prove regret bounds whenever the estimated value of a decision depends on a decision-specific method-of-moments estimator, including other models of equilibrium that are distinct from our model of market clearing through a centralized mechanism.

\section{Finite-Market Guarantees} 
\label{sec:finite}

The continuum economy of Section \ref{sec:model} describes a population-wide functional interaction, in the language of  \citet{manski2013identification}, with a convenient moment-based representation of interactions where an individual buyer has a negligible impact on the market.  A finite market with $n$ buyers involves interactions that do not have such a simple representation. Each buyer's report affects the cutoffs, so an individual's outcome depends on all other buyers' reports and treatments. Individual allocations and outcomes are not smooth functions of the cutoffs. Nevertheless, finite-market counterfactuals  are often more policy relevant. In this section, we show that the continuum market is a sufficiently good approximation to the finite market, so that inference and regret results from the previous section extend to finite-market quantities. 

In the finite market, there are $n$ buyers, with $(B_i(1), B_i(0), W_i, X_i)_{i=1}^n$ drawn i.i.d. from the superpopulation of Section \ref{sec:model}, and we observe $(B_i, W_i, X_i)$, where $B_i = B_i(W_i)$. For a treatment vector $\bm w \in \{0, 1\}^n$, buyer $i$ reports $B_i(w_i)$, so that, as in Assumption \ref{as:id_new}(i), a buyer's report depends on $\bm w$ only through their own treatment. The counterfactual finite economy is $[F_n(\bm w), s^*]$, where $F_n(\bm w)$ is the empirical distribution of $(B_i(w_i), w_i, X_i)_{i=1}^n$, and the (approximate) market-clearing cutoffs $P_n(\bm w) \in \mathcal S$ depend on the reports of all buyers.

\begin{assumption} \label{as:approx_p} \textbf{Finite-Market Clearing.} There is a sequence $a_n = o(n^{-1/2})$ and a constant $c > 0$ such that, with probability at least $1 - e^{-cn}$, for every $\bm w \in \{0, 1\}^n$, $P_n(\bm w)$ satisfies $\| \mathbb E_{F_n(\bm w)}[d(B, X, P_n(\bm w))] - s^* \| \leq a_n$. 
\end{assumption}

There may be multiple cutoffs that meet Assumption \ref{as:approx_p}, but we assume a deterministic selection rule for the choice of $P_n(\bm w)$. The finite-market value of a policy is $\bar V_n(\pi) =  \frac 1n  \sum_{i=1}^n  \mathbb E_{\pi} \big[ y(B_i(W_{\pi, i}), X_i, P_n(\bm W_{\pi})) \big]$. The expectation $\mathbb E_{\pi}[\cdot]$ is over $W_{\pi, i} \sim \mbox{Bernoulli}(\pi(X_i))$ drawn independently for $i \in \{1, \ldots n\}$, holding $(B_i(1), B_i(0), X_i)_{i=1}^n$ fixed. Our estimand averages over treatment assignments since our policy class $\Pi$ may include randomized treatment rules. For the GTE, the finite-market estimand is:  
\begin{equation} \label{eq:fgte}
\bar \tau_{\text{GTE}} = \bar V_n(\pi_1) - \bar V_n(\pi_0) = \frac 1n \sum_{i=1}^n \Big( y(B_i(1), X_i, P_n(\bm 1_n)) - y(B_i(0), X_i, P_n(\bm 0_n)) \Big),
\end{equation}
where $\bm 1_n$ and $\bm 0_n$ are $n$-vectors of ones and zeros.  We use notation that makes explicit the dependence of market-clearing cutoffs on the vector of treatments. Although it is not explicit in the notation, cutoffs also depend on reports and characteristics of all of the buyers. 

\paragraph{Uniform-Price Auction (continued).} A uniform-price auction in a finite market allocates $m = \lfloor n s^* \rfloor$ units to $n$ bidders.  $P_n(\bm w)$ is typically defined as the $m$-th highest bid, although Assumption \ref{as:approx_p}  means we can accommodate any market-clearing rule that deterministically selects a cutoff between the $ (m+1)$th and the $m$th highest bid, with  $a_n = 1/n$. The counterfactual cutoff $P_n(\bm 1_n)$ is a sample quantile of $(B_i(1))_{i=1}^n$ and $p^*_1$ is the $(1 - s^*)$ quantile of the population distribution of $B(1)$. $\bar \tau_{\text{GTE}}$ is the difference in average outcomes between the all-treated and all-control auctions for the $n$ bidders who appear in the observed market, while $\tau^*_{\text{GTE}}$ is the corresponding continuum-market quantity. 

Unlike for the Sample Average Treatment Effect (SATE), an unbiased estimator with finite variance is not immediately available for $\bar \tau_{\text{GTE}}$, because of the interactions among all buyers that occur through the market-clearing cutoffs. A propensity-score-based approach fails because we never observe a single market that is both all-treated and all-control. As for the SATE, however, a consistent estimator is available. Our first result shows that asymptotically, the continuum-market estimand is close to the finite-market estimand. This result sharpens the convergence results in \citet{azevedo2016supply}, who bound the probability that finite-market cutoffs deviate from continuum-market cutoffs, by providing an asymptotic order for the deviation of the finite-market counterfactual, and an explicit form for its leading term. 

\begin{theorem} \label{thm:moment} Under Assumptions \ref{as:id_new}--\ref{as:regularo} and \ref{as:approx_p}, $\bar V_n(\pi) - V^*(\pi) = O_p(n^{-1/2})$ for each $\pi \in \Pi$, and
\begin{equation} \label{eq:moment}
\bar \tau_{\text{GTE}} = \frac 1n \sum_{i=1}^n \Big( q_1(B_i(1), X_i, p^*_1) - q_0(B_i(0), X_i, p^*_0) \Big) + o_p(n^{-1/2}).
\end{equation}
\end{theorem}

Theorem \ref{thm:moment} implies that an average of dependent terms, linked through $P_n(\bm w)$, is approximated by an average of i.i.d. terms. Again, the externality-adjusted outcomes are useful for capturing the fluctuations in outcomes and in market-clearing cutoffs that occur in finite samples. Combining \eqref{eq:moment} with Theorem \ref{thm:norm}, $\sqrt n (\hat \tau_{\text{GTE}} - \bar \tau_{\text{GTE}}) \rightarrow_D N(0, \bar \sigma^2)$, where $\bar \sigma^2 = \mathrm{Var}\big[ \Gamma^q_{1i}(p^*_1; \eta^*_1) - \Gamma^q_{0i}(p^*_0; \eta^*_0) - q_1(B_i(1), X_i, p^*_1) + q_0(B_i(0), X_i, p^*_0) \big]$. $\bar \sigma^2$ is not identified, since it depends on the joint distribution of $B_i(1)$ and $B_i(0)$. However, Proposition \ref{cor:cons} shows that inference that is valid for the continuum-market estimand is conservative for the finite-market estimand, which is often the counterfactual of interest for a policymaker. Asymptotically valid 95\% confidence intervals for $\tau^*_{\text{GTE}}$ are constructed by 
\begin{equation} \label{eq:ci_est} 
\hat \tau_{\text{GTE}} \pm 1.96 \cdot \hat \sigma/\sqrt n, 
\end{equation} 
 where $\hat \sigma^2 = \frac 1n \sum_{i=1}^n ( \hat \Gamma^q_{1i} - \hat \Gamma^q_{0i} - \hat \tau_{\text{GTE}} )^2$ and $\hat \Gamma^q_{wi} = \Gamma^y_{wi}(\hat P_w; \hat \eta_w) - \hat \nu_w ( \Gamma^d_{wi}(\hat P_w; \hat \eta_w) - s^* )$. Each component of $\hat \nu_w$ is estimated using finite differences, with a step size converging to 0 as $n \to \infty$. Consistency of $\hat \sigma^2$ then follows from Theorem 2 of \citet{kallus2019localized}. %

\begin{proposition} \label{cor:cons} Under the assumptions of Theorem \ref{thm:norm}, $\sigma^2 \geq \bar \sigma^2$. 
\end{proposition}

Just as $\bar \tau_{\text{GTE}}$ may be preferred as an estimand to $\tau^*_{\text{GTE}}$, a policymaker may prefer to define regret relative to the maximizer of finite-market value $\bar V_n(\pi)$ rather than continuum-market value $V^*(\pi)$. The dependence structure and lack of smoothness in the allocations in a finite market make it challenging to characterize finite-market regret directly. Instead,  we provide uniform control on the error $\bar V_n(\pi) - V^*(\pi)$, so that the regret results in the previous section extend to settings where the reference value is the optimal value in a market with $n$ buyers.  

\begin{theorem} \label{thm:regret_finite} Under Assumptions \ref{as:id_new}--\ref{as:approx_p}, if $\Pi$ is a VC class of dimension $v$, then \[ \sup \limits_{\pi \in \Pi} \bar V_n(\pi) - \bar V_n(\hat \pi) = O_p(n^{-1/2}).\] 
\end{theorem}

The estimated rule thus performs nearly as well as the best rule in $\Pi$ for the $n$ buyers in the market, at the rate of Theorem \ref{thm:regret}. The additional step in the proof is a finite-market analogue of Proposition \ref{prop:unifP}, which controls the deviation of the finite market from the continuum market. In simulations of a market for schools in Appendix \ref{ap:da_sim}, we show that even in moderate-sized markets, the regret of the empirical welfare maximization procedure is low.

\section{Simulation}  
\label{sec:simulation}
\label{sec:auction}

In this section, we illustrate the robustness properties of the LDR estimator compared to structural modeling approaches using a simulation of a uniform-price auction where bidders' values are generated from different distributions.

In the simulation, treatment affects bids to the auction. There is a $20$-dimensional vector of covariates that is correlated with the bids and affects the probability of selecting the treatment. The auction has $s^*=1/2$, so the top 50\% of bidders in the auction receive a single unit of the good. The treatment affects outcomes through a shift in the distribution of bids submitted to the auction and through a shift in the market-clearing price. The outcome of interest is bidder surplus, where bids submitted to the auction are equal to a bidder's value. The data-generating process is explicitly described in Appendix \ref{app:sim}. For each bidder, we observe the bid $B_i$, the treatment $W_i$, and pre-treatment covariates $X_i$. We compute RMSE and bias of three estimators when the target estimand is $\bar \tau_{\text{GTE}}$ by repeatedly sampling a finite market of size $n=500$, $n=1{,}000$ and $n=10{,}000$. The estimators take as input $(Y_i, B_i, W_i, X_i)_{i=1}^n$, where $Y_i = y(B_i, X_i, P_n(\bm W))$ is the outcome at the observed cutoffs. The estimators are:
\begin{enumerate} 
\item A doubly-robust estimator of the Average Treatment Effect (\texttt{DR-ATE}), which adjusts for selection-on-observables but ignores interactions. \[ \hat \tau_{\text{ATE}} = \frac1n \sum_{i=1}^n \big( \Gamma^y_{1i}(P_n(\bm W); \hat \eta_{\text{obs}}) - \Gamma^y_{0i}(P_n(\bm W); \hat \eta_{\text{obs}}) \big), \]  where $\hat \eta_{\text{obs}}$ collects $\hat e$, the estimated propensity score, and  $\hat \mu^y_w(\cdot, P_n(\bm W))$, a regression of $Y_i$ on $(W_i, X_i)$ evaluated at $W_i = w$. The nuisances are estimated with cross-fitted generalized random forests \citep{athey2019generalized}. The cutoffs are held fixed at the observed market-clearing cutoffs. 
\item A structural model-based estimator of $\bar \tau_{\text{GTE}}$ (\texttt{SM-GTE}), which assumes the conditional distribution of bids is log-normal, as in \eqref{eq:sm_gte}. 
\item The localized doubly-robust estimator (\texttt{LDR-GTE}), which follows the  steps in Section \ref{sec:est_gte} leading to  \eqref{eq:ldr_gte}. The nuisance functions are estimated using generalized random forests. 
\end{enumerate}

\begin{table}[ht]
\centering
\begin{tabular}{lcccccc}
\toprule
 & \multicolumn{2}{c}{$n = 500$} & \multicolumn{2}{c}{$n = 1{,}000$} & \multicolumn{2}{c}{$n = 10{,}000$} \\
\cmidrule(lr){2-3}\cmidrule(lr){4-5}\cmidrule(lr){6-7}
 & \textbf{Bias} & \textbf{RMSE} & \textbf{Bias} & \textbf{RMSE} & \textbf{Bias} & \textbf{RMSE} \\
\midrule
\texttt{\footnotesize DR-ATE} & 0.264 & 0.267 & 0.261 & 0.262 & 0.246 & 0.246 \\
\texttt{\footnotesize SM-GTE} & -0.001 & 0.026 & 0.001 & 0.017 & -0.000 & 0.005 \\
\texttt{\footnotesize LDR-GTE} & 0.011 & 0.037 & 0.005 & 0.024 & -0.001 & 0.007 \\
\bottomrule
\end{tabular}
\caption{Bids follow a log-normal distribution. Metrics are averaged over 100 simulations of each sample size from the data-generating process. } 
\label{tab:sim1} 
\end{table}

In Table \ref{tab:sim1}, as $n$ increases, the structural model-based estimator, which makes the correct parametric assumption on the bid distribution, converges the fastest. The LDR estimator does not make any parametric assumptions and instead uses flexible machine learning estimators for nuisance estimation.  It has an asymptotic distribution that does not depend on the estimation errors of the nuisance functions. The \texttt{DR-ATE} estimator compares treated and control bidders at the observed market-clearing cutoff, so it omits the shift in the cutoff between the all-treated and all-control markets. Its bias does not shrink with $n$.

In the second set of simulations, we generate bids conditional on covariates from a Weibull distribution, with shape $k=2$ and scale $\lambda$ that depends on covariates, rather than a log-normal distribution. The Weibull distribution is also commonly used in the literature on empirical auctions \citep[e.g.,][]{athey2011comparing}. Otherwise, the data-generating process is the same. We report the results of the simulation in Table \ref{tab:sim2}, where we continue to use the now misspecified log-normal estimator for the structural modeling estimator.

\begin{table}[ht]
\centering
\begin{tabular}{lcccccc}
\toprule
 & \multicolumn{2}{c}{$n = 500$} & \multicolumn{2}{c}{$n = 1{,}000$} & \multicolumn{2}{c}{$n = 10{,}000$} \\
\cmidrule(lr){2-3}\cmidrule(lr){4-5}\cmidrule(lr){6-7}
 & \textbf{Bias} & \textbf{RMSE} & \textbf{Bias} & \textbf{RMSE} & \textbf{Bias} & \textbf{RMSE} \\
\midrule
\texttt{\footnotesize DR-ATE} & 0.252 & 0.257 & 0.237 & 0.239 & 0.223 & 0.224 \\
\texttt{\footnotesize SM-GTE} & 0.051 & 0.083 & 0.046 & 0.061 & 0.049 & 0.050 \\
\texttt{\footnotesize LDR-GTE} & 0.022 & 0.047 & 0.013 & 0.034 & 0.002 & 0.009 \\
\bottomrule
\end{tabular}
\caption{Bids follow a Weibull distribution. Metrics are averaged over 100 simulations of each sample size from the data-generating process.} 
\label{tab:sim2}
\end{table}

This time, the structural modeling approach performs poorly, since the log-normal assumption is incorrect, which leads to significant bias.  The LDR estimator does not make any parametric assumptions on the bid distribution and continues to perform very well here. 

If a parametric model is correctly specified, then a maximum-likelihood estimator of that model is asymptotically linear and efficient. In addition, once the primitives of the model are specified and estimated, a variety of counterfactuals can often be evaluated, including those that are more complex than the estimand considered in this paper. The downside of this approach is that if the model is not correctly specified, then the estimator of $\tau^*_{\text{GTE}}$ will be asymptotically biased. Unfortunately, it can be challenging to specify a parametric model that captures the complexity and heterogeneity of individual choice behavior, especially in settings where possible reports are high-dimensional. The LDR estimator performs well, without requiring correct specification of a model of reports. 

In Appendix \ref{ap:da_sim}, we simulate a market with $J$ schools, and vary both $J$ and the number of seats per school. We show that across a range of choices of  $J\geq 10$, even for as few as 20 seats per school, the results in finite markets are well approximated by the asymptotics in Corollary \ref{cor:norm} and Theorem \ref{thm:regret_finite}. We show that confidence intervals from \eqref{eq:ci_est} have close to nominal (but conservative) coverage for $\bar \tau_{\text{GTE}}$, that $\tau^*_{\text{GTE}}$ well approximates $\bar \tau_{\text{GTE}}$, and that the value of optimal policies learned through an empirical welfare maximization procedure are close to the value of the oracle optimal policy. 

\section{Impact Evaluation in the Chilean School Market}
\label{sec:empirical}

In 2015, the Chilean government passed the Inclusion Law, which eliminated school-specific admissions criteria in favor of a centralized admission system based on deferred acceptance \citep{correa2019school}. It was intended to reduce socioeconomic segregation in the Chilean school system by removing discriminatory admissions criteria and reserving some seats at good schools for low-income families. Despite these changes, low-income families attend good-quality schools at a much lower rate than high-income families. One reason for this remaining gap is that some families may lack information about school quality or the returns to schooling. \citet{allende2019approximating} explore this hypothesis using an RCT that randomized information on nearby school quality.  Using a parametric model of school choice and a simulation of the centralized admission system, they find that spillover effects through market-clearing cutoffs substantially reduce the effect of the treatment on allocations of low-income families.

We estimate and perform inference on the effect of information on income inequality by constructing a similar observational dataset on Chilean students. We also find that spillover effects reduce the effect of the intervention on allocations substantially. We combine two datasets from the Ministry of Education for 2018--2020. For the admissions system, we use publicly available data on the centralized admissions process for 2020 for those applying to 9th grade in Chile. These data include the rankings each student submits to the algorithm, their priority, location, and actual assignment. We link this to a demographic survey collected as part of the SIMCE\footnote{Sistema de Medición de la Calidad de la Educación} standardized test system in Chile. For school quality for the 9th grade admissions process, we use the average student math and reading score for the school in 2018 among 10th graders. Students each rank a subset of the 2,500 schools nationwide.\footnote{Of these schools, 57\% reject at least one applicant who ranks them, so $J =1{,}426$ in our application.} 

The treatment we analyze is a proxy for the receipt of information on government school quality. $W_i = 1$ if a parent responds ``Yes'' to the following question: 
\\
{\em 
Do you know the following information about your child's school? Performance category of this school}.\footnote{The survey language (in Spanish) is: ¿Conoce usted la siguiente información del colegio de su hijo(a)? Categoría de desempeño de este colegio. It is the third question in the thirtieth section of the parent survey in the SIMCE dataset.} Of the sample of 114,749 applicants to 9th grade, 53\% have $W_i = 1$. Non-response (24\% of applicants) is coded as $W_i =0$. The observed pre-treatment covariates are location, household size, mother and father education level, whether or not the mother and father are indigenous, and the income of the family. Missing covariates (e.g. due to non-response) are imputed using $k$-nearest neighbors. Table \ref{tab:summary_stats} in Appendix \ref{app:empirical} includes the mean and standard deviation for each of the variables.

We first check that the treatment impacts the rankings that low-income families submit to the allocation mechanism before we examine the effect on allocations. We assume that submitted rankings are not subject to spillover effects. So, we use \texttt{DR-ATE} to estimate the average treatment effect on two outcomes for low-income families in Table \ref{tab:dte}. The first outcome is an indicator that is 1 if the family ranks a top 50\% school first, and the second outcome is the length of the application list that a family submits. Note that the length of the submitted rankings is unrestricted in the Chilean mechanism. The estimated treatment effect on ranking a high-quality school is 2.6 percentage points.\footnote{In the market, 38\% of low-income families with $W_i = 0$ rank a top-50\% school first.} The effect on list length is positive but small. Thus, there is evidence that knowledge about school quality encourages low-income families to apply to better-quality schools. 

\begin{table}[htbp]
  \centering
\begin{tabular}{lcc}
    \toprule
    & \textbf{Top 50\% School Ranked First} & \textbf{Length of Application List} \\
    \midrule
\texttt{DR-ATE}  & 2.6\%  & 0.03 \\
    & (0.45) & (0.02) \\ 
    \bottomrule
  \end{tabular}
    \caption{\texttt{DR-ATE} estimates of the effect of information on applications of low-income families to 9th grade. }
  \label{tab:dte}
\end{table}

Because of capacity constraints, not all families that rank a high-quality school first are admitted to that school. Estimating treatment effects on allocations is more challenging  due to spillovers that occur through the allocation mechanism. Table \ref{tab:results} shows an estimate of treatment effects on low-income families, when the outcome is whether they were accepted to a top 50\% school in Chile. That is, $y(b,x,p) = x_{\text{low}} \sum_{j\in\mathcal G} d_j(b,x,p)/\Pr(X_{\text{low}}=1)$, where $\mathcal G$ defines the set of top 50\% schools. We see that the \texttt{DR-ATE} estimator, which corrects for selection but not spillover effects, estimates a 1.5 percentage point increase in the allocation of low-income families to good-quality schools. The \texttt{LDR-GTE} estimate is 0.7 percentage points, which is much lower. Figure \ref{fig:breakdown} provides a breakdown of the bias of the \texttt{DR-ATE} estimator. In the observed market, the treated families' probability of admission to good schools is higher than in the all-treated market. The control families' probability of admission is lower in the observed market than the all-control market. Estimating $\bar \tau_{\text{GTE}}$ accurately requires estimating the access of treated families in the all-treated market, and control families in the all-control market. 

We briefly discuss a possible source of bias in the \texttt{LDR-GTE} estimate.  There are two possible sources of spillovers from an information treatment; the first is through the mechanism due to capacity constraints, and the second is network-related spillovers. The estimates in Table \ref{tab:results} only account for the first type of spillover. Even if a family does not report receiving school quality information, they may make choices that are correlated with their treated neighbors' choices. If network spillovers may be positive or negative, then a partial identification approach extending Appendix \ref{ap:partial_id} to reports that are a mix of continuous and discrete data is possible, but is left for future work.

\begin{table}
\centering
\begin{tabular}{ | r | r | } 
\hline
\textbf{Estimator}  & \textbf{Treatment Effect Estimate} \textbf{(s.e.)}  \\ 
\hline 
\texttt{LDR-GTE}  &  0.70\% (0.36)   \\ 
\texttt{DR-ATE}  &  1.52\% (0.39)   \\ 
ATE-Bias  &  0.82\% (0.39)   \\ 
\hline
\end{tabular} 
\caption{ Estimates of the treatment effect of informing parents about school quality on allocation of low-income families to good-quality schools. \label{tab:results} } 
\end{table}

\begin{figure} [ht]
\centering
         \includegraphics[width=0.7\textwidth]{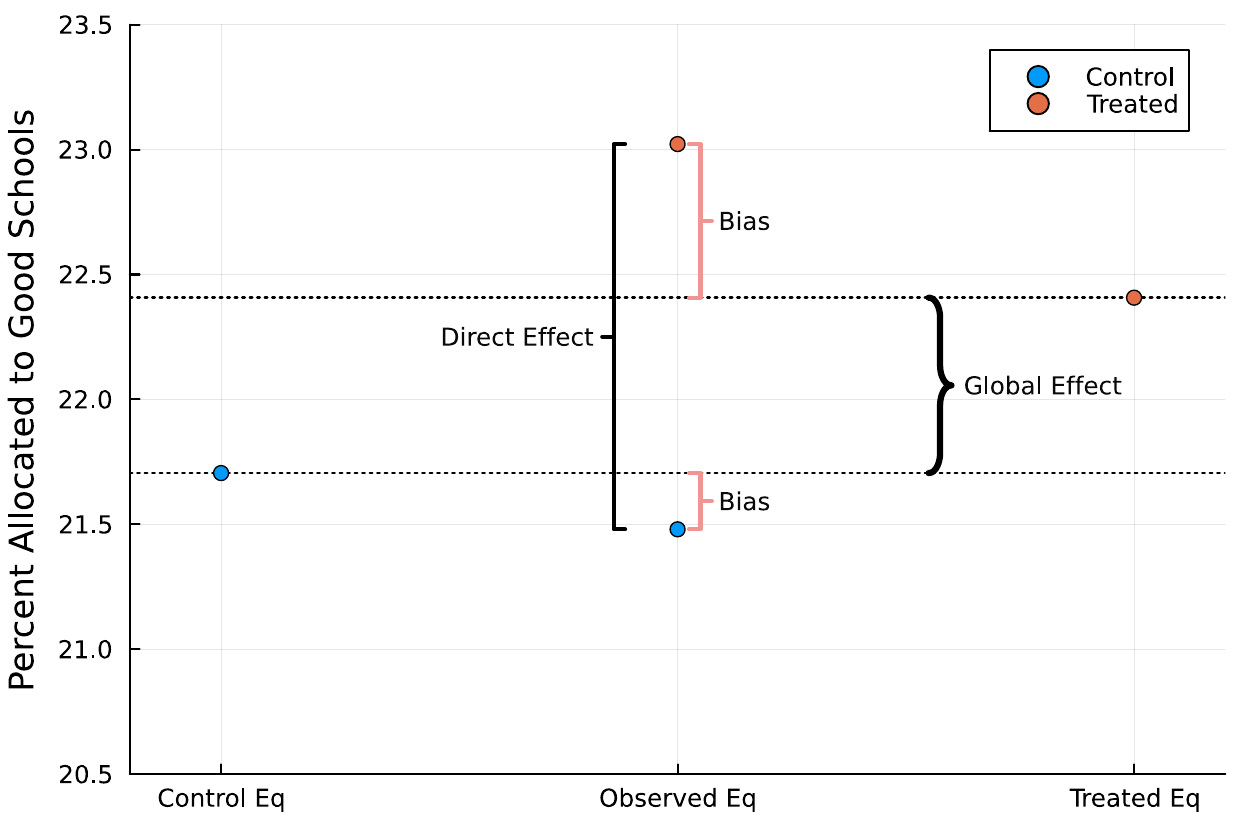}
         \caption{The \texttt{DR-ATE} estimator of the direct effect over-estimates the access of treated families to good-quality schools and under-estimates the access of control families.  } 
         \label{fig:breakdown}
\end{figure}

\begin{figure}[!ht]
 \centering
\includegraphics[width=0.7\textwidth]{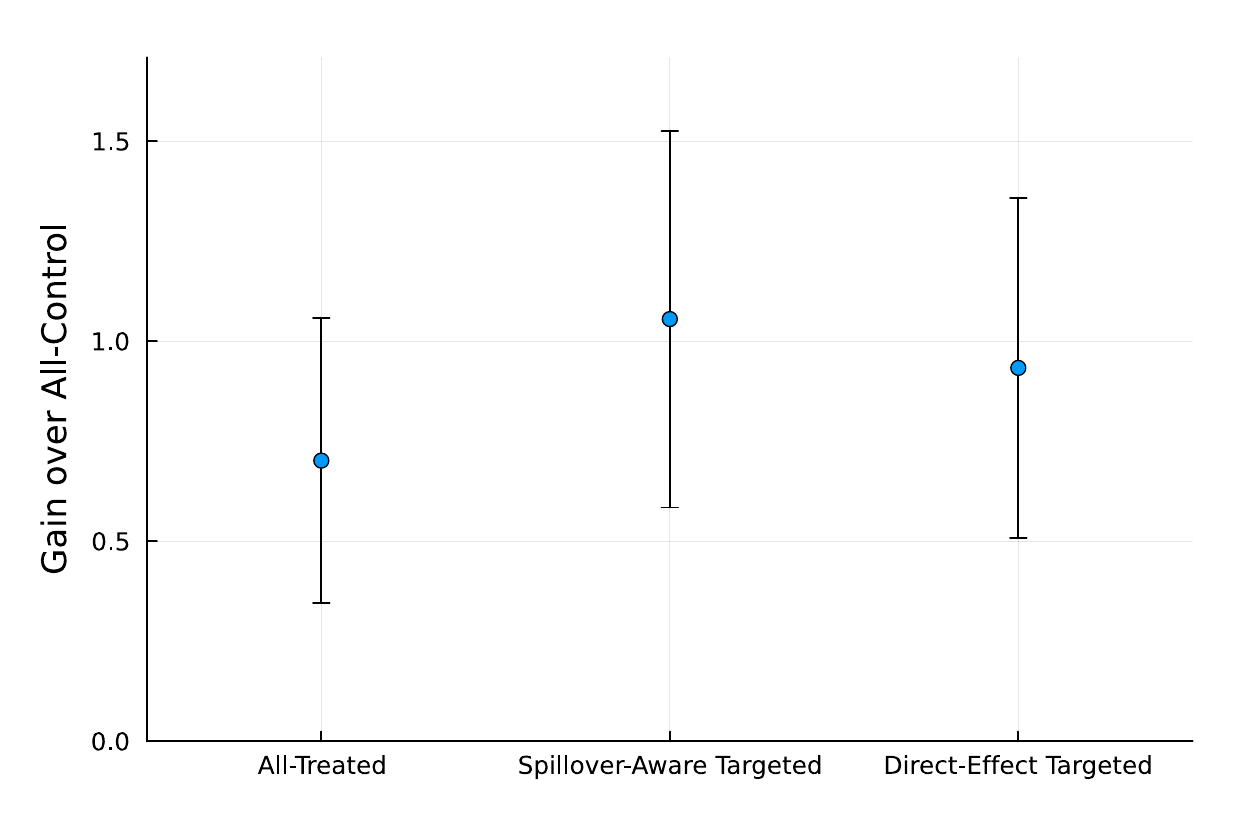}
\caption{The estimated gain, in percentage points, of candidate treatment rules over the all-control baseline. Error bars are standard errors. } 
\label{fig:target}
\end{figure} 

By using a potential outcomes framework to analyze counterfactuals in this setting, heterogeneity in the effect of the treatment on bids is not restricted.  There may be heterogeneity in whether or not individuals respond positively to the information, as well as heterogeneity in how these changes affect congestion in the centralized mechanism. Next, we evaluate treatment rules that treat only a subset of the sample defined by pre-treatment covariates. 

For three candidate treatment rules, Figure \ref{fig:target} plots the gain compared to the counterfactual where nobody is treated. The first rule treats everyone, so the estimated gain is  $\hat \tau_{\text{GTE}}$, estimated using the procedure in Section \ref{sec:est_gte}. The second and third approaches choose a rule using a training sample, which is a random half of the data, and evaluate the value and standard error on the remaining hold-out sample. The spillover-aware targeting rule is the Empirical Welfare Maximization (EWM) rule $\hat \pi \in \arg \max \limits_{\pi \in \Pi} \hat V_n(\pi)$ following Section \ref{sec:est_ewm} and chooses the optimal logistic treatment rule from the class 
  \[
  \Pi = \Big\{ x \mapsto \Lambda\big(\beta_0 + \beta_1  x_{\text{low}} + \beta_2  x_{\text{income}} + \beta_3 
  x_{\text{educ}} + \beta_4  x_{\text{dist}}\big) : \beta \in \mathbb R^5 \Big\}, \qquad \Lambda(z) = \frac{1}{1 + e^{-z}},
  \] 
  where $x_{\text{low}}$ is an indicator that is 1 if the family has low-income priority in the school allocation mechanism, $x_{\text{educ}}$ is the mother's education, $x_{\text{income}}$ is family income, and $x_{\text{dist}}$ is a measure of distance from Santiago. Evaluation of the rule on the hold-out half of the data follows Section \ref{sec:est_gte}, and standard errors are conditional on the estimated rule. 
   The direct-effect targeting rule is a threshold rule. It ignores the cutoff shift and externality adjustment from Theorem \ref{thm:global_target} and treats those with an estimate of $\tau^y(x, p^*_{\text{obs}})$ that is positive, where $p^*_{\text{obs}}$ clears the continuum market of $[F_{\text{obs}}, s^*]$. In a model without spillover effects, an oracle version of this rule is optimal. The estimator of $\tau^y(x, p^*_{\text{obs}})$  is computed by fitting a generalized random forest for the CATE on observed outcomes $Y_i = y(B_i, X_i, P_n(\bm W))$, observed treatments $W_i$ and covariates $X_i$. 

The gain of the spillover-aware targeting rule has a larger point estimate than those of the other two treatment rules, which provides weak evidence for modest heterogeneity in treatment response and the importance of accounting for spillover effects. However, the gain is small and it is not distinguishable at the 5\% level from the gains of the other two treatment rules.

\section{Discussion}

Without some structure, estimating causal effects with general spillovers is infeasible. Under a fully specified and point-identified parametric model of individuals interacting in a market, any counterfactual can be simulated, but the model must be specified correctly. In this paper, we instead use the structure implied by a centralized allocation mechanism, but remain non-parametric about individual choices, which can be difficult to specify correctly. In a continuum-market model, we show that global counterfactuals have a tractable representation as a set of moment conditions. This leads to a computationally simple and doubly-robust estimator for the value of counterfactual policies. With data from the school market in Chile, we show that correcting for congestion effects substantially reduces the estimated effect of information on inequality in school allocations. %

There are a variety of counterfactuals of interest that go beyond the estimands considered in this paper. These include settings with supply-side responses and mechanisms with strategic behavior, where individuals make choices conditional on their expectations of a market equilibrium. For these problems, exploring whether it is possible to derive robust estimators that combine general causal models with economic structure imposed by design is an interesting avenue for future work. 

\subsection*{Data Availability Statement} 

The data that support the findings of this study are available from the Chilean Ministry of Education. Restrictions apply to the availability of these data, which were used with permission for this study. %

\newpage
\singlespacing
\footnotesize
\bibliography{jmp} 
\newpage 
\appendix

\onehalfspacing

\normalsize

\section*{Online Appendix}

\section{Proofs of Main Results}

\subsection{Notation} 
\label{sec:notation}
We first introduce notation which will be used throughout the proofs. The norm $|| \cdot ||$ is the $l_2$-norm. We also define a shorthand notation for empirical averages of doubly-robust scores. For $r \in \{d, y\}$, 
\[ \Gamma^y_{n, \pi}(p; \eta) = \frac 1n \sum_{i=1}^n \Gamma^y_{\pi i}(p; \eta), \]
and $\Gamma^z_{n, \pi}(p; \eta) = \Gamma^d_{n, \pi}(p; \eta) - s^*$. 

 For empirical averages of actual outcomes and allocations rather than doubly-robust scores, we also define: 
\begin{equation*} 
\begin{split}
& Y_{n, \pi}(p) = \frac{1}{n} \sum \limits_{i=1}^n \Big ( W_{\pi, i} y(B_i(1), X_i, p) + ( 1- W_{\pi, i}) y(B_i(0), X_i, p) \Big), \\ 
& Z_{n, \pi}(p) = \frac{1}{n} \sum \limits_{i=1}^n \Big (W_{\pi, i} d(B_i(1),X_i,  p) + ( 1- W_{\pi, i}) d(B_i(0), X_i, p) \Big ) - s^*,
\end{split}
\end{equation*} 
where $W_{\pi, i} \sim \text{Bernoulli}(\pi(X_i))$. The notation $P_n(\bm W)$ is used when $W_i \sim \text{Bernoulli}(e(X_i))$. We use $P_{\pi} = P_n(\bm W_{\pi})$ when $ W_{\pi, i} \sim \text{Bernoulli}(\pi(X_i))$. 

Last, for $r \in \{d, y, z \}$, $r_{\pi}(p; \eta) = \mathbb E_{T}[ \Gamma^r_{n, \pi}(p; \eta)],$  which implies that $y_{\pi}(p; \eta^*_{\pi}) = y_{\pi}(p)$ and $z_{\pi}(p; \eta^*_{\pi}) = z_{\pi}(p)$.

\subsection{Proof of Proposition \ref{prop:id}} 

\begin{proof}
By Assumption \ref{as:id_new}(ii), unconfoundedness implies
\[
    B(w)\mid X
    \ \overset{d}{=}\
    B(w)\mid W=w,X .
\]
By the invariance in Assumption \ref{as:id_new}(i), \(B(w)=B(W) | W=w\), so
\[
    B(w)\mid W=w,X
   \overset{d}{=} 
    B\mid W=w,X .
\]
Overlap in Assumption \ref{as:id_new}(ii) identifies $B\mid W=w,X$ from $F_{\text{obs}}$. Since it also identifies
the marginal distribution of \(X\), it identifies the joint distribution
\(F_w\).
\end{proof}

\subsection{Additional Interactions through Partial Identification} 
\label{ap:partial_id}

In this section, we characterize the identified set for the continuum-market policy value when bids depend on the treatment of others, through equilibrium effects or local spillover effects. Although the text only discusses $\tau^*_{\text{GTE}}$, the discussion and proof here are for general policy $\pi \in \Pi $, so $\tau^*_{\text{GTE}}$ is just a special case. With a known exposure mapping $a(x, p, \pi) \in \mathcal A$, observed bids under the observed policy $e$ are $B = B(W, a(X, p^*_{\text{obs}}, e))$, where $W \sim \mbox{Bernoulli}(e(X))$ and $p^*_{\text{obs}}$ clears $\mathbb E[d(B(W, a(X, p^*_{\text{obs}}, e)), X, p^*_{\text{obs}})] = s^*$. For a policy $\pi \in \Pi$, the counterfactual market-clearing cutoffs are defined as the unique solution to $\mathbb E[d(B(W_{\pi}, a(X, p^*_{\pi}, \pi)), X, p^*_{\pi})] = s^*$, where $W_{\pi} \sim \mbox{Bernoulli}(\pi(X))$. Counterfactual values are 
\begin{equation} \label{eqn:vspartial} 
V^*(\pi) = \mathbb E[y(B_{\pi}, X, p^*_{\pi})], \qquad B_{\pi} = B(W_{\pi}, a(X, p^*_{\pi}, \pi)). 
\end{equation}

Before providing a partial identification result in this setting, we introduce two practical examples of an exposure mapping. 

\paragraph{Market Prices.} When there is a common value component to a bidder's value, then their bid can depend on the market-clearing price. In this case, we can write $a(x, p, \pi) = p$. 

\paragraph{Local Effects.}  Assume that covariate $X$ measures position, such as group membership, or location in space. Let $\omega(x, x') \geq 0$ be a known kernel measuring local connections such that $\mathbb E[\omega(x, X')] = 1$ for all $x \in \mathcal X$, where $X'$ is an independent draw of covariates. The exposure mapping 
\[ a(x, p, \pi) = \mathbb E[\omega(x, X') \pi(X')] \] 
is the expected treated fraction of those connected to buyer $x$. The population limit of the graphon model in \citet{li2022random} has this structure. Other simple examples include group-specific and spatial spillovers. If $x$ is discrete and indicates group membership, then  $\omega(x, x') = \mathbbm{1}(x = x')/Pr(X' = x)$. If $x $ is continuous (e.g. it is in $\mathbb R^2$) and measures physical position, then $\omega(x, x') = \mathbbm{1}( \| x- x' \| \leq r ) / Pr( \| x - X' \| \leq r)$. In this case, bids depend on the expected treated fraction of those within distance $r$.

Now, for partial identification, we replace Assumption \ref{as:id_new} with Assumption \ref{as:partial_id}(ii)--(iv). Assumption \ref{as:unique_new} remains as Assumption \ref{as:partial_id}(i). 

\begin{assumption} \label{as:partial_id}
\begin{enumerate}
\item[(i)] $p^*_{\pi}$ exists and is unique  for all $\pi \in  \Pi$. 
\item[(ii)] For $w \in \{0, 1\}$ and any $(\pi, \pi') \in \Pi \times \Pi$ and any $x \in \mathcal X$,  if $a(x, p^*_{\pi}, \pi) = a(x, p^*_{\pi'}, \pi')$ then almost surely, $B(w, a(x, p^*_{\pi}, \pi)) = B(w, a(x, p^*_{\pi'}, \pi'))$. 
\item[(iii)] Unconfoundedness and overlap hold, so for all $a \in \mathcal A$, $\{ B(1, a), B(0, a) \} \indep W \mid X ,$ and for all $x \in \mathcal X$, $0 < e(x) < 1$. 
\item [(iv)] For a known $\lambda \geq 0$ and $w \in \{0, 1\}$, almost surely, $|B(w, a) - B(w, a')| \leq \lambda |a - a'|$ for all $a, a' \in \mathcal A$.
\end{enumerate}
\end{assumption}

Assumption \ref{as:partial_id}(ii) indicates that $a(x, p, \pi)$ is an exposure mapping. Given a fixed own-treatment $w$, changes in the treatment rule do not impact bids except through $a(\cdot)$. Assumption \ref{as:partial_id}(iii) extends the unconfoundedness assumption to the setting where bids can also depend on a known exposure. Assumption \ref{as:partial_id}(iv) is a new assumption that limits how much bids shift in response to a change in exposure.

The distribution of $(B(w, a(X, p^*_{\text{obs}}, e)), w, X)$ is $F^{\dag}_w$, which is identified from $F_{\text{obs}}$ under Assumption \ref{as:partial_id} and $F^{\dag}_{\pi}$ is the distribution of  $(B^{\dag}_{\pi}, W_{\pi}, X)$, where $B^{\dag}_{\pi} = B(W_{\pi}, a(X, p^*_{\text{obs}}, e))$. But this is not the counterfactual distribution of bids that is needed to identify $V^*(\pi)$ defined in Equation \eqref{eqn:vspartial}. The difference is the exposure; we identify the distribution of $B^{\dag}_{\pi}$, the bids at observed exposure $a(X, p^*_{\text{obs}}, e)$, but we need the distribution of $B_{\pi}$, the bids at counterfactual exposure $a(X, p^*_{\pi}, \pi)$. 

Next, we define a candidate conditional distribution of a counterfactual bid $\tilde B$ given $(B^{\dag}_{\pi}, W_{\pi}, X)$ as $K(\cdot | b, w, x)$. Expectations over $\tilde B$ are over $(B^{\dag}_{\pi}, W_{\pi}, X) \sim F^{\dag}_{\pi}$ and $\tilde B \sim K(\cdot | B^{\dag}_{\pi}, W_{\pi}, X)$. $(p, K)$ are feasible for $\pi$ if a) $\mathbb E[d(\tilde B, X, p)] = s^*$ and b) $K(\cdot | b, w, x)$ satisfies the following support condition derived from Assumption \ref{as:partial_id}(iv) for all $(b, w, x)$: 
\[ K \left ( \Big \{ \tilde B: | \tilde B - b | \leq \lambda | a(x, p, \pi) - a(x, p^*_{\text{obs}}, e) |  \Big \}  \mid b, w, x \right ) = 1. \]

\begin{proposition} \label{prop:partial_id} 
Under Assumption \ref{as:partial_id}, the identified set for $V^*(\pi)$  under observed data $F_{\text{obs}}$ is $\mathcal V_{\pi} = \{ \mathbb E[y(\tilde B, X, p)] : (p, K) \text{ feasible} \}$. 
\end{proposition} 

\begin{proof}

First, we show that for any model meeting Assumption \ref{as:partial_id}(ii)--(iv), $V^*(\pi)$ under that model is in $\mathcal V_{\pi}$. Fix a distribution of potential bids $B(w, a)$ and covariates that satisfy Assumption \ref{as:partial_id}(ii)--(iv) and when $W \sim \mbox{Bernoulli}(e(X))$ generate $F_{\text{obs}}$.  $V^*(\pi) = \mathbb E[ y(B_{\pi}, X, p^*_{\pi})]$ under this model, and $p^*_{\pi}$  satisfies the market-clearing $\mathbb E[d(B_{\pi}, X, p^*_{\pi})] = s^*$. Let $K^*$ be the conditional distribution of $B_{\pi}$ given $(B^{\dagger}_{\pi}, W_{\pi}, X)$. By Assumption \ref{as:partial_id}(iv),  $|B_{\pi} - B^{\dagger}_{\pi}| \leq \lambda|a(X, p^*_{\pi}, \pi) - a(X, p^*_{\text{obs}}, e)|$, so $(p^*_{\pi}, K^*)$ are feasible and  $V^*(\pi) \in \mathcal V_{\pi}$. 

Second, we show that for any $v^* \in \mathcal V_{\pi}$, there is a distribution of potential bids $B(w, a)$ and covariates that satisfy Assumption \ref{as:partial_id}(ii)--(iv) such that when $W \sim \mbox{Bernoulli}(e(X))$, $F_{\text{obs}}$ is generated, and $V^*(\pi) = v^*$ under that model. Fix any feasible $(p, K)$, where $v^* = \mathbb E[y(\tilde B, X, p)]$. Given $X$, for $w \in \{0, 1\}$, draw $B(w, a(X, p^*_{\text{obs}}, e))$ from the distribution of $B \mid W=w, X$ and $\tilde B_w$ from $K(\cdot | B(w, a(X, p^*_{\text{obs}}, e)), w, X)$ and set

\[ B(w, a) = B(w, a(X, p^*_{\text{obs}}, e)) + \mathrm{sign}(\delta_w) \min\big( |\delta_w|, \, \lambda|a - a(X, p^*_{\text{obs}}, e)| \big), \]
where $\delta_w = \tilde B_w - B(w, a(X, p^*_{\text{obs}}, e))$.

Bids depend on $(p, \pi)$ only through $a$ and are independent of $W$ given $X$, so Assumption \ref{as:partial_id}(ii)--(iii) hold. For Assumption \ref{as:partial_id}(iv),
\begin{align*}
|B(w, a) - B(w, a')| &= \big| \min\big( |\delta_w|, \lambda|a - a(X, p^*_{\text{obs}}, e)| \big) - \min\big( |\delta_w|, \lambda|a' - a(X, p^*_{\text{obs}}, e)| \big) \big| \\
&\leq \lambda \, \big| \, |a - a(X, p^*_{\text{obs}}, e)| - |a' - a(X, p^*_{\text{obs}}, e)| \, \big| \\&  \leq \lambda |a - a'|.
\end{align*}
At the observed exposure $a = a(X, p^*_{\text{obs}}, e)$, the minimum is 0, so $B(W, a(X, p^*_{\text{obs}}, e))$ has the distribution of $B \mid W, X$ and we generate $F_{\text{obs}}$.  At $a= a(X, p, \pi)$, then the second term of $B(w, a)$ is $\delta_w$ by the support condition, and $B(w, a(X, p, \pi)) = \tilde B_w$, and bids under $\pi$ are $\tilde B_{W_{\pi}}$, which is a draw of $\tilde B$ by construction. $p$ clears the market under $\pi$ by the definition of feasibility of $(p, K)$, and $V^*(\pi)  = \mathbb E[y(\tilde B, X, p)] =v^*$.

\end{proof}

Proposition \ref{prop:partial_id} implies that sharp bounds for $V^*(\pi)$ are $V^{-}_{\pi} = \inf \limits_{(p, K)  \text{ feasible}} \mathbb E[y(\tilde B, X, p)]$ and $V^{+}_{\pi} = \sup \limits_{(p, K)  \text{ feasible}} \mathbb E[y(\tilde B, X, p)]$. 

We briefly describe how a simple linear programming approach can be used when $b$ and $p$ are scalars and $d(b, x, p)$ is nondecreasing in $b$. Since both the objective and the conditions for feasibility are linear in $K$ at a fixed $p$, with bids and covariates discretized, $V^+_{\pi}$ and $V^-_{\pi}$ are the largest and smallest optimal values of the linear programs solved at each cutoff $p$ with $z^-_{\pi}(p) \leq 0 \leq z^+_{\pi}(p)$, where $z^{\pm}_{\pi}(p) = \mathbb E_{F^{\dagger}_{\pi}}\big[ d\big(B^{\dagger}_{\pi} \pm \lambda|a(X, p, \pi) - a(X, p^*_{\text{obs}}, e)|, X, p\big) \big] - s^*$ defines excess demand when all bids are at their largest and smallest values. This approach can be extended to vector-valued cutoffs under appropriate restrictions on $d(\cdot)$.

\subsection{Proof of Theorem \ref{thm:global_target}} 

The first step is to show that the Fréchet derivative of $V^*(\pi)$ at $\pi$ is the linear functional defined by 
\[ \partial V^*(\pi ) h  = \int h(x) \tau^{q_{\pi}}(x, p^*_{\pi})\,  dF_X(x). \]
where $h : \mathcal X \to [-1, 1]$ and $F_X$ is the distribution of $X \in \mathcal X$. First, we write $V^*(\pi)$ as an integral over $x$: 
\begin{equation*}
\begin{split} 
 V^*(\pi)  &= \mathbb E[\pi(X) y(B(1), X, p^*_{\pi})  + ( 1 -\pi(X)) y(B(0), X, p^*_{\pi})], \\ 
 & = \int \Big ( \mu^y_1(x, p^*_{\pi}) \cdot \pi(x) + \mu^y_0(x, p^*_{\pi}) \cdot ( 1- \pi(x)) \Big)  dF_X(x),  
 \end{split} 
\end{equation*} 
where $\mu^y_w(x, p^*_{\pi}) = \mathbb E[ y(B(w), X, p^*_{\pi}) \mid X =x ]$. We next derive the Fréchet derivative of $V^*(\pi)$ using the product rule, where $\tau^y(x, p^*_{\pi}) = \mu^y_1(x, p^*_{\pi}) - \mu^y_0(x, p^*_{\pi}).$ 
\begin{equation*} 
\begin{split} 
\partial V^*(\pi) h  & =  \int  \tau^y(x, p^*_{\pi}) \cdot h(x) dF_X(x) +  \int \Big ( \partial  \mu^y_1(x, p^*_{\pi}) h  \cdot \pi(x) + \partial \mu^y_0(x, p^*_{\pi}) h  \cdot ( 1- \pi(x)) \Big )dF_X(x) \\ 
& \stackrel{(1)}{=} \int \tau^y(x, p^*_{\pi}) \cdot h(x)  dF_X(x) -  \nu^*_{\pi}  \cdot \int h(x)  \cdot \tau^d(x, p^*_{\pi}) dF_X(x)  \\ 
& = \int h(x) \cdot( \tau^y(x, p^*_{\pi}) - \nu^*_{\pi} \tau^d(x, p^*_{\pi}) )dF_X(x) 
\end{split}
\end{equation*} 

Step (1) is from the chain rule, since 
\[ \int \pi(x) \cdot \partial \mu^y_1(x, p^*_{\pi}) h  +  \partial \mu^y_0(x, p^*_{\pi}) h  \cdot ( 1- \pi(x)) dF_X(x)  = \nabla_p^{\top} y_{\pi}(p^*_{\pi}) \partial p^*(\pi) h  \]  and, by the implicit function theorem, 
\begin{align*} \partial p^*(\pi) h  &  = -  \nabla_p  z_{\pi}(p^*_{\pi})^{-1} \cdot \int h(x) ( \tau^d(x, p^*_{\pi}) )dF_X(x), 
\end{align*} 

where we can swap derivatives and expectation since the derivatives of conditional expectations are bounded by Assumption \ref{as:continuum_regularity}(iii). Since the set of all functions from $\mathcal X$ to $[0, 1]$ is a convex subset of a vector space, Theorem 2 of Chapter 7 of \citet{luenberger1969optimization} indicates that a necessary condition for a local maximum $\pi^*$ is that for all $\pi \in \Pi$, $\partial V^*(\pi^*) (\pi - \pi^*)  \leq 0.$

We can prove by contradiction that the optimal targeting policy must meet the conditions in the theorem. If there is some $\bar \pi$ that is optimal but does not meet the conditions in the theorem, then, one of the following must be true: 

\begin{enumerate} 
\item For $x$ in some set $Q$ that occurs with non-zero probability, $\tau^{q_{\bar \pi}}(x, p^*_{\bar \pi}) < 0$ but $\bar \pi(x) > 0$. But then choose $\pi$ such that $\pi(x) = \bar \pi(x)$ for $x \notin Q$ and $\pi(x) = 0$ for $x \in Q$. We have that  
\begin{equation*}
\begin{split}
 \partial V^*(\bar \pi)(\pi - \bar \pi) = \int_{x \in Q} \tau^{q_{\bar \pi}}(x, p^*_{\bar \pi}) ( 0 - \bar \pi(x)) dF_X(x) > 0,
 \end{split} 
 \end{equation*} 
 which contradicts the optimality of $\bar \pi$.
 \item Or, for $x$ in some set $Q$ that occurs with non-zero probability, $\tau^{q_{\bar \pi}}(x, p^*_{\bar \pi}) > 0$ but $\bar \pi(x) < 1$. Choose $\pi$ such that $\pi(x) = \bar \pi(x)$ for $x \notin Q$ and $\pi(x) = 1$ for $x \in Q$. We have that 
 \begin{equation*}
\begin{split}
 \partial V^*(\bar \pi) (\pi - \bar \pi) = \int_{x \in Q} \tau^{q_{\bar \pi}}(x, p^*_{\bar \pi}) ( 1 - \bar \pi(x)) dF_X(x)> 0,
 \end{split} 
 \end{equation*} 
 which contradicts the optimality of $\bar \pi$.
\end{enumerate}

 \subsection{Proof of Theorem \ref{thm:norm} and Corollary \ref{cor:norm} } 

The proof of Theorem \ref{thm:norm} follows some of the structure and ideas in \citet{kallus2019localized}. For Theorem \ref{thm:norm}, we start with the following expansion: 
\begin{equation*} 
\begin{split} 
\hat V_n(\pi) &  =  \Gamma^y_{n, \pi}(\hat P_{\pi}; \hat \eta_{\pi}) \\ 
&  = \Gamma^y_{n, \pi}(p^*_{\pi}; \eta^*_{\pi}) +  y_{\pi}(\hat P_{\pi}; \hat \eta_{\pi})  - y_{\pi}(p^*_{\pi}; \eta^*_{\pi})  + R_{1n} 
\\ & = \Gamma^y_{n, \pi} (p^*_{\pi}; \eta^*_{\pi}) + y_{\pi}(\hat P_{\pi}; \eta^*_{\pi}) - y_{\pi}(p^*_{\pi}; \eta^*_{\pi}) + R_{1n} + R_{2n} 
\\ & = \Gamma^y_{n, \pi}(p^*_{\pi}; \eta^*_{\pi})  - \nu^*_{\pi} \Gamma^z_{n, \pi}(p^*_{\pi}; \eta^*_{\pi}) + R_{1n} + R_{2n} + R_{3n}.
\end{split} 
\end{equation*} 
To finish the proof, we need to show that each of the remainder terms are $o_p(n^{-1/2})$. 
\begin{equation*} 
\begin{split} 
R_{1n} =  \Gamma^y_{n, \pi}(\hat P_{\pi}; \hat \eta_{\pi})  - \Gamma^y_{n, \pi}(p^*_{\pi}; \eta^*_{\pi}) - y_{\pi}(\hat P_{\pi}; \hat \eta_{\pi}) +  y_{\pi}(p^*_{\pi}; \eta^*_{\pi}) 
\end{split} 
\end{equation*} 
By Lemma \ref{lem:asympe}, $R_{1n} = o_p(n^{-1/2}).$ $R_{2n} = y_{\pi}(\hat P_{\pi}; \hat \eta_{\pi}) - y_{\pi}(\hat P_{\pi}; \eta^*_{\pi})$. By Lemma \ref{lem:unuis}, Assumption \ref{as:nuisance} and the rate for $\hat P_{\pi}$ in Lemma \ref{lem:Pnorm},  $R_{2n} = o_p(n^{-1/2})$. 
For $R_{3n}$, by a Taylor expansion, we have 
\begin{equation*} 
\begin{split} 
y_{\pi}(\hat P_{\pi}; \eta^*_{\pi}) - y_{\pi}(p^*_{\pi}; \eta^*_{\pi})  & =  \nabla_p^{\top}[ y_{\pi}(p^*_{\pi}; \eta^*_{\pi}) ] (\hat P_{\pi} - p^*_{\pi}) + O(\| \hat P_{\pi} - p^*_{\pi} \|^2) \\ 
& \stackrel{(1)}{=} - \nu^*_{\pi} \Gamma^z_{n, \pi}(p^*_{\pi}; \eta^*_{\pi}) + o_p(n^{-1/2}) + O(\| \hat P_{\pi} - p^*_{\pi} \|^2)    \\ 
& \stackrel{(2)}{=}  - \nu^*_{\pi} \Gamma^z_{n, \pi}(p^*_{\pi}; \eta^*_{\pi})  + o_p(n^{-1/2}).  
\end{split} 
\end{equation*}
(1) and (2) are both by  Lemma \ref{lem:Pnorm}. We have now shown that $ \hat V_n(\pi) = \Gamma^y_{n, \pi}(p^*_{\pi}; \eta^*_{\pi})  - \nu^*_{\pi} \Gamma^z_{n, \pi}(p^*_{\pi}; \eta^*_{\pi}) + o_p(n^{-1/2})$. 
We can now apply this expansion to $\hat \tau_{\text{GTE}} = \hat V_n(\pi_1) - \hat V_n(\pi_0)$. 
\[ \hat \tau_{\text{GTE}} = \frac{1}{n} \sum \limits_{i=1}^n  \Gamma^q_{1i}(p^*_{1}; \eta^*_{1}) - \Gamma^q_{0i}(p^*_{0}; \eta^*_{0})  + o_p(n^{-1/2}). \] 
Centering at $\tau^*_{\text{GTE}}$, we have an average of mean-zero and i.i.d. terms with finite variance: 
\[ \hat \tau_{\text{GTE}}  - \tau^*_{\text{GTE}} = \frac{1}{n} \sum \limits_{i=1}^n  \Gamma^q_{1i}(p^*_{1}; \eta^*_{1})  - \mathbb E[\Gamma^q_{1i}(p^*_{1}; \eta^*_{1}) ]  - ( \Gamma^q_{0i}(p^*_{0}; \eta^*_{0})  - \mathbb E[ \Gamma^q_{0i}(p^*_{0}; \eta^*_{0})])     +  o_p(n^{-1/2}). \] 

So, the CLT now applies: 
\[ \sqrt n (\hat \tau_{\text{GTE}}  - \tau^*_{\text{GTE}} )  \rightarrow_D N(0, \sigma^2), \] 

where $\sigma^2 = \mathrm{Var}[\Gamma^q_{1i}(p^*_{1}; \eta^*_{1})  - \Gamma^q_{0i}(p^*_{0}; \eta^*_{0})]$. 

\begin{lemma} \label{lem:Pnorm} \textbf{Central Limit Theorem for $\hat P_{\pi}$: } 
Under the assumptions of Theorem \ref{thm:norm}, for each  $\pi \in \Pi$, 
\[ \sqrt n ( \hat P_{\pi} - p^*_{\pi} )  = - [\nabla_p z_{\pi}(p^*_{\pi})]^{-1} \frac{1}{\sqrt n} \sum \limits_{i=1}^n \Big( \pi(X_i) \Gamma^d_{1i}(p^*_{\pi}; \eta^*_{\pi})  + ( 1 -\pi(X_i)) \Gamma^d_{0i}(p^*_{\pi}; \eta^*_{\pi})  - s^*\Big) + o_p(1). \] 

\end{lemma} 

\begin{proof} 

By Lemma \ref{lem:concp}, $\hat P_{\pi} - p^*_{\pi} = O_p(n^{-1/2})$. We now strengthen this to a central limit theorem that applies for arbitrary $\pi \in \Pi$.  By Lemma \ref{lem:asympe}, 
\[ \Gamma^z_{n, \pi} (\hat P_{\pi}; \hat \eta_{\pi}) - \Gamma^z_{n, \pi}(p^*_{\pi}; \eta^*_{\pi}) = z_{\pi}(\hat P_{\pi}; \hat \eta_{\pi}) - z_{\pi}(p^*_{\pi}; \eta^*_{\pi}) + o_p(n^{-1/2}). \] 

By Lemma \ref{lem:unuis}, 

\[ \Gamma^z_{n, \pi} (\hat P_{\pi}; \hat \eta_{\pi})  - \Gamma^z_{n, \pi}(p^*_{\pi}; \eta^*_{\pi})  = z_{\pi}(\hat P_{\pi};  \eta^*_{\pi}) - z_{\pi}(p^*_{\pi}; \eta^*_{\pi}) + o_p(n^{-1/2}). \] 

Recalling that by Assumption \ref{as:nuisance}, $\Gamma^z_{n, \pi}(\hat P_{\pi}; \hat \eta_{\pi}) = o_p(n^{-1/2})$, we can now use a Taylor expansion: 
\begin{equation*} 
\begin{split}
-\Gamma^z_{n, \pi}(p^*_{\pi}; \eta^*_{\pi})  & = \nabla_p z_{\pi}(p^*_{\pi}) (\hat P_{\pi} - p^*_{\pi}) +  O( \| \hat P_{\pi} - p^*_{\pi}\|^2 ) + o_p(n^{-1/2}) \\ 
-\Gamma^z_{n, \pi}(p^*_{\pi}; \eta^*_{\pi})  & = \nabla_p z_{\pi}(p^*_{\pi}) (\hat P_{\pi} - p^*_{\pi})  + o_p(n^{-1/2}), \\ 
\hat P_{\pi} - p^*_{\pi} & =   - [\nabla_p z_{\pi}(p^*_{\pi})]^{-1} \Gamma^z_{n, \pi}(p^*_{\pi}; \eta^*_{\pi}) +o_p(n^{-1/2})
\end{split} 
\end{equation*} 
where the second line is by Lemma \ref{lem:concp}. This now completes the proof, since $\Gamma^z_{n, \pi}(p^*_{\pi}; \eta^*_{\pi}) = \frac{1}{n} \sum \limits_{i=1}^n  \pi(X_i) \Gamma^d_{1i}(p^*_{\pi}; \eta^*_{\pi}) + ( 1- \pi(X_i)) \Gamma^d_{0i}(p^*_{\pi}; \eta^*_{\pi})- s^*$. 

\end{proof}

\subsection{Proof of Theorem \ref{thm:moment}} 
\label{ap:moment} 
The first part of the Theorem holds by Lemma \ref{lem:fmlem}. For the asymptotically linear expansion,  we next need to prove that for any $\pi \in \Pi$, 
\begin{equation} \label{eq:sample} 
 Y_{n, \pi}(P_{\pi})  =  Y_{n, \pi}(p^*_{\pi}) - \nu^*_{\pi}   Z_{n, \pi}(p^*_{\pi})   + o_p(n^{-1/2}). \\ 
\end{equation} Since $\bar \tau_{\text{GTE}} = Y_{n, 1}(P_{1}) -  Y_{n, 0}(P_{0})$, where the subscript 1  and 0 refers to a treatment rule where everybody and nobody is treated, respectively, then the following argument completes the proof: 
\begin{equation*} 
\begin{split} 
\bar \tau_{\text{GTE}} - \tau^*_{\text{GTE}} &   = Y_{n, 1}(p^*_{1}) - \nu^*_{1}   Z_{n, 1}(p^*_{1})  +  \nu^*_{0}  Z_{n, 0}(p^*_{0})  -  Y_{n, 0}(p^*_{0}) - \tau^*_{\text{GTE}} + o_p(n^{-1/2}),   \\ 
& = \frac{1}{n} \sum \limits_{i=1}^n  q_1(B_i(1), X_i, p^*_1)   - q_0(B_i(0), X_i, p^*_0)  - \tau^*_{\text{GTE}} + o_p(n^{-1/2}), 
\end{split} 
\end{equation*} 

Since  outcomes and net demand are bounded, then the variance of the term in the expansion is finite, and the CLT also applies to this expansion. Thus, to finish the proof, we show  \eqref{eq:sample}. 
\begin{equation*} 
\begin{split} 
Y_{n, \pi}(P_{\pi}) &  =  Y_{n, \pi}(p^*_{\pi}) + y_{\pi}(P_{\pi}) - y_{\pi}(p^*_{\pi}) + o_p(n^{-1/2}),  \\ 
& =   Y_{n, \pi}(p^*_{\pi}) - \nu^*_{\pi}  Z_{n, \pi}(p^*_{\pi})  + o_p(n^{-1/2}).
\end{split}  
\end{equation*} 
The first line is by Lemma \ref{lem:asympt}. The second line is by a combination of a first-order Taylor expansion and  Lemma \ref{lem:fmP}. As a last step for this proof, we prove Lemma \ref{lem:fmP}.

\begin{lemma} \label{lem:fmP} \textbf{Asymptotic Normality of Counterfactual Cutoffs} 
Under the assumptions of Theorem \ref{thm:moment}, then the market-clearing cutoffs under treatment rule $\pi \in \Pi$, which we call $P_{\pi}$, are asymptotically linear: 
\[ \sqrt n (P_{\pi} - p^*_{\pi})  =  - (\nabla_p z_{\pi}(p^*_{\pi}))^{-1}  \frac{1}{\sqrt n} \sum \limits_{i=1}^n (W_{\pi, i} d(B_i(1), X_i, p^*_{\pi}) + ( 1- W_{\pi, i}) d(B_i(0), X_i, p^*_{\pi})  - s^*) +o_p(1). \] 
\end{lemma}

\begin{proof} 
First, by Lemma \ref{lem:concbarp}, we have that $P_{\pi} = p^*_{\pi} + O_p(n^{-1/2})$. To strengthen this to an asymptotic linearity result, we use Theorem 3.3.1 of \citet{vaart1997weak}. 
By Assumption \ref{as:approx_p}, we have the required market-clearing condition, $ Z_{n, \pi}(P_{\pi}) = o_p(n^{-1/2})$. By Lemma \ref{lem:asympt}, we have that $ Z_{n, \pi}(P_{\pi}) - z_{\pi}(P_{\pi}) - Z_{n, \pi}(p^*_{\pi}) + z_{\pi}(p^*_{\pi}) = o_p(n^{-1/2})$. By Assumption \ref{as:continuum_regularity}(ii)--(iii), $ z_{\pi}(p)$ is twice continuously differentiable in $p$ and $\nabla_p z_{\pi}(p)$ is invertible at $p^*_{\pi}$. Since allocations are bounded,  $\mathbb E[(W_{\pi, i} d(B_i(1),  X_i, p) + ( 1- W_{\pi, i}) d(B_i(0), X_i, p) - s^*)^2]$ is bounded. By Theorem 3.3.1 of \citet{vaart1997weak}, verifying these conditions is enough to prove the theorem: 

\[ ( P_{\pi} - p^*_{\pi}) =  - [ \nabla_p z_{\pi}(p^*_{\pi})]^{-1}  Z_{n, \pi}(p^*_{\pi}) + o_p(n^{-1/2}).    \] 

\end{proof}

\subsection{Proof of Proposition  \ref{cor:cons}} 

The two main expansions used here are: 
\begin{align*} 
&  \bar \tau_{\text{GTE}}  = \frac{1}{n}\sum \limits_{i=1}^n [q_1(B_i(1),X_i, p^*_1) - q_0(B_i(0), X_i, p^*_0)]  + o_p(n^{-1/2}),  \\
& \hat \tau_{\text{GTE}} =    \frac{1}{n} \sum \limits_{i=1}^n \left [ \Gamma^q_{1i}(p^*_1; \eta^*_1)  - \Gamma^q_{0i}(p^*_0; \eta^*_0) \right] + o_p(n^{-1/2}), 
\end{align*} 
Notice that $\mathbb E[  \Gamma^q_{1i}(p^*_1; \eta^*_1)  -  \Gamma^q_{0i}(p^*_0; \eta^*_0) \mid X_i, B_i(1), B_i(0)] =  q_1(B_i(1), X_i, p^*_1) - q_0(B_i(0),X_i, p^*_0).$ Combining these, we have that \[ \hat \tau_{\text{GTE}}  - \bar  \tau_{\text{GTE}}    = \frac{1}{n} \sum \limits_{i=1}^n \Gamma^q_{1i}(p^*_1; \eta^*_1)  - \Gamma^q_{0i}(p^*_0; \eta^*_0)  - [q_1(B_i(1), X_i, p^*_1) - q_0(B_i(0), X_i, p^*_0)] +o_p(n^{-1/2}). \] 
Then, using the CLT, 
\begin{equation*}
\begin{split} 
&  \sqrt n \left ( \hat \tau_{\text{GTE}}  - \bar \tau_{\text{GTE}} \right)  \rightarrow_D N(0, \bar \sigma^2), \\ 
&  \sqrt n \left ( \hat \tau_{\text{GTE}} - \tau^*_{\text{GTE}} \right)  \rightarrow_D N(0, \sigma^2), \\ 
\end{split} 
\end{equation*} 
with $\bar \sigma^2 =  \mathbb E[ (\Gamma^q_{1i}(p^*_1; \eta^*_1)  - \Gamma^q_{0i}(p^*_0; \eta^*_0) - q_1(B_i(1), X_i, p^*_1) + q_0(B_i(0), X_i, p^*_0)  )^2 ] $ and $\sigma^2 = \mathbb E[  (\Gamma^q_{1i}(p^*_1; \eta^*_1)  - \Gamma^q_{0i}(p^*_0; \eta^*_0)  -  \tau^*_{\text{GTE}})^2]$. 
Working with $\sigma^2$, let $G_i = q_{1}(B_i(1), X_i, p^*_1)  - q_{0}(B_i(0), X_i, p^*_0) $. 
\begin{equation*} 
\begin{split} 
 \sigma^2  & = \mathbb E[  (\Gamma^q_{1i}(p^*_1; \eta^*_1)  - \Gamma^q_{0i}(p^*_0; \eta^*_0)  -  \tau^*_{\text{GTE}})^2] \\ 
& =  \mathbb E[  (\Gamma^q_{1i}(p^*_1; \eta^*_1)  - \Gamma^q_{0i}(p^*_0; \eta^*_0)  - G_i + G_i - \tau^*_{\text{GTE}}   )^2] \\
& = \mathbb E[ (\Gamma^q_{1i}(p^*_1; \eta^*_1)  - \Gamma^q_{0i}(p^*_0; \eta^*_0) - G_i )^2 ] + \mathbb E[  (G_i - \tau^*_{\text{GTE}})^2]
\\   & \qquad \qquad  + 2 \mathbb E[ (\Gamma^q_{1i}(p^*_1; \eta^*_1)  - \Gamma^q_{0i}(p^*_0; \eta^*_0) - G_i )  (G_i - \tau^*_{\text{GTE}}) ] \\ 
& \stackrel{(1)} {=} \bar \sigma^2 + \mathbb E[  (q_1(B_i(1), X_i, p^*_1) - q_0(B_i(0),  X_i, p^*_0) - \tau^*_{\text{GTE}})^2 ] \\ 
\end{split}
\end{equation*} 
(1) comes from the law of iterated expectations, with the details shown below: 
\begin{equation*} 
\begin{split} 
&  \mathbb E[ (\Gamma^q_{1i}(p^*_1; \eta^*_1)  - \Gamma^q_{0i}(p^*_0; \eta^*_0) - G_i  )  (G_i - \tau^*_{\text{GTE}}) ]  \\ 
& =  \mathbb E[ \mathbb E[(\Gamma^q_{1i}(p^*_1; \eta^*_1)  - \Gamma^q_{0i}(p^*_0; \eta^*_0) - G_i )  ( G_i - \tau^*_{\text{GTE}})  \mid X_i, B_i(1), B_i(0)]]  \\ 
& =  \mathbb E[ \mathbb E[(\Gamma^q_{1i}(p^*_1; \eta^*_1)  - \Gamma^q_{0i}(p^*_0; \eta^*_0) - G_i  ) \mid X_i, B_i(1), B_i(0)]   (G_i  - \tau^*_{\text{GTE}})]  \\ 
& = 0. 
\end{split} 
\end{equation*}
This implies that $\bar \sigma^2 = \sigma^2 -  \mathbb E[  (q_1(B_i(1), X_i, p^*_1) - q_0(B_i(0), X_i, p^*_0) - \tau^*_{\text{GTE}})^2 ] $. Since the second term in the right hand side is weakly positive, $\bar \sigma^2 \leq \sigma^2$, which proves the proposition.

\subsection{Proof of Theorems \ref{thm:regret} and \ref{thm:regret_finite}}

First, we review some notation. Let $\pi \in \Pi$. We have estimated, finite-market and continuum versions of the value function: $   \hat V_n ( \pi) =\Gamma^y_{n, \pi} (\hat P_{\pi}; \hat \eta_{\pi})$, $ \bar V_n(\pi) =    \mathbb E_{\pi} \left [ Y_{n, \pi}(P_{\pi}) \right] $ and  $V^*(\pi) = y_{\pi}(p^*_{\pi}).$

Then, we follow the argument in \citet{kitagawa2018should}. For any $\tilde \pi \in \Pi $,
\begin{equation*} 
\begin{split} 
V^*(\tilde \pi) - V^*(\hat \pi)  & = V^*(\tilde \pi) -   \hat V_n( \hat \pi) + \hat V_n( \hat \pi) - V^*(\hat \pi) \\ 
& \leq V^*(\tilde \pi) - \hat V_n( \tilde \pi) + \hat V_n(\hat \pi) - V^*(\hat \pi)\\
& \leq 2 \sup \limits_{\pi \in \Pi} | V^*(\pi) - \hat V_n(\pi) | 
\end{split} 
\end{equation*} 
For the finite-market regret bound, we have a similar argument. For any $\tilde \pi \in \Pi $, we have that 
\begin{align}
 \bar V_n(\tilde \pi)  - \bar V_n(\hat \pi)  &= \bar V_n(\tilde \pi) - \hat V_n(\hat \pi) + \hat V_n(\hat \pi) - \bar V_n(\hat \pi)  \nonumber \\ 
 & \leq \bar V_n(\tilde \pi) - \hat V_n(\tilde \pi) + \hat V_n(\hat \pi) - \bar V_n(\hat \pi) \nonumber \\
 & \leq 2 \sup \limits_{\pi \in  \Pi  } | \hat V_n(\pi) - \bar V_n(\pi) |  \nonumber \\ 
 & \leq 2 \sup \limits_{\pi \in \Pi} | \hat V_n(\pi) - V^*(\pi) | + 2 \sup \limits_{\pi \in \Pi} | \bar V_n(\pi) - V^*(\pi) |  \label{eq:fmbound}  
\end{align} 
In addition, we have that 
\begin{equation} \label{eqn:a1} 
  \sup \limits_{\pi \in \Pi} | V^*(\pi) - \hat V_n(\pi) |  \leq  \sup \limits_{\pi \in \Pi} |  V^*(\pi) -   \Gamma^y_{n, \pi}(p^*_{\pi}; \eta^*_{\pi})  | +     \sup \limits_{\pi \in \Pi} |  \hat V_n(\pi) - \Gamma^y_{n, \pi}(p^*_{\pi}; \eta^*_{\pi})|   
  \end{equation}   Using notation from Section \ref{sec:notation}, for the first term in \eqref{eqn:a1}, 
\begin{equation*} 
\begin{split} 
\sup \limits_{\pi \in \Pi}  |\Gamma^y_{n, \pi}(p^*_{\pi}; \eta^*_{\pi}) - V^*(\pi) |  & =  \sup \limits_{\pi \in \Pi} | \Gamma^y_{n, \pi}(p^*_{\pi}; \eta^*_{\pi}) - y_{\pi}(p^*_{\pi}; \eta^*_{\pi}) |  \\ 
& \leq \sup \limits_{\pi \in \Pi, p \in \mathcal S} |   \Gamma^y_{n, \pi}(p; \eta^*_{\pi} )- y_{\pi}(p ; \eta^*_{\pi}) | \\ 
& = O_p(n^{-1/2})
\end{split} 
\end{equation*} 
where the conclusion that the term is $O_p(n^{-1/2})$ comes from Lemma \ref{lem:empt}. For the second term in \eqref{eqn:a1}, we use Lemma \ref{lem:rmain2}, so we can now conclude that 
\[   \sup \limits_{\pi \in \Pi } \sqrt n | V^*(\pi) - \hat V_n(\pi) |   = O_p(1).  \] 

This takes care of the regret bound for the continuum market and the first part of \eqref{eq:fmbound}. For the second part of \eqref{eq:fmbound}, to complete the regret bound for the finite market, we use Lemma \ref{lem:fmlem}.

\section{Additional Proofs} 

\subsection{Proof of Semi-Parametric Efficiency Claim in Corollary  \ref{cor:norm}} 
\label{app:eff} 
The proof follows the methodology presented in \citet{bickel1993efficient} and \citet{newey1990semiparametric}. The organization and notation of the proof is similar to other papers that apply this methodology to related estimands, including \citet{hahn1998role} and \citet{hirano2003efficient} for average treatment effects, \citet{firpo2007efficient} for quantile treatment effects, and \citet{chen2021semiparametric} for long-run treatment effects. The presentation and notation is closest to that of \citet{firpo2007efficient}. 

\subsubsection*{Deriving the Score Function}

Under Assumption \ref{as:id_new}, the density of the data $(B(1), B(0), W, X)$ can be factorized as: 
\[ \phi(b(1), b(0), w, x) = f(b(1), b(0) \mid x) e(x)^w ( 1- e(x))^{1- w}f(x)  \] 
Under Assumption \ref{as:id_new}, the density of the observed data $(b, W, X)$ can be factorized as: 
\[ \phi(b, w, x) = [f_1(b\mid x) e(x)]^w [ f_0(b|x) ( 1- e(x))]^{1- w} f(x). \] 

where $f_1(b \mid x) = \int f(b, b_0 \mid x) d b_0$ and $f_0(b\mid x) = \int f(b_1, b \mid x) db_1$. 
We define a regular parametric submodel of the observed data density indexed by $\theta$: 
 \[ \phi(b, w, x; \theta) = [f_1(b\mid x; \theta) e(x; \theta)]^w [ f_0(b \mid x; \theta) ( 1- e(x; \theta))]^{1- w} f(x; \theta) \] 
We can now derive the score of the parametric submodel: 

\[ s(b, w, x; \theta) = w \cdot s_1(b \mid x ; \theta) + ( 1- w) \cdot s_0(b \mid x ; \theta) + \frac{w - e(x) }{ e(x) ( 1- e(x))} e'(x;\theta) + s_x(x; \theta) \] 
where
\begin{align*} 
& s_1(b \mid x; \theta) = \frac{\partial}{\partial \theta} \log f_1(b \mid x ; \theta), \qquad s_0(b \mid x; \theta) = \frac{\partial}{\partial \theta} \log f_0(b \mid x ; \theta), \qquad  e'(x; \theta) = \frac{\partial}{\partial \theta} e(x; \theta), \\&  s_x(x; \theta) = \frac{\partial}{\partial \theta} \log f( x ; \theta). 
\end{align*} 
The tangent space of this model is defined as the set of functions 
\[ g(b, w, x) = w g_1(b \mid x) + ( 1- w) g_0(b\mid x) +  (w - e(x)) g_2(x) + g_3(x) \] 
such that $g_1$ through $g_3$ range through all square integrable functions satisfying 
\begin{align*} 
& \mathbb E[g_1(B \mid X )\mid X= x, W=1]  = 0 \\ 
& \mathbb E[g_0(B \mid X )|  X =x, W=0] = 0 \\ 
 & \mathbb E[g_3(X)] = 0  
\end{align*} 
\subsubsection*{Pathwise Differentiability} 

We derive a Fréchet derivative of $\tau^*_{\text{GTE}} = V^*(\pi_1) - V^*(\pi_0)$, where $V^*(\pi_1) = \mathbb E[y(B(1), X, p^*_1)]$ and $V^*(\pi_0) = \mathbb E[y(B(0), X, p^*_0)]$. We go through the details for $V^*(\pi_1)$, and then state the result for $V^*(\pi_0)$, since the derivation follows the same steps. 

\begin{equation} \label{eqn:pd} 
V_1' =  \nabla_p \mathbb E[y(B(1), X, p^*_1)]^{\top} p_1'  + \frac{\partial}{\partial \theta} \int \int y(b, x, p^*_1) f_1(b \mid x; \theta) f(x;\theta) dbdx 
\end{equation} 

The next step is to derive $p_1'$. By Assumption \ref{as:continuum_regularity}, $p^*_1$ is defined implicitly by $\mathbb E[d(B(1), X, p^*_1) - s^*] = 0$. By the implicit function theorem, we can write
\begin{align*} 
p_1' = - \nabla_p \mathbb E[d(B(1), X, p^*_1) - s^*]^{-1} \frac{\partial}{\partial \theta}  \int \int (d(b, x, p^*_1) - s^*) f_1(b \mid x; \theta) f(x; \theta) db dx. 
\end{align*}

The derivative of the moment conditions, evaluated at $\theta_0$, are as follows, where we write $f(x; \theta_0) = f(x)$ and $f_1(b|x ; \theta_0) = f_1(b\mid x)$. 

\begin{align*} 
& \frac{\partial}{\partial \theta} \int \int y(b, x, p^*_1) f_1(b \mid x; \theta) f(x;\theta) db dx 
 =  \int \int y(b,x, p^*_1) s_1(b \mid x) f_1(b\mid x)  f(x) db dx \\ 
 & \qquad\qquad \qquad \qquad \qquad \qquad \qquad \qquad \qquad+ \int \int y(b,x, p^*_1) f_1(b\mid x)s_x(x)   f(x) db dx, \\ 
 & \frac{\partial}{\partial \theta}  \int \int (d(b, x, p^*_1) - s^*) f_1(b \mid x; \theta) f(x;\theta) db dx 
 = \int \int ( d(b, x, p^*_1) - s^*)  s_1(b \mid x) f_1(b\mid x)  f(x) db dx \\ 
 & \qquad\qquad \qquad \qquad \qquad \qquad \qquad \qquad \qquad+  \int \int ( d(b, x, p^*_1) - s^*) f_1(b\mid x)s_x(x)   f(x) db dx. 
\end{align*} 

Plugging these into the Equation \ref{eqn:pd},
\begin{align*} 
V_1' =  \int \int q^*_1(b,x) s_1(b \mid x) f_1(b\mid x)  f(x) db dx +  \int \int q^*_1(b, x) f_1(b\mid x)s_x(x)   f(x) db dx, 
\end{align*} 

where $q^*_1(b, x) = y(b, x,  p^*_1) - \nu^*_1 (d(b, x, p^*_1)  - s^*)$. Let $q^*_0(b, x) = y(b, x, p^*_0) - \nu^*_0 (d(b,x, p^*_0) - s^*)$. After the same procedure for $V_0'$, we can write 
\begin{align*} \tau' =&   \int \int q^*_1(b, x) s_1(b \mid x) f_1(b\mid x)  f(x) db dx +  \int \int q^*_1(b, x) f_1(b\mid x)s_x(x)   f(x) db dx \\ \qquad \qquad & -  \int \int q^*_0(b, x) s_0(b \mid x) f_0(b\mid x)  f(x) db dx -  \int \int q^*_0(b, x) f_0(b\mid x)s_x(x)   f(x) db dx.  \\ 
 = & \mathbb E[q^*_1(B(1), X) s_1(B(1) \mid X)] + \mathbb E[\mu^q_1(X) s_x(X)] \\ \qquad & - \mathbb E[q^*_0(B(0), X) s_0(B(0) \mid X)] - \mathbb E[\mu^q_0(X) s_x(X)],
 \end{align*} 
where  $\mu^q_w(X) =  \mathbb E[q^*_w(B, X) \mid X, W = w]  $ for $w \in \{0, 1\}$. 
\subsubsection*{Conjectured Efficient Influence Function} 
A function that is in the tangent space is: 
\begin{align*} 
  \psi(B, W, X) = & \mathbb E[q^*_1(B, X) \mid X, W = 1] - \mathbb E[q^*_0(B, X) \mid X, W =0]   - \tau^*_{\text{GTE}}   \\ & +  \frac{ W  (q^*_1(B ,X)   - \mathbb E[q^*_1(B, X) \mid X, W = 1]) }{ e(X) }  \\&   - \frac{ ( 1- W ) (q^*_0(B, X)   - \mathbb E[q^*_0(B, X) \mid X, W = 0]) }{ 1 - e(X) }.
\end{align*} 
We can verify it is in the tangent space. Suppressing the dependence of $q^*_1$ and $q^*_0$ on $X$ to keep notation more concise: 
\begin{enumerate} 
\item $g_1(b \mid x) = \frac{ q^*_1(b) - \mathbb E[q^*_1(B) \mid X = x, W = 1]}{e(x)}.$ For any $x$, \[ \mathbb E[g_1(B \mid X) \mid X =x, W=1 ]= \frac{\mathbb E[q^*_1(B) \mid X = x, W = 1] - \mathbb E[q^*_1(B) \mid X = x, W = 1]}{e(x)} = 0 .\]
\item $g_0(b\mid x) = \frac{ q^*_0(b) - \mathbb E[q^*_0(B) \mid X = x, W = 0]}{1 - e(x)}.$ For any $x$, \[ \mathbb E[g_0(B \mid X) \mid X = x, W = 0] = \frac{ \mathbb E[q^*_0(B) \mid X = x, W = 0] - \mathbb E[q^*_0(B) \mid X = x, W = 0]}{1 - e(x)} = 0. \]
\item $g_2(x) = 0$ 
\item $g_3(x) =  \mathbb E[q^*_1(B) \mid X, W = 1] - \mathbb E[q^*_0(B) \mid X, W =0]   - \tau^*_{\text{GTE}} $
\begin{align*} 
\mathbb E[g_3(X)] & = \mathbb E[\mu^q_1(X)] - \mathbb E[\mu^q_0(X)] - \mathbb E[\mu^q_1(X)] + \mathbb E[\mu^q_0(X)]  = 0. 
\end{align*} 
\end{enumerate} 
Given it is an element of the tangent space, if it is an influence function it is efficient. To verify that is an influence function, we must show that $\mathbb E[ \psi(B, W, X) s(B, W, X) ] = \tau'.$
We can divide $\psi(B, W, X) = \psi_1(B, W, X) - \psi_0(B, W, X)$, where 
\begin{align*}
& \psi_1(B, W, X) =  \mathbb E[q^*_1(B) \mid X, W = 1] -  \mathbb E[\mu^q_1(X)] +  \frac{ W  (q^*_1(B)   - \mathbb E[q^*_1(B) \mid X, W = 1]) }{ e(X) } \\ 
& \psi_0(B, W, X) =  \mathbb E[q^*_0(B) \mid X, W = 0] -  \mathbb E[\mu^q_0(X)]  \\ & \qquad \qquad \qquad \qquad +  \frac{ ( 1- W ) (q^*_0(B)   - \mathbb E[q^*_0(B) \mid X, W = 0]) }{ 1 - e(X) }
\end{align*} 
We work through the details for $\psi_1(\cdot)$, since the process is the same for $\psi_0(\cdot)$. 
\begin{align*} 
  & \mbox{ }  \mathbb E[ \psi_1(B, W, X) s(B, W, X) ]  \\
  & = \mathbb E\Big[ \frac{W (q^*_1(B) - \mu^q_1(X))}{e(X)}\, s_1(B \mid X) \Big] + \mathbb E\big[ (\mu^q_1(X) - \mathbb E[\mu^q_1(X)])\, s_x(X) \big] \\
  & \stackrel{(1)}{=} \mathbb E\big[ (q^*_1(B(1)) - \mu^q_1(X))\, s_1(B(1) \mid X) \big] + \mathbb E[\mu^q_1(X)\, s_x(X)] \\
  & \stackrel{(2)}{=} \mathbb E[q^*_1(B(1))\, s_1(B(1) \mid X)] + \mathbb E[\mu^q_1(X)\, s_x(X)] \\
 & = V_1'
\end{align*} 
(1) follows from unconfoundedness and that $\mathbb E[s_x(X)] = 0$. (2) is because $\mathbb E[s_1(B(1) \mid X) \mid X] = 0$. 

Similarly, we can show that $ \mathbb E[\psi_0(B, W, X) s(B, W, X) ] = V_0'$. We have shown that the function $\psi(B, W, X)$ is an efficient influence function. The semi-parametric efficiency bound is thus: 
\begin{equation*} 
\begin{split} 
\sigma^2  & = \mathbb E[\psi(B, W, X)^2], \\ 
& = \mathbb E[ ( \Gamma^q_{1i}(p^*_1; \eta^*_1) - \Gamma^q_{0i}(p^*_0; \eta^*_0) - \tau^*_{\text{GTE}} )^2].
\end{split} 
\end{equation*} 

\section{Simulation and Empirical Details} 

\subsection{Data-Generating Process for Section \ref{sec:auction}}
\label{app:sim} 

The data-generating process for Section \ref{sec:auction} is as follows, where $\Phi(\cdot)$ is the standard normal CDF: 
  \begin{align*}
  & X_i \sim \mbox{Uniform}(0, 1)^{20}, \qquad W_i \sim \mbox{Bernoulli}\big(\Phi(X_{1i} - 0.5 X_{2i} + 0.5 X_{3i})\big), \qquad B_i(w) \mid X_i \sim F^B_w(\cdot \mid X_i), \\
  & d(b, x, p) = \mathbbm 1(b \geq p), \qquad y(b, x, p) = (b - p)\, \mathbbm 1(b \geq p), \qquad s^* = 1/2,
  \end{align*}

$F^B_1(b \mid x)$ and $F^B_0(b \mid x)$ are varied. In the simulation for Table
  \ref{tab:sim1}, $\log B_i(0) \mid X_i \sim N(0.8 X_{1i} - 0.3 X_{2i} - 0.2 X_{3i},\, 0.3^2)$. In the simulation for Table \ref{tab:sim2}, $B_i(0) \mid X_i$ is
  Weibull with shape $k = 2$ and scale $\lambda(X_i) = \exp(0.8 X_{1i} - 0.3 X_{2i} - 0.2 X_{3i})$. In both simulations, $B_i(1) = 1.5 B_i(0).$

\subsection{Analysis of Coverage and Regret in a Simulated Market for Schools} 
\label{ap:da_sim}

To evaluate finite-sample properties of confidence intervals and the empirical welfare maximization procedure, we construct a simulation of a schools market, where individuals rank schools according to a random utility model, and the treatment affects students' preferences for high-quality schools. Students are assigned via deferred acceptance to one of $J \in \{ 10, 100, 500\}$ schools, where there are seats for 70\% of students. Specifically, each school has fractional capacity $s^*_j = 0.7/J$, so a market with $c \in \{20, 50, 100, 200 \}$ seats per school has $n_{c, J} = J \cdot c/ 0.7$ students. $n_{c, J}$ ranges from 286 to 142,857 in the simulations below.  

School $j$ has academic quality $Q_j$ and negatively correlated non-academic (e.g. ``fun") score $f_j$, with   \[ Q_j = \Phi(\zeta_{1j}), \qquad f_j = \Phi\big(-0.5\,\zeta_{1j} + \sqrt{0.75}\,\zeta_{2j}\big), \]
and $\zeta_{1j}$ and $\zeta_{2j}$ are independent standard normal. For
  each $J$, the set of schools are drawn once. The schools are grouped into $J/10$ submarkets of 10 schools each. Each student indexed by $i$ has covariates $X_{1i}$ and $X_{2i}$, which are independent normal, ability $A_i = \exp(1.2
  X_{1i} + 0.4 X_{2i} + 0.3 u_i)$, a taste for academic rigor $\alpha_i = \Phi(0.3 - 0.8 X_{1i} + 0.8 X_{2i} + 0.5 v_i)$, where $u_i$ and $v_i$ are also independent normal. The student is randomly assigned to one submarket of 10 schools and considers 4 of its 10 schools (where these 4 schools are drawn uniformly without replacement). For $j \in \mathcal C_i$, where $\mathcal C_i$ is the consideration set of four schools, utilities are 
 \[ U_{ij}(0) = \epsilon_{ij}, \qquad U_{ij}(1) = 3\big( \alpha_i A_i Q_j + (1 - \alpha_i) f_j )+ \epsilon_{ij}, \]   where $\epsilon_{ij}$ are i.i.d. Gumbel. The student submits a ranking $R_i(w)$ over the four schools in $\mathcal C_i$ that matches the order of $U_{ij}(w)$. Without treatment, students are uninformed; with treatment, they are informed about the academic and non-academic qualities of the school. The submission to the mechanism under treatment $w$ is $B_i(w) = (R_i(w), S_i)$, where $S_i$ is an i.i.d. standard normal lottery score that determines a student's priority at each school. Treatment assignment in the observed data follows $W_i \sim \text{Bernoulli}(e(X_i))$ with $e(x) = 0.3 + 0.4\,\Phi(0.4 x_1 + 0.4 x_2)$, where $\Phi$ is the standard normal CDF. Under the observed treatment rules, reports to the mechanism are $B_i = B_i(W_i)$. 
 
  The outcome is match value $y(B_i(w), X_i, p) = \sum_{j} d_j(B_i(w), X_i, p)\, m_{ij}$, where $m_{ij} =
  A_i Q_j + \xi_{ij}$, with $\xi_{ij} \sim N(0, 0.25^2)$ represents the expected academic achievement of student $i$ at school $j$ and is increasing in school quality and academic ability. The expected match value $m_{ij}$  is known so that $y(b, x, p)$ is known for the estimators of $\hat \tau_{\text{GTE}}$ and the empirical welfare maximization procedure in the simulation. 

 For each figure, the continuum-market quantities are computed in a market of $10^7$ students. $p^*_w$ is defined as the empirical market-clearing cutoffs when all $10^7$ students receive $W_i=w$, and $\nu^*_w$ is computed from the numerical derivatives of average demand and outcomes at $p^*_w$. From these, we can compute $q_1(B_i(1), X_i, p^*_1)$, $q_0(B_i(0), X_i, p^*_0)$, $\tau^*_{\text{GTE}}$ and $V^*(\pi)$, which  varies slightly by $J$. $\tau^*_{\text{GTE}}$ is $0.105$ for $J=10$, 0.103 for $J=100$ and 0.094 for $J=500$. We also compute the optimal targeting rule, by optimizing $V^*(\pi)$ in the market of $10^6$ students over $\pi \in \Pi$, where $\Pi = \{  \Lambda(\beta_0 + \beta_1 x_1
  + \beta_2 x_2) : \beta \in [-8, 8]^3 \} $ is again a class of logistic rules.

The top panel of  Figure \ref{fig:da_finite} is similar to Figure 4 in \citet{azevedo2016supply}.  It shows how well the finite market approximates the continuum market over 1,000 repetitions of the simulation. In the top panel, it plots the mean (and 5th/95th quantiles) of the difference between $\bar \tau_{\text{GTE}}$ and $\tau^*_{\text{GTE}}$, over repeated draws of $n_{c, J}=J \cdot c/ 0.7$ students from a fixed market with $J$ schools and $c$ seats per school. $\bar \tau_{\text{GTE}}$ takes the difference of average outcomes in a market of size $n_{c, J}$ when a) all $n_{c, J}$ students are treated and b) all $n_{c, J}$ students are in control. The finite market well-approximates the continuum market over a range of $J$; for higher $J$, where $n_{c, J}$ is also larger, the approximation is better. The bottom panel tests how well the normal approximation from Theorem \ref{thm:moment} applies to markets with finite $n$. It computes the coverage of $\tau^*_{\text{GTE}}$ of $\bar \tau_{\text{GTE}} \pm 1.96 \sqrt {\Sigma^*}/\sqrt {n_{c, J}}$, where $\Sigma^* = \mathrm{Var}[q_1(B_i(1), X_i, p^*_1) - q_0(B_i(0), X_i, p^*_0)]$. The ground-truth $\Sigma^*$ is computed in the continuum market  (the $10^7$ market described above) using knowledge of $B_i(1)$ and $B_i(0)$ for each student in simulation, and $\nu^*_w$ and $p^*_w$. There is very close to nominal coverage across a range of choices for $J$ and seats per school. 

Figure \ref{fig:da_estimator} has a similar construction. The top panel shows the mean and quantiles of $\hat \tau_{\text{GTE}} - \bar \tau_{\text{GTE}}$, where the LDR estimator $\hat \tau_{\text{GTE}}$ is computed from a sample of $(B_i, W_i, X_i)$ in a market with $n_{c, J}$ students and $J$ schools as $c$ and $J$ are varied using the algorithm from Section \ref{sec:est_gte} and Equation \ref{eq:ldr_gte}.  Nuisance functions are computed using $k$-nearest neighbors, where $k = \mathrm{round}(\sqrt{n_{\text{train}}}/2)$ and five-fold cross-fitting is used and estimated propensities are clipped to $[0.1, 0.9]$. The bottom panel shows the coverage of confidence intervals based on the normal approximation, with plug-in estimates of the standard errors, as described in Equation \ref{eq:ci_est}. The step size for the finite-differences estimator of the components of $\nu^*_w$ is $\min(0.3, 0.8 c^{-1/3})$. The coverage of $\tau^*_{\text{GTE}}$ is close to nominal, and is slightly conservative for $\bar \tau_{\text{GTE}}$ as predicted by Proposition \ref{cor:cons}. 
 
In this simulation, there is both limited supply and significant heterogeneity in treatment response. The planner can improve outcomes by targeting the treatment to those with high ability and a preference for academic quality versus those who respond less to academic quality. Figure \ref{fig:da_ewm} illustrates the performance of empirical welfare maximization in finite markets. The EWM procedure uses the Nelder--Mead method to search for the best treatment rule in $\Pi$ based on the estimated value $\hat V_n(\pi)$,  choosing the best treatment rule found from one of multiple starting points. Then, the estimated rule is compared to three oracle rules. For each simulation, the reported gain over all-treated is computed on the oracle $\bar V_n(\pi)$, computed numerically over 100 draws of $\bm W$. For a single finite market, $\bar V_n(\pi)$ is computed by calculating average outcomes. The bars show the 5th and 95th percentiles of the gain over all-treated over 100 simulations.   
 \begin{itemize}
 \item Finite-Market Logistic: The rule that maximizes $\bar V_n(\pi)$ over $\Pi$, where this rule is also computed via a Nelder--Mead search over $\Pi$ with $\bar V_n(\pi)$ as the objective. The gain is 0.02--0.04 in this simulation. 
 \item Continuum-Market Logistic: The rule that maximizes $V^*(\pi)$ over $\Pi$, computed using a Nelder--Mead search over $\Pi$.
 \item Direct-Effect Threshold: This rule ignores spillover effects and treats those who have an (oracle) positive response to the treatment at the continuum observed market. Specifically, $\pi(x) = \mathbbm{1}( \tau^y(x, p^*_{\text{obs}}) > 0)$, where $p^*_{\text{obs}}$ clears the market $[F_{\text{obs}}, s^*]$ and $F_{\text{obs}}$ is the distribution of $(B_i, W_i, X_i)$ above when $W_i \sim \mbox{Bernoulli}(e(X_i))$. 
 \end{itemize} 
 
At low sample sizes, there is too much noise in $\hat V_n(\pi)$, and empirical welfare maximization performs poorly. As the sample size grows, however, the empirical welfare maximization procedure chooses an optimal treatment rule that is close to finite-market optimal; the mean regret shrinks at a $1/\sqrt n$ rate. Even for modest sample sizes, the finite-market optimal rules and continuum-market optimal rules have similar value. Ignoring spillover effects with the direct-effect threshold rule leads to a sub-optimal rule, which performs much worse than the EWM rule for most sample sizes.

\begin{figure}[!ht]
\centering
\includegraphics[width=\textwidth]{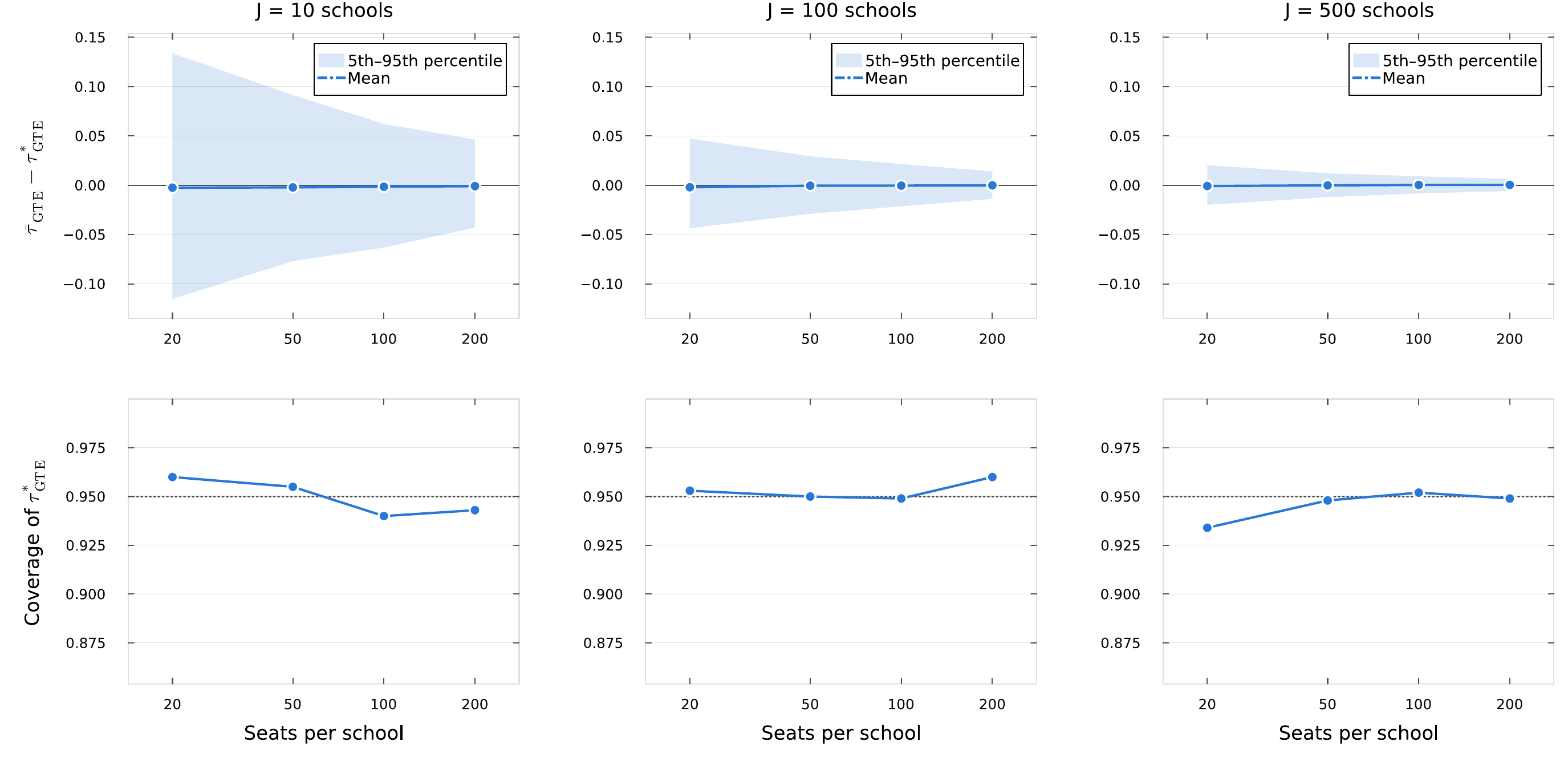}
\caption{Comparison of the finite-market $\bar \tau_{\text{GTE}}$ to the continuum $\tau^*_{\text{GTE}}$. The top panel gives the mean and 5th--95th percentiles of $\bar \tau_{\text{GTE}} - \tau^*_{\text{GTE}}$. The bottom panel gives the coverage of confidence intervals $\bar \tau_{\text{GTE}} \pm 1.96 \sqrt {\Sigma^*}/ \sqrt n$ for $\tau^*_{\text{GTE}}$, where $\Sigma^* = \mathrm{Var}[q_1(B_i(1), X_i, p^*_1) - q_0(B_i(0), X_i, p^*_0)]$. The statistics are computed from 1,000 replications of each simulation. \label{fig:da_finite}}
\end{figure}

  \begin{figure}[!ht]
  \centering
  \includegraphics[width=\textwidth]{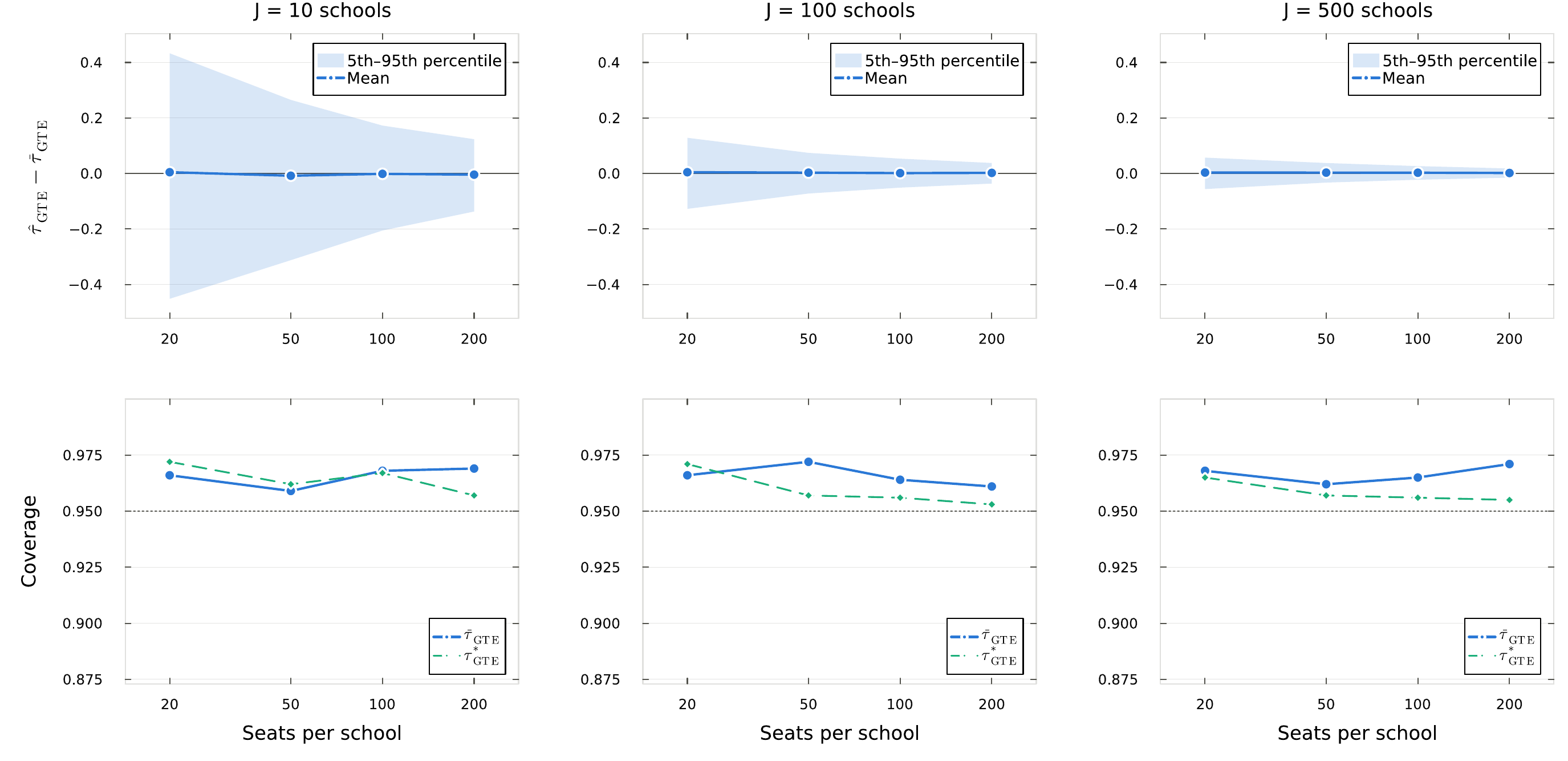}
  \caption{Comparison of the LDR estimator $\hat \tau_{\text{GTE}}$ to the finite-market $\tau_{\text{GTE}}$. The top panel gives the mean and 5th--95th percentiles of $\hat \tau_{\text{GTE}} - \bar \tau_{\text{GTE}}$. The bottom panel gives the coverage of $\bar \tau_{\text{GTE}}$ and $\tau^*_{\text{GTE}}$ of $\hat \tau_{\text{GTE}} \pm 1.96 \hat \sigma/ \sqrt n$. These statistics are computed from 1,000 replications of each simulation. 
  \label{fig:da_estimator}}
  \end{figure}

  \begin{figure}[!ht]
  \centering
  \includegraphics[width=\textwidth]{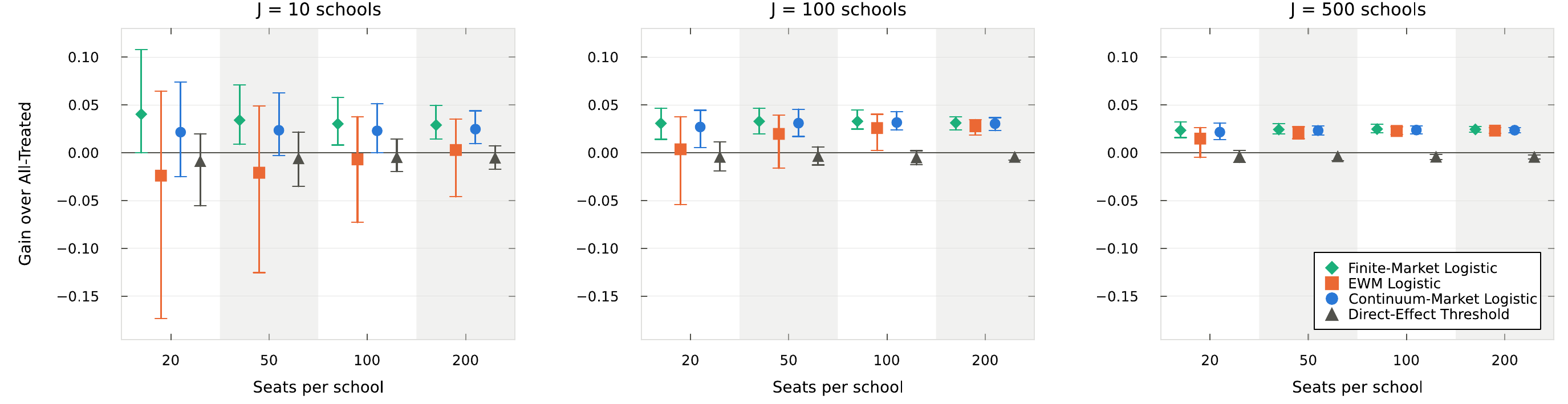}
  \caption{Gain over $\pi_1$ (All-Treated), evaluated using the oracle value $\bar V_n(\pi)$, of four candidate treatment rules. \label{fig:da_ewm}}
  \end{figure}

\clearpage
\subsection{Additional Empirical Figures}
\label{app:empirical}
\begin{table}[!htbp]
  \centering
 \begin{tabular}{lcc}
    \toprule
    \textbf{Variable} & \textbf{Treated} & \textbf{Control} \\
    \midrule
    income & 4.22 & 4.77 \\
     & (3.32) & (3.82) \\
    ma\_educ & 11.01 & 11.46 \\
     & (3.14) & (3.14) \\
    pa\_educ & 11.00 & 11.45 \\
     & (3.45) & (3.45) \\
    ma\_indig & 0.18 & 0.17 \\
     & (0.38) & (0.37) \\
    pa\_indig & 0.15 & 0.14 \\
     & (0.35) & (0.35) \\
    hhsize & 2.45 & 2.46 \\
     & (1.29) & (1.27) \\
    latitude & -34.36 & -34.15 \\
     & (4.90) & (5.04) \\
    longitude & -71.47 & -71.37 \\
     & (1.02) & (1.03) \\
    \bottomrule
  \end{tabular}
\caption{Summary Statistics for $n=114{,}749$ applicants to 9th grade in 2020. $W_i = 1$ indicates a parent reported they were aware of the performance category of the 8th grade school of their child. Income is in \$100,000 pesos, and education is in years. }
      \label{tab:summary_stats}
\end{table}

\clearpage
\section{Extensions}

 \subsection{Verifying Regularity Conditions}
 \label{app:upa} 
 \begin{proposition}\label{prop:upa}   Assume that $ 0 < s^* <1 $ and that market participants are bidding in a uniform-price auction. We impose the following assumptions on the distribution of bids. 
\begin{itemize} 
\item For $w \in \{0, 1\}$, $B_i(w) \in [\underline v, \bar v] \subset \mathbb R$ where $\underline v$ and $\bar v$ are finite and strictly positive.
\item For all $x \in \mathcal X$, the conditional CDF of the bid distribution, $F^B_w(b \mid x)$ with density $f^B_w(b \mid x)$, is twice continuously differentiable in $b$ on $[\underline v, \bar v]$,  for $w \in \{0, 1\}$,  with the absolute value of the first and second derivatives uniformly bounded by finite constant $c'\geq 1$.   In addition, the first derivative is bounded below by strictly positive constant $c_3$. 
\item The outcome is bidder surplus, so $y(b, x, p) = (b - p) d(b, x, p)$. 
\end{itemize} 
Then, Assumptions \ref{as:continuum_regularity}, \ref{as:regularo} and \ref{as:approx_p} hold. 
\end{proposition}

The argument in Proposition \ref{prop:upa} can also be extended to deferred acceptance; see \citet{agarwal2018demand} for verification of many of the required conditions. 

\begin{proof}
For the remainder of this section, we drop the dependence of $d(\cdot)$ on $X_i$, since allocations do not depend on $X_i$ in this example. We start by verifying Assumption \ref{as:continuum_regularity}(i) with $\mathcal S = [\underline v, \bar v]$. The unconditional distribution of $B_i(W_{\pi, i})$ is  
\[ F^B_{\pi}(b) = \int \big[ \pi(x) F^B_1(b \mid x) + ( 1- \pi(x)) F^B_0(b \mid x) \big] \, dF_X(x). \] 
We can define $z_{\pi}(p) = 1 - F^B_{\pi}(p) - s^*$, and since $F^B_{\pi}$ is strictly increasing on $\mathcal S$, $p^*_{\pi} =  (F^B_{\pi})^{-1}(1 - s^*)$ is unique. Since the bid density is at most $c'$, we can use a Taylor expansion to argue that $1 - s^* = F^B_{\pi}(p^*_{\pi}) - F^B_{\pi}(\underline v) \leq c'(p^*_{\pi} -
  \underline v)$ and $s^* \leq c'(\bar v - p^*_{\pi})$. Choosing $c_1 = \min\{s^*, 1 - s^*\}/c'$, $\mathcal S$ always contains a ball of radius $c_1$ centered at $p^*_{\pi}$. For $p \in \mathcal S$, the mean-value
  theorem gives $z_{\pi}(p) = z_{\pi}(p) - z_{\pi}(p^*_{\pi}) = z'_{\pi}(\tilde p)(p - p^*_{\pi})$ for some $\tilde p \in
  \mathcal S$, and $|z'_{\pi}(\tilde p)| \geq c_3$, so $|z_{\pi}(p)| \geq c_3 |p - p^*_{\pi}|$. This means if $|p - p^*_{\pi}|
  \geq c_3 / 2c'$, then $|z_{\pi}(p)| \geq c_3^2 / 2c'$, which is a strictly positive lower bound.
  
  For Assumption \ref{as:continuum_regularity}(ii), we have that
\[ \nabla_p \mathbb E[\pi(X_i) \mu^d_1(X_i, p^*_{\pi} )+ ( 1- \pi(X_i))  \mu^d_0(X_i, p^*_{\pi})] = -  \mathbb E[\pi(X_i) f^B_1(p^*_{\pi} \mid X_i) + (1 - \pi(X_i) )f^B_0(p^*_{\pi} \mid X_i)]. \] 

We can exchange the derivative and expectation by the dominated convergence theorem. To evaluate the derivative, notice that $\mu^d_1(x, p) = \Pr(B_i(1) \geq p \mid X_i =x) = 1 - F^B_1(p \mid x)$.  The magnitude of the RHS is bounded between $c_3$ and $c'$, since $f^B_w(p \mid x)$ is uniformly bounded between $c_3$ and $c'$ and $0 \leq \pi(X_i) \leq 1$. 

For Assumption \ref{as:continuum_regularity}(iii), let $p \in \mathcal S$. $\nabla_p \mu^d_w(x, p) = \nabla_p \Pr(B_i(w) \geq p \mid X_i = x) = -f^B_w(p \mid x)$. The conditional CDF is twice continuously differentiable in $p$, with first and second derivatives bounded by $c'$. For outcomes, $\nabla_p \mu^y_w(x, p) =  \nabla_p \mathbb E[(B_i(w) - p)\mathbbm 1(B_i(w) \geq p)  \mid X_i= x] = \nabla_ p \int_{p}^{\bar v} b \, dF^B_w(b \mid x)  - \nabla_p [ p \, ( 1 - F^B_w(p \mid x)) ]$.  By Leibniz's rule, and that $p$ is bounded, this is also twice continuously differentiable in $p$, with bounded first and second derivatives, by the properties of the conditional distribution of $B_i(w)$.

For Assumption \ref{as:regularo}(i), the class of $d(B_i(w), p)$ indexed by $p \in \mathcal S$ is a VC class of functions (the class of indicator functions is a VC class), so the covering number has the polynomial bound required. The class of linear functions $(B_i(w) - p)$ indexed by $p \in \mathcal S$ also has  a polynomial bound, since the covering number of that class equals the covering number of $\mathcal S$, which is compact. Then, $y(B_i(w), p)$ is a Lipschitz combination of functions each with covering numbers that have a polynomial bound, so by Lemma \ref{lem:comp},  the assumption holds. 

For Assumption \ref{as:regularo}(ii), the weak continuity assumption, we have the following argument, where $F^B_w(\cdot)$ is the CDF of $B_i(w)$. 
\begin{equation*} 
\begin{split} 
\mathbb E[ (d(B_i(w), p) - d(B_i(w) , p') )^2]  & = \mathbb E[ (\mathbbm{1}(B_i(w) \geq p) - \mathbbm{1}(B_i(w) \geq p'))^2] \\ 
& =  (F^B_w(p') - F^B_w(p)) \mathbbm{1}(p' > p)  + (F^B_w(p) - F^B_w(p') ) \mathbbm{1}{(p' \leq p)} \\ 
 & \leq  c' \| p - p' \| 
\end{split} 
\end{equation*} 
where the last step is because the CDF of $B_i(1)$ and $B_i(0)$ is differentiable with bounded first derivatives. For outcomes, 
\begin{equation*} 
\begin{split} 
  \mathbb E[ (y(B_i(w), p) - y(B_i(w) , p') )^2]  & = \mathbb E[ \big( (B_i(w) - p) (d(B_i(w), p) - d(B_i(w), p')) + (p' - p)\, d(B_i(w), p') \big)^2] \\
    & \leq 2 (\bar v - \underline v)^2\, \mathbb E[ (d(B_i(w), p) - d(B_i(w), p'))^2] + 2 |p - p'|^2 \\
    & \leq 2 \big( (\bar v - \underline v)^2 c' + \bar v - \underline v \big) |p - p'|
\end{split} 
\end{equation*} 
where we use the result for $d(\cdot)$ in the last step and that $|p - p'| \leq \bar v - \underline v $. 

We can finish by verifying the finite-market-clearing assumption in Assumption \ref{as:approx_p}. Fix $\bm w \in \{0, 1\}^n$.  A typical market-clearing rule for a uniform-price auction sets $P_n(\bm w)$ to the $ \lfloor n s^* \rfloor$-th highest bid.  Let
  \[ Z_n(\bm w)  = \frac{1}{n} \sum_{i=1}^n \mathbbm 1\big(B_i(w_i) \geq
  P_n(\bm w)\big) - s^* \]
  denote excess demand at this cutoff. Since the bids  $B_1(w_1), \ldots, B_n(w_n)$ have a continuous distribution, with probability 1 they are distinct, and exactly $ \lfloor n s^* \rfloor$ bids satisfy $B_i(w_i)
  \geq P_n(\bm w)$, and
  \[ |Z_n(\bm w)|  = \frac{n s^* - \lfloor n s^* \rfloor}{n} < \frac{1}{n}. \]
  Since this holds with probability 1 for every $\bm w \in \{0, 1\}^n$, Assumption \ref{as:approx_p} holds with $a_n = 1/n$.

\end{proof}

\subsection{Using IV for Identification and Estimation} 
\label{app:iv} 

This section provides a brief discussion of how a restricted version of the Global Treatment Effect can be estimated when unconfoundedness does not hold, but there is a binary instrumental variable that affects take-up of a binary treatment.  We drop the dependence of $y(\cdot)$ and $d(\cdot)$ on $x$ for this section to keep notation concise.  In an IV setting, we have potential treatments $W(h)$ that depend on an instrument $h \in \{0, 1\}$. With spillover effects, there are a variety of counterfactuals that can be defined. One relevant counterfactual when there may be control over the instrument, but not the treatment directly, is the intent-to-treat effect.  This is the effect on continuum-market outcomes when all individuals receive the instrument, compared to a setting where no agents receive the instrument. In the setting where all buyers receive instrument $h$, the market-clearing cutoff $p^*_{H=h}$ solves $\mathbb
 E[d(B(W(h)), X, p^*_{H=h})] = s^*$, and the continuum-market global intent-to-treat effect (GITT) is
 \[ \tau^*_{\text{GITT}} = \mathbb E[y(B(W(1)), X, p^*_{H=1})] - \mathbb E[y(B(W(0)), X, p^*_{H=0})]. \]
 
Assume that a version of  Assumptions \ref{as:unique_new} and \ref{as:id_new} hold, in that $p^*_{H=h}$ is unique, $B(W(h))$ does not depend on the treatment rule, and $\{B(W(1)), B(W(0))\} \indep H \mid X$ with $0 < \Pr(H = 1 \mid X) < 1$. Then $\tau^*_{\text{GITT}}$ is identified, and it can be estimated using Section \ref{sec:est_gte} with $H_i$ replacing $W_i$. 
 
 Under monotonicity, $W(1) \geq W(0)$,  and
 \begin{equation*}
 \begin{split}
 \tau^*_{\text{GITT}} = {} & \mathbb E\big[ \mathbbm 1(W(1) > W(0)) \big( y(B(1), X, p^*_{H=1}) - y(B(0), X, p^*_{H=0})
 \big) \big] \\
 & + \sum_{w \in \{0, 1\}} \mathbb E\big[ \mathbbm 1(W(1) = W(0) = w) \big( y(B(w), X, p^*_{H=1}) - y(B(w), X, p^*_{H=0})
 \big) \big].
 \end{split}
 \end{equation*}
 The first term is the effect on compliers, who change their reports. The second is the spillover on always-takers and
 never-takers: their reports do not change, but they face the cutoff $p^*_{H=1}$ instead of $p^*_{H=0}$.

The complier-only \citep{imbens1994identification} quantity in this decomposition can be separately identified and estimated: see Appendix A of  \citet{kallus2019localized}. 

\subsection{Connecting to Research Design Meets Market Design} 
\label{ap:angrist}

    In this paper, we identify and estimate the effect of an individual-level treatment on allocations in a centralized market.
  \citet{abdulkadirouglu2017research} consider a different type of causal effect.  They are interested in the effect of
  allocations (e.g. attending a charter school) on some stochastic outcome, like test scores or future income. We briefly
  discuss how these two approaches can be combined, under an additional assumption that the treatment only affects outcomes
  through the school a student is assigned to, when a student's priority at a school is determined by a lottery score with a known distribution. 

  Let $\{ Y(j) : j \in \{ 1, \ldots, J \} \}$ define general potential outcomes, where the treatment affects
  outcomes only through an individual's allocation, and the outcome with no assignment is normalized to zero.

  Write the priority scores as $S = S^0 + U$, where $S^0 \in \mathbb Z^J$ gives a student's priority group at each school and
  $U \in [0, 1)^J$ is the lottery score, drawn from a known distribution $F_U$ independently of $(R(1), R(0), S^0, W, X)$ and
  of the potential outcomes. For a ranking $r$, priority  $s^0$, covariates $x$ and cutoffs $p$, let
  \[ e^d_j(r, s^0, x, p) = \int d_j\big((r, s^0 + u), x, p\big)\, dF_U(u) \]
  be the probability of assignment to school $j$. The allocation-specific propensity score under the
  observed treatment rule is the $J$-length vector $e^d = (e^d_1, \ldots, e^d_J)$ with $e^d_j = e^d_j(R(W), S^0, X, p^*_{\text{obs}})$. 

 We must have that, for each school $j$ whose assignment share changes with the treatment, the lottery mechanism is such that for some non-negligible group, $0 < e^d_j < 1$. By comparing students with similar allocation-specific propensity scores who have
  different allocations, then  \citet{abdulkadirouglu2017research} identify the effect of allocations. In our notation, this
  type of causal effect is $\mathbb E[Y(j) - Y(k) \mid e^d]$. Using the approach in
  \citet{abdulkadirouglu2017research}, we can use the observed data to identify $\mathbb E[Y(j) \mid 0 < e^d_j
  < 1]$. The approach in our paper, on the other hand, identifies the change in
  allocations from a counterfactual treatment allocation: the algorithm in Section \ref{sec:est_gte} is used to estimate
  $\mathbb E[d_j(B(w), X, p^*_w)]$. 

  Thus, the two approaches can be linked to identify $\tau^*_{\text{GTE}}$ under the assumption that the treatment only impacts
  outcomes through attendance at a given school, and that the students randomized by the lottery are representative of all
  students.The second assumption is strong. It requires that a) students randomized into school $j$ in the observed economy have the same mean outcome at
  school $j$ as all students,
  \[ \mathbb E[Y(j) \mid 0 < e^d_j < 1] = \mathbb E[Y(j)] \]
and b) a student's probability of assignment to school $j$ in the all-treated and all-control economies is uncorrelated with
  their outcome at school $j$,
  \[ \mathbb E\big[ e^d_j(R(w), S^0, X, p^*_w)\, Y(j) \big] = \mathbb E\big[ e^d_j(R(w), S^0, X, p^*_w) \big]\, \mathbb
  E[Y(j)], \qquad w \in \{0, 1\}. \]
  Then, defining $\tau^*_{\text{GTE}}$ for this section by replacing the known outcome $y(b, x, p)$ with the stochastic potential outcome $Y(j)$, 
  \begin{equation*}
  \begin{split}
  \tau^*_{\text{GTE}} & = \mathbb E\Big[ \sum_{j=1}^J \big( d_j(B(1), X, p^*_1) - d_j(B(0), X, p^*_0) \big)\, Y(j) \Big] \\
  & = \sum_{j=1}^J \mathbb E\Big[ \big( e^d_j(R(1), S^0, X, p^*_1) - e^d_j(R(0), S^0, X, p^*_0) \big)\, Y(j) \Big] \\
  & = \sum_{j=1}^J \mathbb E[Y(j) \mid 0 < e^d_j < 1] \Big( \mathbb E[d_j(B(1), X, p^*_1)] - \mathbb E[d_j(B(0), X, p^*_0)]
  \Big).
  \end{split}
  \end{equation*}
  The third  line follows from our assumption that the students randomized by the lottery are representative of all students, recognizing that $\mathbb
  E[e^d_j(R(w), S^0, X, p^*_w)] = \mathbb E[d_j(B(w), X, p^*_w)]$. Accommodating settings where outcomes depend on the treatment in more complex ways is a subject for future work.

\section{Concentration Results}  \label{sec:conc} 

\begin{lemma} \label{lem:asympe} \textbf{Asymptotic Equicontinuity with Estimated Nuisances} 

Under the assumptions of Theorem \ref{thm:norm}, we have the following asymptotic equicontinuity result: 
\begin{equation*} 
\begin{split} 
\Gamma^y_{n, \pi}(\hat P_{\pi}; \hat \eta_{\pi})  - \Gamma^y_{n, \pi}(p^*_{\pi}; \eta^*_{\pi}) - y_{\pi}(\hat P_{\pi}; \hat \eta_{\pi}) +  y_{\pi}(p^*_{\pi}; \eta^*_{\pi})  & = o_p(n^{-1/2}) ,  \\ 
\Gamma^z_{n, \pi}(\hat P_{\pi}; \hat \eta_{\pi})  - \Gamma^z_{n, \pi}(p^*_{\pi}; \eta^*_{\pi}) - z_{\pi}(\hat P_{\pi}; \hat \eta_{\pi}) +  z_{\pi}(p^*_{\pi}; \eta^*_{\pi})  & = o_p(n^{-1/2}) . 
\end{split} 
\end{equation*} 
\end{lemma}

\begin{proof} 
We prove this for $\Gamma^y_{n, \pi}(\cdot)$ and the proof is the same for $\Gamma^z_{n, \pi}(\cdot)$. As in Section \ref{sec:est_gte}, we suppress the fold notation. The number of folds is fixed, so we prove the statement for the average over one evaluation fold. The nuisances $\hat \eta_{\pi}$ are fitted on the other folds, so we can treat $\hat \eta_{\pi}$ as fixed. In this section, we let $R_n = \Gamma^y_{n, \pi}(\hat P_{\pi}; \hat \eta_{\pi})  - \Gamma^y_{n, \pi}(p^*_{\pi}; \eta^*_{\pi}) - y_{\pi}(\hat P_{\pi}; \hat \eta_{\pi}) +  y_{\pi}(p^*_{\pi}; \eta^*_{\pi})$, where the average is taken over one evaluation fold. 

\[ \mathcal F_{\hat \eta_{\pi}} = \{ (X_i, B_i(W_i), W_i) \mapsto  \Gamma^y_{\pi i}(p; \hat \eta_{\pi}) : p \in \mathcal S \} \] 

By Assumption \ref{as:regularo} for some finite $C$, the $\epsilon$ covering number of $\mathcal F_{\hat \eta_{\pi}} $ is bounded by $C (1/\epsilon)^{2 \alpha}$, for all $0 < \epsilon < 1$. This means that $\mathcal F_{\hat \eta_{\pi}}$ is a Donsker class of functions. Since we also have weak continuity of $y(B_i(w), X_i, p)$ in the sense of Assumption \ref{as:regularo}, by Lemma 19.24 of \citet{van2000asymptotic}, for all $t > 0$,
we have $\lim \limits_{n \to \infty} \Pr(\sqrt n |R_n|  > t \mid \hat \eta_{\pi}) \to 0$. Conditional convergence in probability implies unconditional convergence in probability, since $\Pr(\sqrt n |R_n| > t ) = \mathbb E [ \Pr(\sqrt n |R_n|  > t \mid \hat \eta_{\pi}) ]$, and the probability is bounded so we can swap the limit and the expectation. This means $R_n = o_p(n^{-1/2})$. 

We have now shown that

\[ \Gamma^y_{n, \pi}(\hat P_{\pi}; \hat \eta_{\pi})  - \Gamma^y_{n, \pi}(p^*_{\pi}; \eta^*_{\pi}) - y_{\pi}(\hat P_{\pi}; \hat \eta_{\pi}) +  y_{\pi}(p^*_{\pi}; \eta^*_{\pi})   = o_p(n^{-1/2}). \] The proof follows the same argument for $\Gamma^z_{n, \pi}(\cdot)$.

\end{proof} 

\begin{lemma} \label{lem:unuis} \textbf{Uniform Nuisance Convergence.}

Under the assumptions of Theorem \ref{thm:norm}, there is a finite $C > 0$ such that with probability at least $1 - o(1)$, 
\begin{equation*} 
\begin{split} 
\sup \limits_{\pi \in \Pi} \sqrt n  \| z_{\pi}(\hat P_{\pi} ; \eta^*_{\pi}) - z_{\pi}(\hat P_{\pi}; \hat \eta_{\pi}) \|  & \leq C \sqrt n \big( \rho_{e, n} \sup \limits_{\pi \in \Pi} \| \hat P_{\pi} - p^*_{\pi} \| + \rho_{\mu, n}\rho_{e, n} + \rho_{e,n} \rho_{\theta, n} \big), \\ 
 \sup \limits_{\pi \in \Pi}  \sqrt n |  y_{\pi}(\hat P_{\pi} ; \hat \eta_{\pi}) - y_{\pi}(\hat P_{\pi} ; \eta^*_{\pi})|  & \leq C \sqrt n \big( \rho_{e, n} \sup \limits_{\pi \in \Pi} \| \hat P_{\pi} - p^*_{\pi} \| + \rho_{\mu, n}\rho_{e, n} + \rho_{e,n} \rho_{\theta, n} \big). \\ 
\end{split} 
\end{equation*} 

This type of inequality also holds pointwise, in that for the same $C$, with probability at least $1 - o(1)$, for each $\pi \in \Pi$, we have: 
\begin{equation*} 
\begin{split} 
 \sqrt n  \| z_{\pi}(\hat P_{\pi} ; \eta^*_{\pi}) - z_{\pi}(\hat P_{\pi}; \hat \eta_{\pi}) \|  & \leq C \sqrt n \big( \rho_{e, n} \| \hat P_{\pi} - p^*_{\pi} \| + \rho_{\mu, n}\rho_{e, n} + \rho_{e,n} \rho_{\theta, n} \big), \\ 
 \sqrt n |  y_{\pi}(\hat P_{\pi} ; \hat \eta_{\pi}) - y_{\pi}(\hat P_{\pi} ; \eta^*_{\pi})|  & \leq C \sqrt n \big( \rho_{e, n} \| \hat P_{\pi} - p^*_{\pi} \| + \rho_{\mu, n}\rho_{e, n} + \rho_{e,n} \rho_{\theta, n} \big). \\ 
\end{split} 
\end{equation*}

\end{lemma} 

\begin{proof} 

We prove this for $z_{\pi}(\cdot)$ and the argument for $y_{\pi}(\cdot)$ is the same. 
\begin{equation*} 
\begin{split} 
 z_{\pi}(\hat P_{\pi} ; \eta^*_{\pi}) - z_{\pi}(\hat P_{\pi} ; \hat \eta_{\pi}) & = \mathbb E_{T}[\pi(X_i)( \Gamma^d_{1i}(\hat P_{\pi}; \eta^*_{\pi} ) - \Gamma^d_{1i}(\hat P_{\pi} ; \hat \eta_{\pi}))] \\ & \phantom{=} + \mathbb E_{T}[(1 - \pi(X_i)) (\Gamma^d_{0i}(\hat P_{\pi} ;  \eta^*_{\pi}) -  \Gamma^d_{0i}(\hat P_{\pi};  \hat \eta_{\pi})) ].
 \end{split} 
 \end{equation*} 
 We bound the treated terms and the argument for the control terms is the same. 
 \begin{equation} 
 \label{eqn:zetae} 
 \begin{split} 
 \mathbb E_{T}[\pi(X_i)( \Gamma^d_{1i}(\hat P_{\pi}; \hat \eta_{\pi}) - \Gamma^d_{1i}(\hat P_{\pi} ; \eta^*_{\pi}))]  &= 
 \mathbb E_{T} \left  [  \pi(X_i) (d(B_i(1),X_i, p)  - \mu^d_1(X_i, p^*_{\pi})) \left (\frac{W_i}{ \hat e(X_i)} - \frac{W_i}{e(X_i)} \right ) \right ]_{p = \hat P_{\pi}}  \\ 
& \phantom{=} + \mathbb E_{T} 
\left [ \pi(X_i) (\hat \mu^d_1(X_i,  \tilde P_{\pi} ) - \mu^d_1(X_i, p^*_{\pi}))  \left ( \frac{W_i}{e(X_i)} - \frac{W_i}{\hat e(X_i)}  \right)   \right ]
\\ & \phantom{=} + \mathbb E_{T} \left[   \pi(X_i) (\hat \mu^d_1(X_i,  \tilde P_{\pi}) - \mu^d_1(X_i, p^*_{\pi}) ) \left (  1- \frac{W_i}{e(X_i)} \right )  \right ]
\end{split} 
\end{equation} 

The last term is equal to zero. For the first term, we can bound the absolute value of each element of the vector. With probability at least $1 - o(1)$, 

\begin{align} 
 &  \left |\mathbb E_{T} \left  [  \pi(X_i) (d_j(B_i(1), X_i, p)  - \mu^d_{1, j} (X_i, p^*_{\pi})) \left (\frac{W_i}{ \hat e(X_i)} - \frac{W_i}{e(X_i)} \right ) \right ]  \right |_{p = \hat P_{\pi}} \nonumber   \\ &   =  \left |    \mathbb E_{T} \left  [  \pi(X_i) ( \mu^d_{1, j}(X_i, p)   - \mu^d_{1, j}(X_i, p^*_{\pi}))  \left (\frac{e(X_i) - \hat e(X_i)}{ \hat e(X_i)}  \right ) \right ]  \right |_{p = \hat P_{\pi}}  \nonumber  \\ & 
 \leq     \mathbb E_{T} \left  [   \left |  \frac{\pi(X_i)}{\hat e(X_i)} \right |  \left | ( \mu^d_{1, j}(X_i, p)   - \mu^d_{1, j}(X_i, p^*_{\pi})) \right |   \left | \hat e(X_i) - e(X_i) \right |   \right ] \nonumber   \\ &  \label{eqn:muhate} 
 \leq \frac{1}{ \kappa}     \mathbb E_{T} \left  [   \left | ( \mu^d_{1, j}(X_i, p)   - \mu^d_{1, j}(X_i, p^*_{\pi})) \right |   \left | \hat e(X_i) - e(X_i) \right |   \right ]     \\ & 
 \leq \frac{\sqrt J c'}{\kappa} \| \hat P_{\pi} - p^*_{\pi} \| \sqrt {  \mathbb E_{T} \left [  ( \hat e(X_i) - e(X_i) )^2  \right ] }  \nonumber \\ & 
 \leq \frac{\sqrt J c'}{\kappa} \rho_{e, n}  \| \hat P_{\pi} - p^*_{\pi}\|   \nonumber  
\end{align} 

The second-last step holds because each $\mu^d_{1, j}(x, p)$ is $\sqrt J c'$-Lipschitz in $p$ (Assumption \ref{as:continuum_regularity}(iii)). 

Similarly, we can show that 
\begin{equation} \label{eqn:zeta2} 
\begin{split} 
  & \mathbb E_{T} 
\left [ \pi(X_i) (\hat \mu^d_{1, j}(X_i,  \tilde P_{\pi} ) - \mu^d_{1, j}(X_i, p^*_{\pi}))  \left ( \frac{W_i}{e(X_i)} - \frac{W_i}{\hat e(X_i)}  \right)   \right ]
\\ &   \leq \frac{1}{ \kappa}     \mathbb E_{T} \left  [  \left (  \left | ( \hat \mu^d_{1, j}(X_i, \tilde P_{\pi} )   - \mu^d_{1, j}(X_i, \tilde P_{\pi})) \right | + \left |   (  \mu^d_{1, j}(X_i, \tilde P_{\pi} )   - \mu^d_{1, j}(X_i, p^*_{\pi} )) \right | \right )   \left | \hat e(X_i) - e(X_i) \right |   \right ]    \\ & 
 \leq \frac{1}{ \kappa}    \sqrt {   \mathbb E_{T} \left  [   ( \hat  \mu^d_{1, j}(X_i, \tilde P_{\pi})   - \mu^d_{1, j}(X_i, \tilde P_{\pi}))^2 \right] } \sqrt {  \mathbb E_{T} \left [  ( \hat e(X_i) - e(X_i) )^2  \right ] }  \\ & \phantom{\leq}  +    \frac{1}{\kappa} \sqrt {   \mathbb E_{T} \left  [   (  \mu^d_{1, j}(X_i, \tilde P_{\pi})   - \mu^d_{1, j}(X_i, p^*_{\pi}))^2 \right] } \sqrt {  \mathbb E_{T} \left [  ( \hat e(X_i) - e(X_i) )^2  \right ] }  
 \\ & \leq  \frac{1}{\kappa} \rho_{\mu, n} \rho_{e, n} + \frac{\sqrt J c'}{\kappa} \rho_{e,n}  \rho_{\theta, n}  
\end{split} 
\end{equation} 
We have now shown that with probability at least $ 1- o(1)$, that 
\begin{equation*} 
\begin{split} 
  \| z_{\pi} (\hat P_{\pi} ; \eta^*_{\pi}) - z_{\pi}(\hat P_{\pi};  \hat \eta_{\pi}) \|   & \leq 2 \sqrt J \left  ( \frac{1}{\kappa} \rho_{\mu, n} \rho_{e, n} + \frac{\sqrt J c'}{\kappa} \rho_{e,n}  \rho_{\theta, n}  +   \frac{\sqrt J c'}{\kappa} \rho_{e, n}   \| \hat P_{\pi}  - p^*_{\pi} \| \right ),  \\ 
 \sup \limits_{\pi \in \Pi} \| z_{\pi} (\hat P_{\pi} ; \eta^*_{\pi}) - z_{\pi}(\hat P_{\pi};  \hat \eta_{\pi}) \|  & \leq 2 \sqrt J \left  ( \frac{1}{\kappa} \rho_{\mu, n} \rho_{e, n} + \frac{\sqrt J c'}{\kappa} \rho_{e,n}  \rho_{\theta,n}  +   \frac{\sqrt J c'}{\kappa} \rho_{e, n} \sup \limits_{\pi \in \Pi}  \| \hat P_{\pi}  - p^*_{\pi} \| \right ). 
\end{split} 
\end{equation*} 

\end{proof} 

\begin{lemma} \label{lem:empt} 
Under the assumptions of Theorem \ref{thm:norm}, 
\begin{equation*} 
\begin{split} 
 \sup \limits_{ p \in \mathcal S} | \Gamma^y_{n, \pi}(p ; \eta^*_{\pi}) - y_{\pi}(p; \eta^*_{\pi})|  &= O_p \left (n^{-1/2} \right), \\ 
 \sup \limits_{ p \in \mathcal S} \| \Gamma^z_{n, \pi}(p ; \eta^*_{\pi}) - z_{\pi}(p; \eta^*_{\pi})\|  & = O_p \left (n^{-1/2} \right). 
\end{split} 
\end{equation*} 
Under the assumptions of Theorem \ref{thm:regret}, 
\begin{equation*} 
\begin{split} 
 \sup \limits_{ \pi \in \Pi, p \in \mathcal S} | \Gamma^y_{n, \pi}(p ; \eta^*_{\pi}) - y_{\pi}(p; \eta^*_{\pi})|  &= O_p \left (n^{-1/2} \right), \\ 
 \sup \limits_{ \pi \in \Pi, p \in \mathcal S} \| \Gamma^z_{n, \pi}(p ; \eta^*_{\pi}) - z_{\pi}(p; \eta^*_{\pi})\|  & = O_p \left (n^{-1/2} \right). 
\end{split} 
\end{equation*} 

\end{lemma} 

\begin{proof} 
\begin{equation*} 
\begin{split} 
  \Gamma^y_{n, \pi}(p ; \eta^*_{\pi}) - y_{\pi} ( p; \eta^*_{\pi}) & = \frac{1}{n} \sum \limits_{i=1}^n \pi(X_i)  \Gamma^y_{1i}(p  ; \eta^*_{\pi})  - \mathbb E[\pi(X_i) \mu^y_1( X_i, p)]    \\ & \phantom{=} + \frac{1}{n} \sum \limits_{i=1}^n ( 1- \pi(X_i)) \Gamma^y_{0i}(p; \eta^*_{\pi}) - \mathbb E[ ( 1- \pi(X_i)) \mu^y_0( X_i, p)] 
 \end{split} 
 \end{equation*} 
 
 To bound $\sup \limits_{p \in \mathcal S, \pi \in \Pi} |   \Gamma^y_{n, \pi}(p ; \eta^*_{\pi}) - y_{\pi} ( p; \eta^*_{\pi}) | $, we will just bound the treated terms, since the argument for the control terms is the same. First, we expand the treated terms: 
 
 \begin{equation*} 
 \begin{split} 
 &  \sup \limits _{\pi \in \Pi, p \in \mathcal S}  \left | \frac{1}{n} \sum \limits_{i=1}^n \pi(X_i)  \Gamma^y_{1i}(p  ; \eta^*_{\pi}) - \mathbb E[\pi(X_i) \mu^y_1(X_i, p) ]\right |   \\ &  \leq   \sup \limits _{\pi \in \Pi, p \in \mathcal S}  \left | \frac1n \sum \limits_{i=1}^n  \pi(X_i) \frac{W_i}{e(X_i)} y(B_i(1), X_i, p) - \mathbb E[\pi(X_i) y(B_i(1), X_i, p)] \right |   \\ & \qquad \qquad +  \sup \limits _{\pi \in \Pi, p \in \mathcal S} \left | \frac1n \sum \limits_{i=1}^n \pi(X_i) \mu^y_1(X_i, p) \left ( 1 - \frac{W_i}{e(X_i)} \right ) \right |  
 \end{split} 
 \end{equation*} 
 
Since $\Pi$ is a VC class of dimension $v$, by Theorem 2.6.7 of \citet{vaart1997weak}, it has uniform covering numbers that are bounded by $C(1/\epsilon)^{2v }$ for some constant $C$. Assumption \ref{as:regularo}  implies that the function class $\mathcal F_y = \{ (B(w), X) \mapsto y(B(w), X, p) :  p \in \mathcal S \}$ has covering numbers that are bounded by $C(1/\epsilon)^{\alpha}$. Then, by Lemma \ref{lem:comp}, the function class $\mathcal G = \{ (W, X, B(1))  \mapsto  \pi(X) \frac{W}{e(X)} y(B(1), X, p) : p \in \mathcal S, \pi \in \Pi \}$ has covering numbers that are bounded by $C(1/\epsilon)^V$ for finite $V$ that is of order $v + \alpha$. By Lemma \ref{lem:tailb}, we can now conclude that  

\begin{equation} \label{eq:gytb} 
 \sup \limits _{\pi \in \Pi, p \in \mathcal S}  \left | \frac{1}{n} \sum \limits_{i=1}^n \pi(X_i)  \Gamma^y_{1i}(p  ; \eta^*_{\pi}) - \mathbb E[\pi(X_i) \mu^y_1(X_i,p) ]\right |  = O_p\left ( n^{-1/2} \right). 
\end{equation} 
$\mu_1^y(X_i, p)$ is $\sqrt J c'$-Lipschitz in $p$. Since $p \in \mathcal S$, and $\mathcal S$ is a compact subset of $\mathbb R^J$, we can show the function class $\mathcal F_{\mu} =  \{  X \mapsto \mu_1^y(X, p) : p \in \mathcal S\}$ has uniform covering number that is bounded by $C \left ( \frac{1}{\epsilon} \right)^{J}$ for some constant $C >0$. Theorem 2.7.11 of \citet{vaart1997weak} shows that the $2\epsilon \sqrt J c'$ bracketing number of $\mathcal F_{\mu}$ is bounded by the covering number of $\mathcal S$, which in turn is bounded by $C(1/\epsilon)^{J}$ for some constant $C$ (see, for example, Lemma 2.7 of \citet{sen2018gentle}). Since the $\epsilon$-uniform covering number of  $\mathcal F_{\mu}$ is bounded by the $2\epsilon$-bracketing number (see Definition 2.1.6 of \citet{vaart1997weak}), this is enough to bound the uniform covering number of $\mathcal F_{\mu}$. Again using the composition result of Lemma \ref{lem:comp} and Lemma \ref{lem:tailb} (as above), we can now conclude that 

\begin{equation}  \sup \limits _{\pi \in \Pi, p \in \mathcal S} \left | \frac1n \sum \limits_{i=1}^n \pi(X_i) \mu^y_1(X_i, p) \left ( 1 - \frac{W_i}{e(X_i)} \right ) \right |  = O_p \left (n^{-1/2} \right). 
\end{equation}  
With the same argument for the control terms, we have now concluded that: 
  \[ \sup \limits_{p \in \mathcal S, \pi \in \Pi} |   \Gamma^y_{n, \pi}(p ; \eta^*_{\pi}) - y_{\pi} ( p; \eta^*_{\pi})  |  = O_p\left (n^{-1/2} \right). \] 
Using the same argument, we can also bound each of  $ \sup \limits_{p \in \mathcal S, \pi \in \Pi} | \Gamma^z_{n, \pi, j} (p; \eta^*_{\pi}) - z_{\pi, j}(p; \eta^*_{\pi})|$  for $j \in \{ 1, \ldots, J \}$ and, using a union bound also  conclude that: 
  \[ \sup \limits_{p \in \mathcal S, \pi \in \Pi}  \| \Gamma^z_{n, \pi}(p; \eta^*_{\pi}) - z_{\pi}(p; \eta^*_{\pi}) \|  = O_p\left (n^{-1/2} \right).  \] 
 For the first part of the Lemma, we can follow the same argument as above without taking the supremum over $\Pi$. 

\end{proof} 

\begin{lemma} \label{lem:empe} 

Under the assumptions of Theorem \ref{thm:norm}, 
\begin{equation*} 
\begin{split} 
 \sup \limits_{ p \in \mathcal S} | \Gamma^y_{n, \pi}(p ; \hat \eta_{\pi}) - y_{\pi}(p; \hat \eta_{\pi})|  &= O_p \left (n^{-1/2} \right), \\ 
 \sup \limits_{ p \in \mathcal S} \| \Gamma^z_{n, \pi}(p ; \hat \eta_{\pi}) - z_{\pi}(p; \hat \eta_{\pi})\|  & = O_p \left (n^{-1/2} \right), 
\end{split} 
\end{equation*} 

Under the assumptions of Theorem \ref{thm:regret}, 
\begin{equation*} 
\begin{split} 
 \sup \limits_{\pi \in \Pi, p \in \mathcal S} | \Gamma^y_{n, \pi}(p ; \hat \eta_{\pi}) - y_{\pi}(p; \hat \eta_{\pi})|  &= O_p \left (n^{-1/2} \right), \\ 
 \sup \limits_{\pi \in \Pi, p \in \mathcal S} \| \Gamma^z_{n, \pi}(p ; \hat \eta_{\pi}) - z_{\pi}(p; \hat \eta_{\pi})\|  & = O_p \left (n^{-1/2} \right), 
\end{split} 
\end{equation*} 

\end{lemma}

\begin{proof} 

We start with the second part of the Lemma. 
\begin{equation} \label{eq:abedw}
\begin{split} 
  \Gamma^y_{n, \pi}(p ;  \hat \eta_{\pi}) - y_{\pi} ( p;  \hat \eta_{\pi}) & = \frac{1}{n} \sum \limits_{i=1}^n \pi(X_i)  \Gamma^y_{1i}(p  ;  \hat \eta_{\pi})  - \mathbb E_{T}[ \pi(X_i)  \Gamma^y_{1i}(p  ;  \hat \eta_{\pi}) ]    \\ & \phantom{=} + \frac{1}{n} \sum \limits_{i=1}^n ( 1- \pi(X_i)) \Gamma^y_{0i}(p;  \hat \eta_{\pi}) - \mathbb E_{T}[ ( 1- \pi(X_i)) \Gamma^y_{0i}(p;  \hat \eta_{\pi})]  
  \end{split} 
  \end{equation} 
We show the details for the treated terms only since the argument for the control terms is the same. As in Section \ref{sec:est_gte}, we suppress the fold notation. The data in the evaluation fold, $\tilde P_{\pi}$ and $\hat e(\cdot), \hat \mu(\cdot)$ are mutually independent. 
We have the below expansion.  
 \begin{equation*} 
 \begin{split} 
& \sup \limits_{\pi \in \Pi, p \in \mathcal S} \left |   \frac{1}{n} \sum \limits_{i=1}^n  \pi(X_i)  \Gamma^y_{1i}(p  ;  \hat \eta_{\pi})  - \mathbb E_{T}[ \pi(X_i)  \Gamma^y_{1i}(p  ;  \hat \eta_{\pi}) ]  \right | \\ 
    & \leq   \sup \limits _{\pi \in \Pi, p \in \mathcal S}  \left | \frac{1}{n} \sum \limits_{i=1}^n  \pi(X_i) \frac{W_i}{ \hat e(X_i)} y(B_i(1), X_i, p) - \mathbb E_T\left [ \frac{W_i}{ \hat e(X_i)} \pi(X_i) y(B_i(1), X_i, p) \right] \right |   \\ & \phantom{=} 
    +  \sup \limits _{\pi \in \Pi} \left | \frac{1}{n} \sum \limits_{i=1}^n  \pi(X_i) \hat \mu^y_1(X_i, \tilde P_{\pi}) \left ( 1 - \frac{W_i}{\hat e(X_i)} \right ) - \mathbb E_{T} \left [  \pi(X_i) \hat \mu^y_1(X_i, \tilde P_{\pi} ) \left ( 1 - \frac{W_i}{\hat e(X_i)} \right ) \right ]  \right |  \\ 
    & \stackrel{(1)}{\leq}  O_p(n^{-1/2}) + \sup \limits _{\pi \in \Pi, p \in \mathcal S} \left | \frac{1}{n} \sum \limits_{i=1}^n \pi(X_i) \hat \mu^y_1(X_i, p) \left ( 1 - \frac{W_i}{\hat e(X_i)} \right ) - \mathbb E_{T} \left [  \pi(X_i) \hat \mu^y_1(X_i, p ) \left ( 1 - \frac{W_i}{\hat e(X_i)} \right ) \right ]  \right | \\
    & \stackrel{(2)}{=} O_p(n^{-1/2}) 
 \end{split} 
 \end{equation*} 
 
For term that we handle in (1), we can condition on $\hat e(\cdot)$ and treat it as fixed. Conditional on $\hat e(\cdot)$, this term is mean-zero. Then, because of the uniform overlap condition, the tail bound for this term constructed in the same way as in \eqref{eq:gytb} does not depend on the estimated part of the nuisance function, so unconditionally, we also have that the term is $O_p(n^{-1/2})$. 

For the next term, we rely on the additional assumption in Assumption \ref{as:anuis} and the assumption that estimated conditional mean functions are uniformly bounded. Again, we can use the composition result and tail bound in  Lemma \ref{lem:comp} and Lemma \ref{lem:tailb} to construct a tail bound for the term that does not depend on the specific instance of the estimated function.

This argument applies for each split, and can be applied also to the control terms, and to each of the components of $ \Gamma^z_{n, \pi}(\cdot)$, so we can now conclude that: 
\begin{equation*} 
\begin{split} 
 \sup \limits_{\pi \in \Pi, p \in \mathcal S} | \Gamma^y_{n, \pi}(p ; \hat \eta_{\pi}) - y_{\pi}(p; \hat \eta_{\pi})|  &= O_p \left (n^{-1/2} \right), \\ 
 \sup \limits_{\pi \in \Pi, p \in \mathcal S} \| \Gamma^z_{n, \pi}(p ; \hat \eta_{\pi}) - z_{\pi}(p; \hat  \eta_{\pi})\|  & = O_p \left (n^{-1/2} \right). 
\end{split} 
\end{equation*} To finish the proof, without using Assumption \ref{as:anuis}, then under the assumptions of Theorem \ref{thm:norm}, we have 
\begin{equation*} 
\begin{split} 
& \sup \limits_{p \in \mathcal S} \left |   \frac{1}{n} \sum \limits_{i=1}^n  \pi(X_i)  \Gamma^y_{1i}(p  ;  \hat \eta_{\pi})  - \mathbb E_{T}[ \pi(X_i)  \Gamma^y_{1i}(p  ;  \hat \eta_{\pi}) ]  \right | \\ 
    & \leq   \sup \limits _{p \in \mathcal S}  \left | \frac{1}{n} \sum \limits_{i=1}^n  \pi(X_i) \frac{W_i}{ \hat e(X_i)} y(B_i(1), X_i, p) - \mathbb E_T\left [ \frac{W_i}{ \hat e(X_i)} \pi(X_i) y(B_i(1), X_i, p) \right] \right |   \\ & \phantom{=} 
    + \left | \frac{1}{n} \sum \limits_{i=1}^n \pi(X_i) \hat \mu^y_1(X_i, \tilde P_{\pi}) \left ( 1 - \frac{W_i}{\hat e(X_i)} \right ) - \mathbb E_{T} \left [  \pi(X_i) \hat \mu^y_1(X_i, \tilde P_{\pi} ) \left ( 1 - \frac{W_i}{\hat e(X_i)} \right ) \right ]  \right |  \\ 
    & =  O_p(n^{-1/2}), 
    \end{split} 
\end{equation*} 
where the first term is $O_p(n^{-1/2})$ by the same argument as above (when we also take the supremum over $\pi \in \Pi$). Conditional on the estimated nuisances, the second term is mean-zero with finite variance. By the CLT, then conditional on estimated nuisances, it is $O_p(n^{-1/2})$, where we can choose constants in the $O_p(n^{-1/2})$ definition that are uniform over all possible instances of the nuisance parameters, by the uniform boundedness of the estimated nuisances. So, the second term is $O_p(n^{-1/2})$ as well. 

By \eqref{eq:abedw}, (and since the same argument applies to $ \Gamma^z_{n, \pi}(\cdot)$), we have now shown that : 
\begin{equation*} 
\begin{split} 
 \sup \limits_{ p \in \mathcal S} | \Gamma^y_{n, \pi}(p ; \hat \eta_{\pi}) - y_{\pi}(p; \hat \eta_{\pi})|  &= O_p \left (n^{-1/2} \right), \\ 
 \sup \limits_{ p \in \mathcal S} \| \Gamma^z_{n, \pi}(p ; \hat \eta_{\pi}) - z_{\pi}(p; \hat  \eta_{\pi})\|  & = O_p \left (n^{-1/2} \right). 
\end{split} 
\end{equation*} 

\end{proof} 

\begin{lemma}\label{lem:concp} \textbf{Concentration of estimated market-clearing cutoffs} 

Under the assumptions of Theorem \ref{thm:regret}, 
\[  \sup \limits_{\pi \in \Pi} \|  \hat P_{\pi} - p^*_{\pi}  \|   = O_p(n^{-1/2}). \] 

Under the assumptions of Theorem \ref{thm:norm}, for each $\pi \in \Pi$, 
\[   \|  \hat P_{\pi} - p^*_{\pi}  \|   = O_p(n^{-1/2}). \] 

\end{lemma}

\begin{proof} 

We start with a version of uniform consistency. %

 By the twice continuous differentiability of $z_{\pi}(p; \eta^*_{\pi})$ in $p$ with bounded derivatives,  then $\nabla_p z_{\pi}(p)$ is Lipschitz continuous in $p$ with constant $c'$. Specifically, for any $\epsilon >0$ and any $p$ that is an element of the open ball $ \mathcal B(p^*; \epsilon/c')$, then $\| \nabla_p z_{\pi}(p) - \nabla_p z_{\pi}(p^*_{\pi} )\| \leq J \epsilon  $.  By the mean-value form of the Taylor expansion, there exists a $\bar p$ such that 
\begin{equation}  \label{eq:wellsep}
\begin{split} 
\| z_{\pi}(p; \eta^*_{\pi}) - z_{\pi}(p^*_{\pi}; \eta^*_{\pi})\|  & =   \| \nabla_p z_{\pi}(\bar p) (p - p^*_{\pi}) \| \\
& \geq \| \nabla_p z_{\pi}(p^*_{\pi} )(p - p^*_{\pi}) \| - \| ( \nabla_p z_{\pi}(\bar p) - \nabla_p z_{\pi}(p^*_{\pi}))(p - p^*_{\pi}) \| \\ 
& \stackrel{(1)} \geq \| \nabla_p z_{\pi}(p^*_{\pi} )(p - p^*_{\pi}) \| -  \epsilon J \| (p - p^*_{\pi}) \| \\ 
& \stackrel{(2)}{\geq} \|  \nabla_p z_{\pi}(p^*_{\pi} ) (p - p^*_{\pi}) \| - \frac{1}{2} \|  \nabla_p z_{\pi}(p^*_{\pi} ) (p - p^*_{\pi}) \|  \\ 
& = \frac{1}{2} \|  \nabla_p z_{\pi}(p^*_{\pi} ) (p - p^*_{\pi})  \| \\ 
& \geq \frac{c_3}{2}  \| p - p^*_{\pi} \| 
\end{split} 
\end{equation} 

So, for any $p \in\mathcal B(p^*_{\pi}; \frac{c_3}{2J c'} )$, $2 \| z_{\pi}(p; \eta^*_{\pi}) \| \geq c_3 \| p - p^*_{\pi} \|$. 
In addition, by Assumption \ref{as:continuum_regularity}(i), for any $p \in \mathcal S \backslash \mathcal B(p^*_{\pi}; \frac{c_3}{2J c'} ), 2 \| z_{\pi}(p) \| \geq c_2$. 

\begin{equation*} 
\sup \limits_{\pi \in \Pi}  \min \{ c_3 \|\hat P_{\pi} - p^*_{\pi} \|, c_2 \}  \leq 2 \sup \limits_{\pi \in \Pi}  \left  | \left| z_{\pi}(\hat P_{\pi}; \eta^*_{\pi})  \right | \right| 
\end{equation*} 

To finish the proof of uniform consistency, then, we need to show that with probability $1 - o(1)$, 

\[  \sup \limits_{\pi \in \Pi} \|  z_{\pi} (\hat P_{\pi}; \eta^*_{\pi}) \|  \leq g_n,  \] 

for $g_n = o(1)$. Since $c_3 >0$ and $c_2 > 0 $ are fixed constants, this implies for sufficiently large $n$, that with probability $1 - o(1)$, that  $\| \hat P_{\pi} - p^*_{\pi} \| \leq  b_n $ for $b_n = o(1)$. 

We proceed using the following decomposition: 
\begin{equation*} 
\begin{split} 
\sup \limits_{\pi \in \Pi} \| z_{\pi}(\hat P_{\pi}; \eta^*_{\pi})\|  &  \leq  \underbrace{ \sup \limits_{\pi \in \Pi}  \|  z_{\pi}(\hat P_{\pi} ; \eta^*_{\pi}) - z_{\pi}(\hat P_{\pi} ; \hat \eta_{\pi})\| }_{(i)}  +  \underbrace{ \sup \limits_{\pi \in \Pi}  \| z_{\pi} (\hat P_{\pi}; \hat \eta_{\pi}) - \Gamma^z_{n, \pi}(\hat P_{\pi}; \hat \eta_{\pi})\| }_{(ii)}  \\ &\qquad \qquad  +  \underbrace{  \sup \limits_{\pi \in \Pi} \| \Gamma^z_{n, \pi}(\hat P_{\pi}; \hat \eta_{\pi}) \| }_{(iii)} 
\end{split} 
\end{equation*} Since $\|p- p^*_{\pi} \|$is bounded, $(i)$ is $o_p(1)$ by Lemma \ref{lem:unuis}. Lemma \ref{lem:empe} indicates that (ii) is $O_p(n^{-1/2})$. For $(iii)$, we use the last  part of Assumption \ref{as:nuisance} which implies that $\sup \limits_{\pi \in \Pi }\| \Gamma^z_{n, \pi}(\hat P_{\pi}; \hat \eta_{\pi}  ) \|= o_p(n^{-1/2})$. 

Combining the bounds for each of these terms, we have now shown that $ \sup \limits_{\pi \in \Pi} \| \hat P_{\pi} - p^*_{\pi} \| = o_p(1)$. Next, we want to strengthen the uniform consistency result into a rate. We want to show that 
\begin{equation} \label{eq:Pexp} \sup \limits_{\pi \in \Pi}  \sqrt n \| \hat P_{\pi} - p^*_{\pi} \| \leq \sup \limits_{\pi \in \Pi}  \| (\nabla_p z_{\pi}(p^*_{\pi}))^{-1}\| \| \sqrt n \Gamma^z_{n, \pi}(p^*_{\pi}; \eta^*_{\pi}) \| + \sqrt n \frac{J^{3/2} c'}{2 c_3} R_{1n}  \sup \limits_{\pi \in \Pi} \| \hat P_{\pi} - p^*_{\pi} \| + R_{2n},   \end{equation} 

where $R_{1n} = o_p(1)$ and $R_{2n} = O_p(1)$. Once we have this, the proof is straightforward. Since the singular values of $\nabla_p z_{\pi}(p^*_{\pi})$ are uniformly bounded by $c_3$ from below and $z_{\pi}(p^*_{\pi}; \eta^*_{\pi}) = 0$. 
\begin{equation*} 
\begin{split} 
 \sup \limits_{\pi \in \Pi} \sqrt n \| \hat P _{\pi} - p^*_{\pi} \| \left( 1 - \frac{J^{3/2} c'}{2 c_3} R_{1n} \right) &  \leq \frac{1}{c_3}  \sup \limits_{\pi \in \Pi} \| \Gamma^z_{n, \pi}(p^*_{\pi}; \eta^*_{\pi})  - z_{\pi}(p^*_{\pi} ; \eta^*_{\pi}) \| + R_{2n}   \\ 
  \sup \limits_{\pi \in \Pi} \sqrt n \|  \hat P _{\pi} - p^*_{\pi} \|  \left( 1 - \frac{J^{3/2} c'}{2 c_3} R_{1n} \right) & \leq  \frac{1}{c_3} \sup \limits_{\pi \in \Pi, p \in \mathcal S } \| \Gamma^z_{n, \pi}(p; \eta^*_{\pi}) - z_{\pi}(p; \eta^*_{\pi}) \|  + R_{2n}
 \end{split}
 \end{equation*} 
 
 Since $1/(1 - \frac{J^{3/2} c'}{2 c_3} R_{1n}) = O_p(1)$, $R_{2n}  = O_p(1)$, and by Lemma  \ref{lem:empt}, $ \sup \limits_{\pi \in \Pi, p \in \mathcal S } \| \Gamma^z_{n, \pi}(p; \eta^*_{\pi}) - z_{\pi}(p; \eta^*_{\pi}) \|  = O_p(n^{-1/2})$, then $ \sup \limits_{\pi \in \Pi} \sqrt n \| \hat P_{\pi} - p^*_{\pi} \| = O_p(1).$ So, to finish the proof, we must show \eqref{eq:Pexp}, with the required convergence properties for $R_{1n}$ and $R_{2n}$. We start with the following expansion:
\begin{equation*}
\begin{split}  
 \Gamma^z_{n, \pi}(\hat P_{\pi}; \eta^*_{\pi}) - \Gamma^z_{n, \pi}(p^*_{\pi} ; \eta^*_{\pi}) +  U_{1n}  & = z_{\pi}(\hat P_{\pi} ; \eta^*_{\pi}) - z_{\pi}( p^*_{\pi} ; \eta^*_{\pi}),  \\ 
 \Gamma^z_{n, \pi}(\hat P_{\pi} ; \hat \eta_{\pi})  - \Gamma^z_{n, \pi}(p^*_{\pi} ; \eta^*_{\pi})  + U_{1n} + U_{2n} & = z_{\pi}(\hat P_{\pi} ; \eta^*_{\pi}) - z_{\pi}( p^*_{\pi} ; \eta^*_{\pi}),    \\ 
 -\Gamma^z_{n, \pi}(p^*_{\pi}; \eta^*_{\pi}) +U_{1n} + U_{2n}  & =  \nabla_p z_{\pi}(p^*_{\pi}) (\hat P_{\pi} - p^*_{\pi})  + \mathcal O( \| \hat P_{\pi} - p^*_{\pi} \|^2),   \\  
\end{split} 
\end{equation*} 

where $U_{1n} = \Gamma^z_{n, \pi}(p^*_{\pi}; \eta^*_{\pi}) - z_{\pi}(p^*_{\pi}; \eta^*_{\pi})  -\Gamma^z_{n, \pi}(\hat P_{\pi}; \eta^*_{\pi}) + z_{\pi}(\hat P_{\pi}; \eta^*_{\pi}) $, $U_{2n} = - \Gamma^z_{n, \pi}(\hat P_{\pi}; \eta^*_{\pi}) +  \Gamma^z_{n, \pi}(\hat P_{\pi}; \hat \eta_{\pi})  $, and the last step is by the mean-value form for a Taylor expansion. 

By the mean-value form for a Taylor expansion of $z_{\pi}(\hat P_{\pi}) - z_{\pi}(p^*_{\pi})$ and Assumption \ref{as:continuum_regularity}(iii), then the previous step implies:
\begin{equation} \label{eqn:rate1} 
 \begin{split} \sqrt n \|\hat P_{\pi} - p^*_{\pi} \| & \leq  \| (\nabla_p z_{\pi}(p^*_{\pi}))^{-1} \| \biggl( \sqrt n \|\Gamma^z_{n, \pi}(p^*_{\pi}; \eta^*_{\pi})\| \\ & \phantom{\leq} \qquad + \frac{J^{3/2} c'}{2} \sqrt n \| \hat P_{\pi} - p^*_{\pi}\|^2 + \sqrt n \| U_{1n} \| + \sqrt n \| U_{2n} \| \biggr). \end{split} 
 \end{equation} 

To finish showing a version of \eqref{eq:Pexp}, we examine $U_{1n}$ and $U_{2n}$ more closely.

\begin{equation*} 
\begin{split} 
 \sup \limits_{\pi \in \Pi} \sqrt n \|  U_{1n} \| &  = \sup \limits_{\pi \in \Pi} \sqrt n  \| \Gamma^z_{n, \pi}(p^*_{\pi}; \eta^*_{\pi}) - z_{\pi}(p^*_{\pi}; \eta^*_{\pi})  -\Gamma^z_{n, \pi}(\hat P_{\pi}; \eta^*_{\pi}) + z_{\pi}(\hat P_{\pi}; \eta^*_{\pi}) \| \\ 
 & \leq 2 \sqrt n \sup \limits_{p \in \mathcal S, \pi \in \Pi}  \| \Gamma^z_{n, \pi}(p; \eta^*_{\pi}) - z_{\pi}(p; \eta^*_{\pi})\| \\ 
 & = O_p(1), 
\end{split} 
\end{equation*} 
where the equality sign follows from Lemma \ref{lem:empt}. 
\begin{equation*} 
\begin{split} 
\sup \limits_{\pi \in \Pi} \|U_{2n} \|  & =  \sup \limits_{\pi \in \Pi}  \| \Gamma^z_{n, \pi}(\hat P_{\pi}; \eta^*_{\pi}) -  \Gamma^z_{n, \pi}(\hat P_{\pi}; \hat \eta_{\pi}) \| \\ 
& \leq  \| z_{\pi}(\hat P_{\pi} ; \eta^*_{\pi}) - z_{\pi}(\hat P_{\pi}; \hat \eta_{\pi}) \| +  \|  \Gamma^z_{n, \pi}(\hat P_{\pi}; \eta^*_{\pi}) - z_{\pi}(\hat P_{\pi}; \eta^*_{\pi}) - \Gamma^z_{n, \pi}(\hat P_{\pi}; \hat \eta_{\pi}) + z_{\pi}(\hat P_{\pi}; \hat \eta_{\pi}) \|  \\ 
 \end{split} 
\end{equation*} 
 
For the second term, we rely on Lemma \ref{lem:empt} and \ref{lem:empe} yet again, which implies that $\sup \limits_{\pi \in \Pi, p \in \mathcal S} \| \Gamma^z_{n, \pi}(p ; \hat \eta_{\pi}) - z_{\pi}(p ; \hat \eta_{\pi}) \|  = O_p(n^{-1/2})$ and $\sup \limits_{\pi \in \Pi, p \in \mathcal S} \| \Gamma^z_{n, \pi}(p; \eta^*_{\pi}) - z_{\pi}(p; \eta^*_{\pi}) \| = O_p(n^{-1/2}).$  
 
For the first term, Lemma \ref{lem:unuis} implies that 
\begin{equation*} 
\begin{split} 
\sup \limits_{\pi \in \Pi} \sqrt n  \| z_{\pi}(\hat P_{\pi} ; \eta^*_{\pi}) - z_{\pi}(\hat P_{\pi}; \hat \eta_{\pi}) \|  & \leq  A_n  \sup \limits_{\pi \in \Pi} \sqrt n \| \hat P_{\pi} - p^*_{\pi} \|  + o_p(1) \\ 
\end{split} 
\end{equation*} 
where $A_n = o_p(1)$ by Assumption \ref{as:nuisance}. Plugging these bounds for $U_{1n}$ and $U_{2n}$ back into \eqref{eqn:rate1}, we have now shown a version of \eqref{eq:Pexp}, which completes the proof: 

\begin{equation*} 
\begin{split} \sup \limits_{\pi \in \Pi} \sqrt n \|\hat P_{\pi} - p^*_{\pi} \| & \leq  \| ( \nabla_p z_{\pi}(p^*_{\pi}))^{-1} \| \|\Gamma^z_{n, \pi}(p^*_{\pi}; \eta^*_{\pi})\| \\ & \phantom{\leq} +  \left(\frac{J^{3/2} c'}{2 c_3} \| \hat P_{\pi} - p^*_{\pi} \| + o_p(1) \right) \|  \hat P_{\pi} - p^*_{\pi} \|  + o_p(1) + O_p(1). \end{split} 
\end{equation*}

Under the assumptions of Theorem \ref{thm:norm}, we can follow the above argument pointwise for each $\pi \in \Pi$ rather than uniformly over $\pi$.  For the pointwise results, whenever Lemma \ref{lem:empe} is used in the above argument, we only need the part that is uniform over $p \in \mathcal S$, which requires only the assumptions of Theorem \ref{thm:norm}. 

\end{proof}

\begin{lemma} \label{lem:asympt} \textbf{Asymptotic Equicontinuity} 

Under the assumptions of Theorem \ref{thm:moment}, 
\begin{equation*}
\begin{split} 
   Y_{n, \pi} (P_{\pi} ) - Y_{n, \pi} (p^*_{\pi}) -   y_{\pi}( P_{\pi}  )  + y_{\pi}(p^*_{\pi})   & = o_p(n^{-1/2}), \\
   Z_{n, \pi} ( P_{\pi}  ) - Z_{n, \pi} (p^*_{\pi}) -  z_{\pi}(P_{\pi} )  + z_{\pi}(p^*_{\pi})    & = o_p(n^{-1/2}),  
\end{split} 
\end{equation*} 

\end{lemma} 

\begin{proof} 

We prove this for $Y_{n, \pi}(\cdot)$ and the proof is the same for each element of the $J$-length vector $Z_{n, \pi}(\cdot)$. Let $\mathcal F = \{ (X_i, B_i(W_{\pi, i}), W_{\pi, i}) \mapsto  W_{\pi, i} y(B_i(1), X_i, p) + ( 1- W_{\pi, i}) y(B_i(0), X_i, p) : p \in \mathcal S \}$. 

Notice that $\mathbb E[ Y_{n, \pi}(p)]  =  y_{\pi}(p)$.  By Assumption \ref{as:regularo}, for some finite $C$, the $\epsilon$ covering number of $\mathcal F$ is bounded by $C (1/\epsilon)^{2 \alpha}$, for all $0 < \epsilon < 1$. So, $\mathcal F$ is a Donsker class of functions. Since $\| P_{\pi} - p^*_{\pi} \| = O_p(n^{-1/2})$ by Lemma \ref{lem:concbarp}, and we also have weak continuity of $W_{\pi, i} y(B_i(1), X_i, p) + ( 1- W_{\pi, i}) y(B_i(0), X_i, p)$ in the sense of Assumption \ref{as:regularo}, by Lemma 19.24 of \citet{van2000asymptotic}, we have that $Y_{n, \pi} (P_{\pi}) - Y_{n, \pi} (p^*_{\pi}) -  y_{\pi}(P_{\pi})  + y_{\pi}(p^*_{\pi})   = o_p(n^{-1/2})$. 

\end{proof} 

\begin{lemma} \label{lem:unifpi}
Let $U_1, \ldots, U_n$ be i.i.d.\ $U[0, 1]$, independent of the data, and set $W_{\pi, i} = \mathbbm 1(U_i < \pi(X_i))$ for every $\pi \in \Pi$. Under the assumptions of Theorem \ref{thm:regret_finite}, 
\begin{align*} 
\mathbb E \Big[ \sup \limits_{\pi \in \Pi, p \in \mathcal S} \sqrt n \| Z_{n, \pi}(p) - z_{\pi}(p) \| \Big] & = O(1), \\ 
\mathbb E \Big[ \sup \limits_{\pi \in \Pi, p \in \mathcal S} \sqrt n | Y_{n, \pi}(p) - y_{\pi}(p) | \Big] & = O(1). 
\end{align*} 
Under the assumptions of Theorem \ref{thm:moment}, the above for fixed $\pi \in \Pi$, dropping the supremum.
\end{lemma} 

\begin{proof}  
We prove this for $Y_{n, \pi}(\cdot)$, and the proof is the same for each component of $Z_{n, \pi}(\cdot)$. Let $\mathcal G = \{ g_{\pi, p} : \pi \in \Pi, p \in \mathcal S \}$, where 
\[ g_{\pi, p}(u, b_1, b_0, x) = \mathbbm 1(u < \pi(x)) \, y(b_1, x, p) + \mathbbm 1(u \geq \pi(x)) \, y(b_0, x, p). \] 
We have that $O_i = (U_i, B_i(1), B_i(0), X_i)$ is i.i.d., so  $Y_{n, \pi}(p)$ is the average of $g_{\pi, p}(O_i)$ over $i$. Since $U_i$ is independent of $(B_i(1), B_i(0), X_i)$, $\mathbb E[g_{\pi, p}(O_i)] = y_{\pi}(p)$. This means $\sup \limits_{\pi \in \Pi, p \in \mathcal S} \sqrt n | Y_{n, \pi}(p) - y_{\pi}(p) |$ is the supremum over $g \in \mathcal G$ of $\big| \frac{1}{\sqrt n} \sum_{i=1}^n ( g(O_i) - \mathbb E[g(O_i)] ) \big|$. 

We now bound the uniform covering numbers of $\mathcal G$. Since $\Pi$ is a VC class of dimension $v$, its subgraphs $\{ (x, u) : u < \pi(x) \}$ form a VC class of sets, so $\mathcal W = \{ (u, x) \mapsto \mathbbm 1(u < \pi(x)) : \pi \in \Pi \}$ has uniform covering numbers bounded by a polynomial in $1/\epsilon$ (Theorem 2.6.4 of \citet{vaart1997weak}). By Assumption \ref{as:regularo}(i), $\mathcal F_y$ also has uniform covering numbers bounded the same way. Each $g_{\pi, p}$ is a Lipschitz combination of functions  from these classes. By Lemma \ref{lem:comp}, $\mathcal G$ also has uniform covering numbers bounded by a polynomial in $1/\epsilon$. Lemma \ref{lem:tailb} then bounds the tail of this supremum by $C t^V e^{-t^2 / (2M^2)}$, which implies the expectation is $O(1)$. For a fixed $\pi$, the same argument applies to $\{ g_{\pi, p} : p \in \mathcal S \}$, without $\mathcal W$. 
\end{proof} 

\begin{lemma} \label{lem:concbarp} \textbf{Concentration of finite-market cutoffs} 
Under the assumptions of Theorem \ref{thm:moment},  $\mathbb E_{\pi} [ \| P_{\pi} - p^*_{\pi} \|]  = O_p(n^{-1/2})$ and $ \| P_{\pi} - p^*_{\pi} \| = O_p(n^{-1/2})$. Under the assumptions of Theorem \ref{thm:regret_finite}, $ \sup \limits_{\pi \in \Pi}  \mathbb E_{\pi} [ \| P_{\pi} - p^*_{\pi} \| ]   = O_p(n^{-1/2}).$

\end{lemma} 

\begin{proof} 
 $\mathbb E_{\pi}$ depends only on the distribution of $\bm W_{\pi}$ given the data. So, we couple the draws for all $\pi \in \Pi$ as in Lemma \ref{lem:unifpi}. Then $\mathbb E_{\pi} = \mathbb E_{\bm U}$, the expectation over $\bm U$ conditional on the data for every $\pi$, so by Jensen's inequality $ \sup \limits_{\pi \in \Pi}  \mathbb E_{\pi} [ \| P_{\pi} - p^*_{\pi} \| ]  \leq \mathbb E_{\bm U}  \left [  \sup \limits_{\pi \in \Pi}  \| P_{\pi} - p^*_{\pi} \| \right ].$
By \eqref{eq:wellsep}, we have that
 
\[ \sup \limits_{\pi \in \Pi}  \min \{ c_3 \| P_{\pi} - p^*_{\pi} \|, c_2 \}  \leq 2 \sup \limits_{\pi \in \Pi}  \left  | \left| z_{\pi}(P_{\pi})  \right | \right|.   \] 
So, we can finish the proof by showing that $ \mathbb E_{\bm U} \left [ \sup \limits_{\pi \in \Pi} \| z_{\pi}(P_{\pi})\| \right]  = O_p(n^{-1/2})$.

\begin{align} 
\mathbb E_{\bm U}  \left [ \sup \limits_{\pi  \in \Pi} \| z_{\pi}(P_{\pi})\| \right] &  \leq  \mathbb E_{\bm U} \left [ \sup \limits_{\pi \in \Pi} \| z_{\pi}(P_{\pi}) -  Z_{n, \pi}( P_{\pi}) \| \right ]  + \mathbb E_{\bm U} \left [  \sup \limits_{\pi \in \Pi} \|  Z_{n, \pi}(P_{\pi}) \| \right]  \nonumber \\ 
& \leq  \mathbb E_{\bm U} \left [ \sup \limits_{\pi \in \Pi, p \in \mathcal S} \| z_{\pi}(p) -  Z_{n, \pi}( p) \| \right] + \mathbb E_{\bm U} \left [ \sup \limits_{\pi \in \Pi} \|  Z_{n, \pi}(P_{\pi}) \| \right]   \label{eq:abd} \\ 
& = O_p(n^{-1/2}) \nonumber 
\end{align} 

The first term in \eqref{eq:abd} is $O_p(n^{-1/2})$, since by Lemma \ref{lem:unifpi} and Markov's inequality, \[ \mathbb E_{\bm U} \left [ \sup \limits_{\pi \in \Pi, p \in \mathcal S} \sqrt n \| z_{\pi}(p) - Z_{n, \pi}( p) \| \right]  = O_p(1). \] 

For the second term, by Assumption \ref{as:approx_p}, with probability at least $1 - e^{-cn}$, $\| Z_{n, \pi}(P_{\pi}) \| \leq a_n$ for every $\pi \in \Pi$, since $P_{\pi} = P_n(\bm W_{\pi})$ and is otherwise always bounded. So $\mathbb E \left [ \sup \limits_{\pi \in \Pi} \sqrt n \| Z_{n, \pi}(P_{\pi}) \| \right ] \leq \sqrt n a_n + \sqrt{n J} e^{-cn} = o(1)$  and by Markov's inequality, \[ \mathbb E_{\bm U} \left [ \sup \limits_{\pi \in \Pi} \sqrt n \|  Z_{n, \pi}(P_{\pi}) \| \right ]  = o_p(1). \] Also by Markov's inequality, following the argument in \eqref{eq:abd} pointwise for each $\pi$, with the pointwise part of Lemma \ref{lem:unifpi}, shows that $  \mathbb E_{\pi} [ \| P_{\pi} - p^*_{\pi} \|]   = O_p(n^{-1/2})$ and $   \| P_{\pi} - p^*_{\pi} \|   = O_p(n^{-1/2})$, which is enough to prove the first part of the Lemma. 
\end{proof} 

\begin{lemma}\label{lem:fmlem} 

Under the assumptions of Theorem \ref{thm:moment}, $\sqrt n | \bar V_{n}(\pi) - V^*(\pi) | = O_p(1).$ Under the assumptions of Theorem \ref{thm:regret_finite}, $ \sup \limits_{\pi \in \Pi} \sqrt n | \bar V_n(\pi) - V^*(\pi) | = O_p(1).$

\end{lemma} 

\begin{proof} 
First, we make the following expansion: 
\begin{equation*} 
\begin{split} 
 \bar V_n(\pi) - V^*(\pi) = \mathbb E_{\pi}[ Y_{n, \pi}( P_{\pi} ) - y_{\pi}(P_{\pi})] + \mathbb E_{\pi}[ y_{\pi}(P_{\pi})] - y_{\pi}(p^*_{\pi}). 
\end{split} 
\end{equation*} 

Then, we work with expected outcome functions instead: 
\begin{equation*} 
\begin{split} 
 \sup \limits_{\pi \in \Pi} | \bar V_n(\pi) - V^*(\pi) | \leq \sup \limits_{\pi \in \Pi} | \mathbb E_{\pi} [ y_{\pi}( P_{\pi} ) ]- y_{\pi}(p^*_{\pi}) | + \sup \limits_{\pi \in \Pi} \mathbb E_{\pi} \big[ \sup \limits_{p \in \mathcal S} | Y_{n, \pi}(p) - y_{\pi}(p) | \big] .  
\end{split} 
\end{equation*} 

For the first term, $\sup \limits_{\pi \in \Pi} |\mathbb E_{\pi} [ y_{\pi}(P_{\pi}) ] - y_{\pi}(p^*_{\pi}) |  \leq \sqrt J c' \sup \limits_{\pi \in \Pi} \mathbb E_{\pi}[ \| P_{\pi} - p^*_{\pi} \|]  = O_p(n^{-1/2})$, since $y_{\pi}$ is $\sqrt J c'$-Lipschitz in $p$ (Assumption \ref{as:continuum_regularity}(iii)), where the uniform bound on $ \mathbb E_{\pi} [\| P_{\pi} - p^*_{\pi}\|]$ comes from Lemma \ref{lem:concbarp}, under the assumptions of Theorem \ref{thm:regret_finite}. For the second term, with coupling as in Lemma \ref{lem:unifpi}, $\mathbb E_{\pi} = \mathbb E_{\bm U}$, so by Jensen's inequality, Lemma \ref{lem:unifpi} and Markov's inequality, $\sup \limits_{\pi \in \Pi} \mathbb E_{\pi} [ \sup \limits_{p \in \mathcal S} | Y_{n, \pi}(p) - y_{\pi}(p) | ] \leq \mathbb E_{\bm U} [ \sup \limits_{\pi \in \Pi, p \in \mathcal S} | Y_{n, \pi}(p) - y_{\pi}(p) | ] = O_p(n^{-1/2})$. Under the assumptions of Theorem \ref{thm:moment}, the same argument can be used to show the bound pointwise in $\pi$, using the pointwise results in Lemmas \ref{lem:unifpi} and \ref{lem:concbarp} rather than the uniform results.

\end{proof}

\begin{lemma}\label{lem:rmain2} 
Under the assumptions of Theorem \ref{thm:regret}, 

\[  \sup \limits_{\pi \in \Pi} \sqrt n |  \hat V_n(\pi) -  \Gamma^y_{n, \pi}(p^*_{\pi}; \eta^*_{\pi})  | = O_p(1).  \] 
\end{lemma} 

\begin{proof}

 First, we make the following expansion. 
\begin{equation*} 
\begin{split} 
  \hat V_n (\pi)  -  \Gamma^y_{n, \pi}(p^*_{\pi}; \eta^*_{\pi})  & =    \Gamma^y_{n, \pi} (\hat  P_{\pi}; \hat \eta_{\pi}) - \Gamma^y_{n, \pi} (p^*_{\pi}; \eta^*_{\pi})  \\
 &  = \Gamma^y_{n, \pi}(\hat P_{\pi} ; \hat \eta_{\pi}) - \Gamma^y_{n, \pi}(\hat P_{\pi}; \eta^*_{\pi})   +  \Gamma^y_{n, \pi}(\hat P_{\pi}; \eta^*_{\pi} ) - \Gamma^y_{n, \pi}( p^*_{\pi}; \eta^*_{\pi}) 
\end{split} 
\end{equation*} 
Then, we work with expected outcome functions instead: 
\begin{equation*} 
\begin{split} 
\sup \limits_{\pi \in \Pi} |  \hat V_n (\pi)  - \Gamma^y_{n, \pi}(p^*_{\pi}; \eta^*_{\pi})   | & \leq \sup \limits_{\pi \in \Pi}  |  y_{\pi}(\hat P_{\pi} ; \hat \eta_{\pi}) - y_{\pi}(\hat P_{\pi} ; \eta^*_{\pi})|   + \sup \limits_{\pi \in \Pi}  | y_{\pi}( \hat P_{\pi}  ; \eta^*_{\pi}) - y_{\pi}(p^*_{\pi}; \eta^*_{\pi})  | \\ & \phantom{=} +  \sup \limits_{p \in \mathcal S, \pi \in \Pi } 2 | \Gamma^y_{n, \pi}(p ; \hat \eta_{\pi}) - y_{\pi} ( p; \hat \eta_{\pi}) |  +  \sup \limits_{p \in \mathcal S, \pi \in \Pi } 2 |\Gamma^y_{n, \pi}(p ; \eta^*_{\pi}) - y_{\pi} (p ; \eta^*_{\pi}) | \\ 
& = O_p(n^{-1/2}) 
\end{split} 
\end{equation*} 

The first term is $O_p(n^{-1/2})$ by Lemma \ref{lem:unuis}.  The rate of the second term comes from a Taylor expansion and the uniform convergence rate for $\hat P_{\pi}$ in Lemma \ref{lem:concp}. The rate of the third term comes from Lemma \ref{lem:empe} and the fourth term comes from Lemma \ref{lem:empt}.

\end{proof} 

\begin{lemma} \label{lem:tailb} 
 Let $Z_1, \ldots, Z_n$ be i.i.d.\ and let $\mathcal F$ be a class of measurable functions $f: \mathcal Z \to  [-M, +M]$, where $M  \in \mathbb R$ and $M < \infty$. For some constants $V \geq 1$ and $K \geq 1$, $\sup \limits_{Q} N(\epsilon, \mathcal F, L_2(Q)) \leq \left ( \frac{K}{\epsilon} \right)^V$, for every $0 <  \epsilon < K$. Then, there is a finite constant $C$ such that, for every $t > 0$,

\[ \Pr \left ( \sup \limits_{f \in \mathcal F} \left | \frac{1}{\sqrt n} \sum \limits_{i=1}^n f(Z_i) - \mathbb E[f(Z_i) ]  \right | >t \right ) \leq  C t^V e^{- t^2/(2M^2)}  . \] 
\end{lemma} 

\begin{proof} 
This tail bound is Theorem 2.14.9 of \citet{vaart1997weak} (Theorem 2.14.28 in the second edition). Note to match the conditions of the  theorem exactly, we rescale $f$ to map to $[0, 1]$. Applying the theorem to $(f + M)/(2M)$ replaces $t$ by $t/(2M)$ in its tail bound, which gives the exponent $-t^2/(2M^2)$ and changes the constant $C$. 
\end{proof}
\begin{lemma} \label{lem:comp} \textbf{Lipschitz composition rules.}  $\mathcal F_1, \ldots \mathcal F_K$ are classes of measurable functions from $\mathcal Z \to \mathbb R$.  
Let $\psi(\mathcal F) = \{ \psi(f_1, f_2, f_3, \ldots, f_K) : f_1 \in \mathcal F_1, \ldots, f_K \in \mathcal F_K \}$ be a class that combines each of these functions, where the map $\psi: \mathbb R^{K} \to \mathbb R$ is Lipschitz in that 
\[ |\psi(f(z)) - \psi(g(z))|^2  \leq \sum \limits_{k=1}^K L^2_k | f_k(z) - g_{k}(z)|^2. \] 

for every $f, g \in \mathcal F_1 \times \ldots \times \mathcal F_K$ and every $z \in \mathcal Z$ and $L$ is positive. Then, 

\[ \sup \limits_{Q} N(\epsilon \| L \cdot F \|_{Q, 2} , \psi(\mathcal F), L_2(Q) ) \leq \prod_{k=1}^K \sup \limits_{Q_k} N(\epsilon \| F_k\|_{Q_k, 2} , \mathcal F_k, L_2(Q_k) ),\] 

where $L \cdot F =\left( \sum \limits_{k=1}^K (L^2_k F^2_k)\right)^{1/2}$ and $F_k$ denotes an envelope function for $\mathcal F_k$. 

\end{lemma}

\begin{proof} 
This is Lemma A.6 of \citet{chernozhukov2014gaussian}. 
\end{proof}

\end{document}